\documentclass[12pt]{article}
\usepackage{graphicx}
\usepackage{color}
\usepackage{amsmath}
\usepackage{amssymb}
\usepackage{amscd}
\usepackage{amsthm}
\usepackage{amsopn}
\usepackage{xspace}
\usepackage{verbatim}
\usepackage{amsmath}
\usepackage{amsfonts}
\usepackage{epic}
\usepackage{eepic}
\usepackage{stmaryrd}
\usepackage{empheq}
\usepackage[active]{srcltx}
\usepackage[a4paper,margin=3cm]{geometry}
\definecolor{red}{rgb}{1,0,0}
\definecolor{green}{rgb}{0,1,0}
\definecolor{SeaGreen}{RGB}{46,139,87}
\definecolor{Maroon}{RGB}{128,0,0}

\newcommand{\Supp}{\textrm{Supp~}}
\newcommand{\N}{\mathbb{N}}
\newcommand{\Z}{\mathbb{Z}}

\newcommand{\C}{{\mathbb{C}}}
\newcommand{\R}{{\mathbb{R}}}

\newcommand{\A}{{\mathcal A}}
\newcommand{\B}{\mathcal B}
\newcommand{\CC}{\mathcal C}

\def\Dg {{\mathcal D}}
\newcommand{\F}{\mathcal F}

\def\Jg {{\mathcal J}}
\newcommand{\K}{\mathcal K}
\def\Kg {{\mathcal K}}
\def\Eg {{\mathcal E}}
\newcommand{\LL}{\mathcal L}

\def\PP{\mathcal P}

\def\Rg {{\mathcal R}}
\def\Sg {{\mathcal S}}
\def\Ug {{\mathcal U}}

\def\Tg {{\mathcal T}}
\def\jf {{\mathfrak j}}

\newcommand{\sca}[2]{\langle#1,#2\rangle}

\def\sign{\text{\rm sign\,}}

\newcommand{\supp}{\rm Supp\,}

\renewcommand {\Re}{{\rm Re\,}}
\renewcommand{\Im}{{\rm Im\,}}
\newcommand {\pa}{\partial}

\newcommand {\ar}{\to}

\def\Ai{\text{\rm Ai\,}}

\def\0{\mathbf  0}

\def\XXint#1#2#3{{\setbox0=\hbox{$#1{#2#3}{\int}$ }
\vcenter{\hbox{$#2#3$ }}\kern-.6\wd0}}

\newcommand{\Union}{\mathop{\bigcup}\limits}

 

\errorcontextlines=0 \numberwithin{equation}{section}
\theoremstyle{plain}
\newtheorem{theorem}{Theorem}[section]
\newtheorem{lemma}[theorem]{Lemma}

\newtheorem{proposition}[theorem]{Proposition}
\newtheorem{assumption}[theorem]{Assumption}
\newtheorem{definition}[theorem]{Definition}
\newtheorem{remark}[theorem]{Remark}
\newtheorem{corollary}[theorem]{Corollary}
\newtheorem{example}[theorem]{Example}
\newtheorem{conjecture}[theorem]{Conjecture}
\title{On
  a Schr\"odinger operator with a purely imaginary potential in the 
  semiclassical limit} 

\author{ Y. Almog, Department of
  Mathematics, Louisiana State University,\\ 
    Baton Rouge, LA 70803, USA,\\~\\
 D. S. Grebenkov, 
 Laboratoire de Physique de la Mati\`ere Condens\'ee, \\ 
CNRS -- Ecole Polytechnique, University Paris-Saclay, \\ F-91128 Palaiseau, France,
\\
  and \\
  B. Helffer, Laboratoire de Math\'ematiques Jean Leray, \\CNRS and Universit\'e de Nantes, \\
  2 rue de la Houssini\`ere, 44322 Nantes Cedex France.}

\date{}

\begin{document}
\bibliographystyle{siam}

\maketitle
\begin{abstract}
We consider the operator $\A_h=-h^2\Delta+iV$ in the semi-classical
limit $h\to0$, where $V$ is a smooth real potential with no critical
points. We obtain both the left margin of the spectrum, as well as
resolvent estimates on the left side of this margin. We extend here
previous results obtained for the Dirichlet realization of $\A_h$ by
removing significant limitations that were formerly imposed on $V$.
In addition, we apply our techniques to the more general Robin
boundary condition and to a transmission problem which is of
significant interest in physical applications.  
\end{abstract}

\section{Introduction}
\label{sec:1}
Consider the Schr\"odinger operator with a purely imaginary potential 
\begin{subequations}  \label{eq:1}
\begin{equation}
\A_h = -h^2\Delta + i\, V \,,
\end{equation}
in which $V$ is a $C^3$-potential in $\overline{\Omega}$, for an open
bounded set $\Omega\subset \mathbb R^n$ with smooth boundary
(Fig. \ref{fig:domain}). The Dirichlet and Neumann realizations of
$\A_h$, which we respectively denote by $\A_h^D$ and $\A_h^N$, have
already been considered in
\cite{Alm,Hen2,AlHe}.  Their respective domains are given by
\begin{align*}
  &   D(\A_h^D)=H_0^1(\Omega,\C)\cap H^2(\Omega,\C)\,, \\
  &   D(\A_h^N)= \{u\in H^2(\Omega,\C)\,\Big|\, \partial_\nu u=0 \mbox{ on } \partial
  \Omega\} 
\end{align*}
where $\nu$ is a unit normal vector pointing outwards on
$\partial\Omega$. In the present contribution we consider two
different realizations of $\A_h$.  The first of them is the Robin
boundary condition for which the realization of $\A_h$ is denoted by
$\A_h^R$. The domain of $\A_h^R$ is given by
\begin{equation}
\label{eq:71}
  D(\A_h^R)= \{ u \in H^2 (\Omega) \,,\, h^2 \partial_\nu u =  - \mathcal K\, u \mbox{ on } \partial \Omega\}
\end{equation}
\end{subequations}
where $\mathcal K$ denotes the Robin coefficient.  We note that
$\A_h^R(\Kg)$ is a generalization both $\A_h^N$, corresponding to
$\Kg=0$, and $\A_h^D$ which is obtained in the limit $\Kg\to\infty$. The
form domain of $\A_h^R$ is $H^1(\Omega)$ and the associated quadratic form
reads
\begin{equation}
\label{eq:formR}
u \mapsto  q_V^R(u) := h^2 \,\| \nabla u\|^2_{ \Omega} 
  + i \int_{ \Omega} V(x)  | u(x) |^2 \, dx + \mathcal K  \int_{\partial \Omega} |u|^2 ds \,.
\end{equation}
We shall consider the semiclassical limit $h\to0$ when
\begin{equation}
\label{eq:94}
\mathcal K = h^{\frac 43}\, \kappa\,,
\end{equation}
for a fixed value of $\kappa$.  The motivation for considering
this scaling is provided in \cite{GH}.

A second problem we address in this work is the so called transmission
boundary condition that is described for a one-dimensional setting in
\cite{GHH} and for a more general setup in \cite{GH}.  A typical case
is that of non empty open connected sets $\Omega_- $, $\Omega_+$ and
$\Omega$ such that
\begin{equation}\label{caseT}
 \Omega_- \subset  \Omega \,,\, \Omega \setminus \Omega_- = \overline{\Omega_+}\,,
\end{equation}
where $\Omega_-$ must be simply connected (Fig. \ref{fig:domain}). In
this case, the transmission boundary condition is prescribed on
$\partial \Omega_-$ and a Neumann condition is prescribed on $\partial
\Omega$. In this case we introduce
\begin{equation*}
\Omega^T:= \Omega_- \cup \Omega_+
\end{equation*}
and observe that $\partial \Omega^T= \partial \Omega_- \cup \partial
\Omega = \partial \Omega_+$.  The quadratic form associated with this specific
realization of $\A_h$, defined on $H^1(\Omega_-) \times
H^1(\Omega_+)$, reads
\begin{multline}
\label{eq:formT}
u =(u_-,u_+) \mapsto  q_V^T(u) := h^2 \,\| \nabla u_-\|^2_{ \Omega_-} + h^2 \,\| \nabla u_+\|^2_{ \Omega_+} \\
 + i \int_{ \Omega_-} V(x)  | u_-(x) |^2 \, dx + i \int_{ \Omega_+} V(x)  | u_+(x) |^2 \, dx \\
 + \mathcal K  \int_{\partial \Omega_-} |u_+-u_-|^2  ds \,.
\end{multline}
The domain of the associated operator $\A_h^{TN}$ is  
\begin{multline}
\label{eq:7}
  D(\A_h^{TN})=  \{ u \in H^2 (\Omega_-)\times H^2(\Omega_+) \,,\,\\
  h^2 \partial_\nu u_-=  h^2
  \partial_\nu u_+=   \mathcal K (u_+-u_-)   \mbox{ on } \partial \Omega_- \,,\,  \partial_\nu
  u_+=0 \text{ on } \partial\Omega \}\,,
\end{multline}
where $\nu$ is pointing outwards of $\Omega_-$ at the points of
$\partial \Omega_-$ and outwards of $\Omega$ at the points of
$\partial \Omega$ .  For $\mathcal K=0\,$, the transmission problem is
reduced to two independent Neumann problems in $\Omega_-$ and
$\Omega_+$.

\begin{figure}
\begin{center}
\includegraphics[width=140mm]{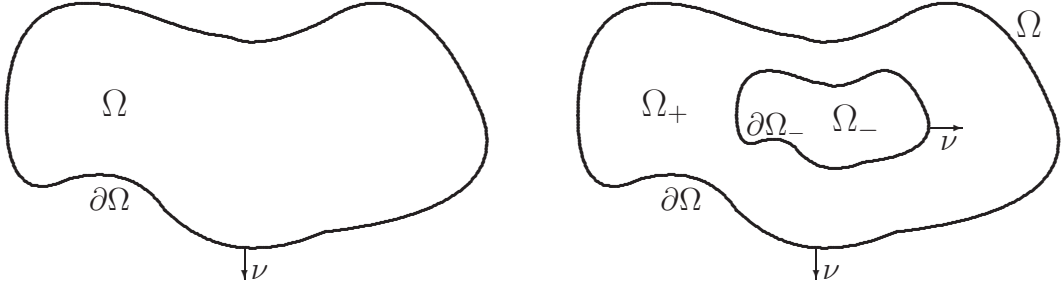}
\end{center}
\caption{
Geometric illustration of the problem: (left) a bounded open set
$\Omega\subset \R^n$ with a smooth boundary $\partial\Omega $, on
which Dirichlet, Neumann or Robin boundary condition is imposed;
(right) non empty open connected sets $\Omega_- $, $\Omega_+$, $\Omega
\subset \R^n$ satisfying (\ref{caseT}), with the transmission boundary
condition imposed on $\partial \Omega_-$, and Neumann or Dirichlet
condition imposed on $\partial \Omega$.  The unit normal vector $\nu$
pointing outwards is also shown. }
\label{fig:domain}
\end{figure}

In addition, we shall address, as in \cite{GH}, a Dirichlet condition
on $\partial \Omega$ instead of a Neumann condition. To distinguish
between these two situations, we write $ \A_h^{TN}$ and $\A_h^{TD}$.
Sometimes, we use the notation $\A_{h,\mathcal K} ^{TD}$ to keep in
mind the reference to the transmission parameter $\mathcal K$, which
could be $h$-dependent.  In the Dirichlet case, the form domain
$H^1(\Omega_-) \times H^1(\Omega_+)$ should be replaced by
$H^1(\Omega_-) \times \tilde H^1(\Omega_+)$, where $\tilde
H^1(\Omega_+) =\{ u\in H^{1} (\Omega_+)\,,\, u=0\mbox{ on } \partial
\Omega\}$. The domain of the operator is then 
\begin{multline}
\label{eq:7D}
  D(\A_h^{TD})=
\{ u=(u_-,u_+) \in H^2 (\Omega_-)\times H^2(\Omega_+) \\ \hspace*{4em} h^2 \partial_\nu u_-=  h^2 \partial_\nu
  u_+=   \mathcal K (u_+-u_-)  \mbox{ on } \partial \Omega_-\,,\, u_+ =0 \mbox{
    on } \partial \Omega \}\,, 
  \end{multline}
where $\nu$ is pointing outwards of $\Omega_-$ at the points of
$\partial \Omega_- \,$.
 
The spectral analysis of the various realizations of $\A_h$ has
several applications in mathematical physics, among them are the
Orr-Sommerfeld equations in fluid dynamics \cite{sh03}, the
Ginzburg-Landau equation in the presence of electric current (when
magnetic field effects are neglected)
\cite{Alm,aletal10,aletal13,aletal12}, the null controllability of
Kolmogorov type equations \cite{BHHR}, and the diffusion nuclear
magnetic resonance \cite{Stoller91,deSwiet94,Grebenkov07}. In
particular, the transmission problem naturally arises in diffusion or
heat exchange between two sets separated by a partially
permeable/isolating interface (see \cite{Grebenkov10,GH} and
references therein).  In this setting, the first relation in
(\ref{eq:7}) or (\ref{eq:7D}) ensures the continuity of flux between
two sets, whereas the second relation accounts for the drop of the
transverse magnetization $u$ across the interface, with the
transmission coefficient $\mathcal K$.  As in \cite{Alm,Hen2,AlHe} for
the Dirichlet or Neumann case, we seek an approximation for $\inf
\Re\sigma(\A_h^R)$, $\inf \Re\sigma(\A_h^{TN})$ or $\inf
\Re\sigma(\A_h^{TD})$ as $h\to 0$ ($h>0$).  In \cite{Hen2,AlHe,GH},
various constructions of quasimodes give some idea for the location of
the leftmost eigenvalue (i.e. with smallest real part).

In the following, we formulate the assumptions, the notation and the
main statements of the paper.

\begin{assumption} 
\label{notdeg}
The potential $V$ satisfies
\begin{equation*}
\nabla V(x)\neq0\,,\, \forall x \in \overline{\Omega}\,.
\end{equation*}
\end{assumption}
Since, aside from minor differences, the treatment of all boundary
conditions is similar, we use a general notation that can describe all
problems. Thus, we define $\Omega^\#$ by $\Omega^\#= \Omega$ if
$\#=D,N, R$ and by $\Omega^\#=\Omega_-\cup \Omega_+$ in the
transmission case: $\#= TD$ or $\#=TN$.  Let $\partial\Omega_\perp^\#$
denote the subset of $\partial\Omega^\#$ where $\nabla V$ is
orthogonal to $\partial\Omega^\#$:
\begin{equation}\label{defPaPerp}
 \partial\Omega_\perp^\# = \{x\in\partial\Omega^\#: \nabla V(x) = (\nabla V(x) \cdot \vec \nu(x))\,\vec \nu(x) \}\,,
\end{equation}
where $\vec \nu(x)$ denotes the outward normal on $\partial\Omega^\#$ at
$x$\,.

Let $\#\in\{D,N,R,T\} $ and $ \mathfrak D^\#$ be defined in the
following manner
\begin{equation}
\label{eq:98}
  \begin{cases}
    \mathfrak D^\#=\{u\in H^2_{loc}(\overline{\R_+}) \,| \, u(0)=0 \} & \#=D \\
    \mathfrak D^\# = \{u\in H^2_{loc}(\overline{\R_+}) \,| \, u^\prime(0)=0 \} & \#=N \\
     \mathfrak D^\# = \{u\in H^2_{loc}(\overline{\R_+}) \,| \, u^\prime(0)= \kappa \, u (0) \} & \#=R \\
    \mathfrak D^\# = \{ u\in H^2_{loc}(\overline{\R_-}) \times H^2_{loc}
    (\overline{\R_+}) \,| \,  u^\prime_+(0)= & \\ \kern 10em u^\prime_-(0)=\kappa\, [u_+(0)-u_-(0)] \}
    & \#=T \,.
  \end{cases}
\end{equation}
In the two last cases we occasionally write $\mathfrak D^\#(\kappa)$
in order to emphasize the dependence on the Robin or transmission
parameter $\kappa$.  In the above $u_\pm=u|_{\R_\pm}$ if we identify
$L^2(\mathbb R_-\cup\mathbb R_+)$ and $ L^2(\R_-)\times L^2(\R_+)$. We
further set $\R_\#=\R_+$ when $\#\in\{D,N,R\}$ and
$\R_\#=\R_+\cup\R_-$ for $\#=T$. Then, we define the operator
\begin{displaymath}
   \LL^\#(\jf ) = -\frac{d^2}{dx^2}+i\, \jf \,x \,,
\end{displaymath}
whose domain is given  by 
\begin{equation}
\label{eq:101}
  D(\LL^\# (\jf ))= H^2(\R^\#)\cap L^2(\R^\#;|x|^2dx)\cap \mathfrak D^\#\,,
\end{equation}
(where $H^2(\mathbb R_-\cup \mathbb R_+)=H^2(\mathbb R_-) \times H^2(\mathbb
R_+)$) and set
\begin{equation} 
\label{eq:75}
  \lambda^\#(\jf ) = \inf \Re\sigma(\LL^\#(\jf )) \,.
\end{equation}
Again, when $\kappa$ is involved, we occasionally write $\LL^\# (\jf
,\kappa)$, $\lambda^\#(\jf ,\kappa)$.

Next, let 
\begin{equation}
\label{eq:97} 
\Lambda_{m}^\#=\inf_{x\in\partial\Omega_\perp^\#}\lambda^\#(|\nabla V(x)|) \,,
\end{equation}

In the transmission case, for $\#= T\flat$ with $\flat \in \{D,N\}$
the above formula should be interpreted as
\begin{equation}\label{newdef}
\Lambda_{m}^{T\flat}(\kappa)  = \min \left( \inf_{x\in\partial\Omega_\perp^\# \cap \partial \Omega_-} 
\lambda^T(|\nabla V(x)|,\kappa) \,, \inf_{x\in\partial\Omega_\perp^\# \cap \partial \Omega}  \lambda^\flat(|\nabla V(x)|) \right)\,.
\end{equation}

In all cases we denote by $\Sg^\#$ the set 
\begin{equation}\label{eq:2}
  \Sg^\#=\{x\in\partial\Omega_\perp^\# \,: \,\lambda^\# (|\nabla V(x)|,\kappa) = \Lambda_{m}^\#(\kappa) \,\}\,.
\end{equation}

When $\#\in\{D,N\}$ it can be verified by a dilation argument that,
when $\jf >0\,$,
\begin{equation}
   \lambda^\#(\jf) =\lambda^\#(1)\, \jf ^{2/3}\,,
\end{equation}
and when $\#\in \{R,T\}$ with parameter $\kappa\geq0$ that (see
Sections \ref{s2} and \ref{s3}),
\begin{equation}
   \lambda^\#(\jf ,\kappa) =\lambda^\#(1, \kappa \,\jf ^{-\frac 13})\, \jf ^{2/3}\,.
\end{equation}
For the Robin case ($\#=R$) we establish in Appendix \ref{appC} that $
\lambda^R(\jf,\kappa)$ is monotonously increasing with $\jf $. Hence, for $\#\in\{D,N,R\}$ we
have for 
\begin{equation}\label{defJm}
  \jf_m = \min_{x\in\partial\Omega_\perp}|\nabla V(x)|\,,
\end{equation}
the property 
\begin{equation*}
\Lambda_{m}^\# = \lambda^\# (\jf_m)\,, \mbox{ for } \#\in\{D,N\}\,,\, \Lambda_{m}^R(\kappa)=  \lambda^R (\jf _m,\kappa)\,.
\end{equation*}
For the transmission case $\#= T\flat$ (with $\flat =D$ or $\flat
=N$), as the monotonicity of $\jf \mapsto \lambda^{T} (\jf ,\kappa)$
has not been established, it is more difficult to define $\jf_m $ and
we shall refrain from using it.

We next make the following additional assumption:
\begin{assumption}\label{nondeg2}
At each point $x$ of $\Sg^\#$, 
\begin{equation}
\label{eq:105}
 \alpha (x)=\text{\rm det} D^2V_\partial(x) \neq 0\,,
\end{equation}
where $V_\partial$ denotes the restriction of $V$ to
$\partial\Omega^\#$, and $D^2V_\partial$ denotes its Hessian matrix.
\end{assumption}
It can be easily verified that \eqref{eq:105} implies that $\mathcal
S^\#$ is finite.  Equivalently we may write
\begin{subequations}
   \label{eq:95}
\begin{equation}
\alpha(x) =\Pi_{i=1}^{n-1}\alpha_i(x) \neq0\,,
\end{equation}
where
\begin{equation}
  \{\alpha_i\}_{i=1}^{N-1} = \sigma(D^2V_\partial) \,,
\end{equation}
where each eigenvalue is counted according to its multiplicity. 
\end{subequations}

The following has been established by R. Henry in \cite{Hen2}
\begin{theorem}
\label{thmNSAschro1} 
Under Assumptions \ref{notdeg} and \ref{nondeg2}, we have
\begin{equation}\label{limSpect1}  
 \varliminf\limits_{h\to0}\frac{1}{h^{2/3}}\inf \bigl\{\Re\, \sigma(\mathcal A _h^D) \bigr\} \geq  \Lambda_{m}^{D}\,,
\qquad  \Lambda_{m}^{D} = \frac{|a_1|}{2}\jf_m ^{2/3}\,,
\end{equation}  
where $a_1<0$ is the rightmost zero of the Airy function $\Ai$\,.
Moreover, for every $\varepsilon>0\,$, there exist $h_\varepsilon>0$
and $C_\varepsilon>0$ such that
\begin{equation}\label{estRes1}  
 \forall h\in(0,h_\varepsilon),~~~
 \sup_{
\begin{subarray}{c} 
\gamma\leq \Lambda_{m}^{D}  \\[0.5ex]
     \nu \in\mathbb{R}
\end{subarray}
}
 \|(\mathcal A _h^D-(\gamma -\varepsilon)h^{2/3}-i\nu)^{-1}\|\leq\frac{C_\varepsilon}{h^{2/3}}\,.
\end{equation}
\end{theorem}
In its first part, this result is essentially a reformulation of the
result stated by the first author in \cite{Alm}.  Note that the second part
provides, with the aid of the Gearhart-Pr\"uss theorem, an effective
bound (with respect to both $t$ and $h$) of the decay of the
associated semi-group as $t\ar +\infty\,$.  The theorem holds in
particular in the case $V(x)=x_1$ where $\Omega$ is a disk (and hence
$S^T$ consists of two points) and in the case of an annulus (four
points).  Note that $\jf_m =1$ in this case.

A similar result can be proved for the Neumann case where
\eqref{limSpect1} is replaced by
\begin{equation}\label{limSpect1N} 
\varliminf\limits_{h\to0}\frac{1}{h^{2/3}}\inf \bigl\{\Re\, \sigma(\mathcal A _h^N) \bigr\} \geq  \Lambda_{m}^{N}\,,
\qquad  \Lambda_{m}^{N} = \frac{|a'_1|}{2}\, \jf_m ^{2/3}\,,
\end{equation}
where $a'_1<0$ is the rightmost zero of $\Ai^\prime$, and
\eqref{estRes1} is replaced by
\begin{equation}\label{estRes1N}  
 \forall h\in(0,h_\varepsilon),~~~
 \sup_{
\begin{subarray}{c} 
\gamma\leq \Lambda_{m}^{N}  \\[0.5ex]
     \nu \in\mathbb{R}
\end{subarray}
}
 \|(\mathcal A _h^N-(\gamma -\varepsilon)h^{2/3}-i\nu )^{-1}\|\leq\frac{C_\varepsilon}{h^{2/3}}\,.
\end{equation}

We establish here the corresponding results, for both the Robin
boundary condition and the various Transmission problems.
\begin{theorem}\label{thmNSAschro1R} 
Under Assumptions \ref{notdeg} and \ref{nondeg2}, 
\begin{equation}\label{limSpect1R} 
\varliminf\limits_{h\to0}\frac{1}{h^{2/3}}\inf \bigl\{\Re\, \sigma(\mathcal A _{h,\mathcal K(h)} ^R) \bigr\} \geq  \Lambda_{m}^{R}(\kappa), 
\qquad  \Lambda_{m}^{R}(\kappa) = \lambda^R (\jf_m,\kappa)\,, 
\end{equation}
where $\mathcal A _{h,\mathcal K}^R$ is the  Robin realization (with
parameter $\mathcal K\geq 0$) of $\mathcal A_h$\,, and $\mathcal
K=\mathcal K (h) $ satisfies \eqref{eq:94}.\\ 
Moreover, for every $\varepsilon>0\,$, there exist $h_\varepsilon >0$
and $C_\varepsilon>0$ such that
\begin{equation}\label{estRes1R} 
 \forall h\in(0,h_\varepsilon),~~~
 \sup_{
\begin{subarray}{c}
\gamma\leq \Lambda_{m}^{R}(\kappa), \\[0.5ex]
\nu \in\mathbb{R}
\end{subarray}
}
 \|(\mathcal A _{h,\mathcal K(h)}^R-(\gamma -\varepsilon)h^{2/3}-i\nu)^{-1}\|\leq\frac{C_\varepsilon}{h^{2/3}}\,.
\end{equation}
\end{theorem}

\begin{theorem}\label{thmNSAschro1T}
Under Assumptions \ref{notdeg} and \ref{nondeg2},
\begin{equation}\label{limSpect1T}  
\varliminf\limits_{h\to0}\frac{1}{h^{2/3}}\inf \bigl\{\Re\, \sigma(\mathcal
A _{h,\mathcal K(h)}^{T\flat }) \bigr\} \geq \Lambda_{m}^{T\flat}(\kappa)\,, 
\end{equation}
where $\mathcal A _{h,\mathcal K}^{T\flat}$ is the
$T\flat$-realization (with parameter $\mathcal K\geq 0$) of $\mathcal
A_h$ with $\flat $-condition on $\partial \Omega$ (\,$\flat \in
\{D,N\}$) and $T$-condition along $\partial \Omega_-$ and $\mathcal
K=\mathcal K(h)$ satisfies \eqref{eq:94}.\\
Moreover, for every $\varepsilon>0\,$, there exist $h_\varepsilon >0$
and $C_\varepsilon>0$ such that
\begin{equation}\label{estRes1T}
 \forall h\in(0,h_\varepsilon),~~~
 \sup_{
\begin{subarray}{c}
\gamma\leq \Lambda_{m} ^{T\flat}(\kappa) ,\\[0.5ex]
\nu \in\mathbb{R}
\end{subarray}
}
 \|(\mathcal A _{h,\mathcal K (h)}^{T\flat}-(\gamma -\varepsilon)h^{2/3}-i\nu)^{-1}\|\leq\frac{C_\varepsilon}{h^{2/3}}\,.
\end{equation}
\end{theorem}

We now look at upper bounds for the left margin of the spectrum. Our
main theorem is:
\begin{theorem}
\label{upperbound}
Under Assumptions \ref{notdeg}  and \ref{nondeg2},
for $\#\in \{D,N,R,T\flat\}$ with  $\flat \in \{D,N\}$, one has
\begin{equation}\label{limSpect1Nat} 
  \varlimsup\limits_{h\to0}\frac{1}{h^{2/3}}\inf \bigl\{ \Re\, \sigma(\mathcal
  A _{h,\mathcal K (h)}^{\#}) \bigr\} =   \Lambda_{m}^{\#} (\kappa)  \,,  
\end{equation}
where $\mathcal K(h)$ satisfies \eqref{eq:94}.
\end{theorem}
\begin{remark}
  An immediate conclusion which follows from the previous statements is
  that under Assumptions \ref{notdeg} and \ref{nondeg2}, for $\#\in
  \{D,N,R,T\flat\}$ with $\flat \in \{D,N\}$, we have
\begin{equation}\label{limitpect1Nat} 
 \lim_{h\to0}\frac{1}{h^{2/3}}\inf \bigl\{ \Re\, \sigma(\mathcal
  A _{h,\mathcal K (h)}^{\#}) \bigr\} =   \Lambda_{m}^{\#} (\kappa)  \,,  
\end{equation}
where $\mathcal K(h)$ satisfies \eqref{eq:94}\,.
\end{remark}

In the case of the Dirichlet problem, this theorem was obtained in
\cite[Theorem 1.1]{AlHe} under the stronger assumption that, at each
point $x$ of $ \Sg^D$, the Hessian of $ V_\partial:= V_{/\partial \Omega^\#}$ is
positive definite if $\partial_\nu V (x) < 0\,$ or negative definite if
$\partial_\nu V (x) > 0\,$, with $\partial_\nu V:=\nu\cdot \nabla V$.  This additional
assumption reflects some technical difficulties in the proof, that we
overcome in Section \ref{s7} by using tensor products of semigroups, a
point of view that is missing \cite{AlHe}. This generalization allows
us to obtain the asymptotics of the left margin of $\sigma(\A_h^\#)$, for
instance, when $V(x_1,x_2)=x_1$ and $\Omega$ is either an annulus or the
exterior of a disk, where the above assumption is not satisfied.  For
this particular potential, an extension to the case when $\Omega$ is
unbounded is of significant interest in the physics literature
\cite{Grebenkov17}. We may assume in this case that $\partial \Omega^\#$ is
bounded and add for the potential $V$ the assumption (having in mind
the case $V=x_1$) that there exist a compact set $K$ and positive
constants $c$, $C$ such that, $\forall x \notin K\,$, $c \leq |\nabla V (x)| $ and
$\sum_{1\leq |\alpha|\leq 2} | \partial_x^\alpha V (x)| \leq C$\,. We leave this problem to
future research.

The rest of this paper is organized as follows:\\ In the next section
we briefly review properties of the Robin realization of the complex
Airy operator in $\R_+$ that were established in \cite{GHH}, and
extend them slightly further to accommodate our needs in the sequel.
We do the same in Section \ref{s3} for the transmission problem. In
Section \ref{sModels} we consider the operator $-\Delta+iJ\cdot x$
(for some $J\in\R^n$) in $\R^n$ and in the half-space, where the
boundary set on the hyperplane $x_n=0$. Most of the results in this
section have been obtained in \cite{Hen2,Alm}, but some refined
semigroup and resolvent estimates that are necessary in the last
section are provided as well. In Section \ref{s5} we characterize the
domain of operators with quadratic potential both in $\R^n$ (in fact,
we address there a much more general class of operators) and in the
presence of a boundary or an interface (the half-space). In
Section \ref{s6} we prove Theorems \ref{thmNSAschro1R} and
\ref{thmNSAschro1T}. In the last section \ref{s7} we prove Theorem
\ref{upperbound}.  Finally, in Appendix \ref{appB} we prove a simple
inequality to assist the reader, and in Appendix \ref{appC}, we
provide more information on the monotonicity of the real part of the
eigenvalue of the one-dimensional complex Airy operator with respect
to a parameter, which is less crucial for the sake of proving lower
and upper bounds for $\inf \Re\sigma(\A_h)$, than what is covered in
Sections \ref{s2} and \ref{s3} but allows for a simpler formulation of
some of the results.

\section{The complex Airy operator on the half-line: Robin case}\label{s2}  

For $\jf \neq 0$ and $\kappa \geq 0$, we consider
\begin{equation}
\label{eq:99}
  \LL^R(\mathfrak{j},\kappa)= -\frac{d^2}{dx^2} + i\, \mathfrak{j} \, x 
\end{equation}
defined on (cf. \cite{GHH})
\begin{equation}
\label{eq:107}
  D( \LL^R(\mathfrak j,\kappa))= \{ u\in H^2(\R_+,\C) \cap L^2(\R_+,\C\,;\,x^2dx)\,| \,
  u^\prime(0)= \kappa \, u(0)\,\} \,.
\end{equation}
The operator is associated with the sesquilinear form defined on
$H^1(\R_+,\C)\times H^1(\R_+,\C)$ by
\begin{displaymath}
a^{R} (u,v) = \int_0^{+\infty} u'(x)\, \bar v '(x) \,dx + i \, \jf \int_0^{+\infty} x
\, u(x)\,  \bar v (x)\, dx + \kappa \, u(0)\, \bar v (0)\,.
\end{displaymath}
We begin by recalling some of the results of \cite{GHH} with the Robin
boundary condition which naturally extends from both Dirichlet and
Neumann cases. One should be more careful with the dilation
argument. \\
{\bf Dilation argument.\\} 
For $\jf >0$, if $U_\mathfrak{j} \in \mathcal L(L^2(\R_+,\C))$ denotes
the following unitary dilation operator
\begin{displaymath}
   ( U_\mathfrak{j} u)(x) = \mathfrak{j}^{1/3} u(\mathfrak{j}^{1/3}\,x) \,,
\end{displaymath}
we observe that 
\begin{equation}
\label{eq:69a}
  \LL^R(1,\mathfrak{j}^{-1/3}\kappa)= \mathfrak{j}^{2/3}U_\mathfrak{j}^{-1}\LL^R (\mathfrak{j},\kappa)U_\mathfrak{j}\,,
\end{equation}
It is then enough by dilation to consider the case $\jf =\pm 1$, but
with a new Robin parameter and by using the complex conjugation $\jf
=1$. \\
 
In the Robin case, the distribution kernel (or the Green's function)
of the resolvent is given by
\begin{equation*}
\mathcal G^{R} (x,y\,;\lambda) =  \mathcal G_0(x,y\,;\lambda) + \mathcal G
_1^{R}(x,y\,;\kappa, \lambda)\,  \quad \textrm{for}~ (x,y)\in \mathbb R_+^2\,,
\end{equation*}
where 
\begin{equation}
\label{eq:63}
\begin{split}
  \mathcal G^{R}_1(x,y\,;\kappa, \lambda) & = - 2\pi \frac{i e^{-i2\pi/3}
    \Ai'(e^{-i2\pi/3} \lambda) + \kappa\, \Ai(e^{-i2\pi/3}\lambda)}
  {i e^{i2\pi/3} \Ai'(e^{i2\pi/3} \lambda) + \kappa\, \Ai(e^{i2\pi/3} \lambda)}  \\
  & \times  \Ai\bigl(e^{i2\pi/3} (-ix+\lambda)\bigr)~ \Ai\bigl(e^{i2\pi/3} (-iy+\lambda)\bigr)\,,  \\
\end{split}
\end{equation}
and $\mathcal G _0 (x,y\,;\lambda) $ is the resolvent of $D_x^2+i x$
on $\R$, which is an entire function of $\lambda$.

Setting $\kappa = 0\,$, one retrieves the Neumann case, while the
limit $\kappa\to +\infty$ yields the Dirichlet case.  As in the
Dirichlet case \cite{Alm,Hen2}, the resolvent is compact and in the
Schatten class $\mathcal C^p$ for any $p>\frac 32$.  Its
(complex-valued) poles are determined by solving the equation
\begin{equation}
\label{eigenvalueR}
f^R(\kappa, \lambda):=i e^{-i2\pi/3} \Ai '(e^{-i2\pi/3} \lambda) - \kappa \Ai (e^{-i2\pi/3} \lambda) = 0\,.
\end{equation}
Denote by $\lambda_j^R(\kappa)$ ($j\in \mathbb N^*$) the sequence of
eigenvalues that we order by their non decreasing real part. Except
for the case of small (respectively. large) $\kappa$, in which the
eigenvalues can be shown to be close to the eigenvalues of the Neumann
(respectively Dirichlet) problem, it does not seem easy to obtain the
precise value of $\lambda_j^R$ for any $j\in\N$.  Nevertheless, one
can prove that the zeros of $f^R(\kappa, \cdot)$ are simple. If indeed
$\lambda$ is a common zero of $f^R$ and $(f^R)'$, then either $\lambda
+\kappa^2=0$, or $e^{-i 2\pi/3} \lambda$ is a common zero of $\Ai$ and
$\Ai'$.  The second option is excluded by uniqueness of the trivial
solution for the initial value problem $-u^{\prime\prime}+zu=0$,
$u(z_0)=u^\prime(z_0)=0$, whereas the first option is excluded for
$\kappa \geq 0$ because the spectrum is contained in the positive
half-plane.

Since the numerical range of the Robin realization of $D_x^2+ix$
is contained in the first quadrant of the complex plane we have
\begin{proposition}
\begin{subequations} \label{eq:64}
\begin{equation}
  \| \mathcal G^R(\kappa, \lambda)\| \leq \frac{1}{|\Re \lambda|}\,,\, \mbox{ if } \Re \lambda < 0\,,
\end{equation}
and
\begin{equation}
  \| \mathcal G^R(\kappa,\lambda)\| \leq \frac{1}{|\Im \lambda|}\,,\, \mbox{ if }  \Im \lambda < 0\,.
\end{equation}
\end{subequations}
\end{proposition}
The above, together with the Phragm\'en-Lindel\"of principle (see
\cite{Agm}) and the fact that the resolvent is in $\mathcal C^p$, for
any $p>\frac 32$, implies (after a dilation to treat general $\jf$)
the proposition:
\begin{proposition}\label{CompRa}
For any $\kappa \geq 0$ and $\jf \neq 0$, the space generated by the
eigenfunctions of $\LL^R(\mathfrak{j},\kappa)$ is dense in
$L^2(\R_+,\C)$.
\end{proposition}

We conclude this section with some semigroup estimates.\\

\begin{proposition}\label{lemmaunifbis}
Let $\lambda^R(\mathfrak{j},\kappa)$ denote the real value of the
leftmost eigenvalue of $\LL^R(\mathfrak{j},\kappa)$. Then for any
positive $\mathfrak{j}_0\,$, $\mathfrak{j}_1$, $\kappa_0$ and
$\epsilon$ there exists
$C(\mathfrak{j}_0,\mathfrak{j}_1,\kappa_0,\epsilon)>0$ such that, for
$0 < \mathfrak{j}_0 \leq \mathfrak{j} \leq \mathfrak{j}_1$ and $\kappa
\in [0,\kappa_0]$,
\begin{equation} 
\label{eq:104}
\|e^{ -t\LL^R (\mathfrak{j},\kappa) }\|\leq {C (
    \mathfrak{j}_0,\mathfrak{j}_1,\kappa_0, \epsilon) }\, e^{-t(\lambda^R(\mathfrak{j},\kappa)-
    \epsilon)}\,. 
\end{equation}
\end{proposition}
\begin{proof}
As already observed, it suffices to consider the dependence of $\|e^{
-t\LL^R (1,\kappa^*)} \|$ on
\begin{equation}\label{defkappa*}
 \kappa^*=\mathfrak{j}^{-1/3}\kappa\,.
\end{equation}

Recall that $\lambda_1^R(\kappa^*)=\lambda_1^R(1,\kappa^*)$ denotes
the leftmost eigenvalue of $\LL^R(1,\kappa^*)$. Since $\lambda^R_1
(\kappa^*)$ is a simple zero of solution of $f^R(\kappa^*, \lambda)$
it must be a $C^1$ function of $\kappa^*$ on $|0,+\infty)$ and since
$\lambda^R(1,\kappa^*)=\Re \lambda_1^R(\kappa^*)$ we readily obtain,
for any bounded interval $[0,\kappa_0^*]$, that
\begin{equation}  \label{eq:66}
\sup_{\kappa^* \in [0,\kappa_0^*]} \lambda^R(1,\kappa^*)< + \infty\,.
\end{equation}

Let then $\epsilon>0$ and 
\begin{displaymath}
  D_\rho= \{z\in\C\,|\, |z| \geq \rho,\; \Re z\leq\lambda^R(1,\kappa^*) -\epsilon\,\} \,.
\end{displaymath}
By applying the same technique as in \cite{Hel2,GHH} we can prove that
there exist $\rho_0>0$ and $C>0$ such that for all $\rho > \rho_0\,$,
\begin{equation}  \label{eq:68}
  \sup_{\kappa^*\in [0,\kappa_0^*]}\sup_{\lambda\in D_\rho(\kappa^*)}\| (\LL^R(1,\kappa^*)-\lambda)^{-1}\|\leq C \,.
\end{equation}

Next, let $|\lambda|<\rho $, and
$\Re\lambda\leq\lambda^R(1,\kappa^*)-\epsilon$. Here we can bound the
resolvent norm by its Hilbert-Schmidt norm and then use \eqref{eq:63}
to obtain
\begin{displaymath} 
  \| (\LL^R(1,\kappa^*)-\lambda)^{-1}\|\leq C \, \Big|\frac{i e^{-i 2\pi/3 }
    \Ai'(e^{-i 2\pi/3 } \lambda) +\kappa^* \Ai(e^{-2i\pi/3 }\lambda)}
{i e^{i 2\pi/3 } \Ai'(e^{i 2\pi/3 } \lambda) + \kappa^* \Ai(e^{i2\pi/3 } \lambda)}\Big|\,,
\end{displaymath}
where $C$ is independent of $\kappa^*$ in $[0,\kappa_0^*]$. Hence, we may
infer from the above and \eqref{eq:66} that
\begin{displaymath}
  \sup_{\kappa^*\in  [0,\kappa_0^*]} \sup_{
    \begin{subarray}{r}
   \quad  \qquad   | \lambda|< \rho  \\
   \quad  \qquad   \Re\lambda\leq \lambda^R(1,\kappa^*)-\epsilon
    \end{subarray}}
\| (\LL^R(1,\kappa^*)-\lambda)^{-1}\|\leq C_\epsilon \,.
\end{displaymath}
Combining the above with \eqref{eq:68} yields that for some $M_\epsilon>0$
\begin{displaymath}
  \sup_{\kappa^*\in  [0,\kappa_0^*]} \sup_{\Re\lambda\leq\lambda^R(1,\kappa^*) -\epsilon}
\| (\LL^R(1,\kappa^*)-\lambda)^{-1}\|\leq M_\epsilon \,.
\end{displaymath}
Since $M_\epsilon$ is independent of $\kappa^*\in [0,\kappa_0^*]$ we can deduce from
the Gearhart-Pr\"uss Theorem (cf. \cite{H} or \cite{Sj}) that for some
$C_\epsilon>0$, independent of $\kappa^*\in [0,\kappa_0^*]$,
\begin{displaymath}
   \|e^{ -t\LL^R(1,\kappa^*)}\|\leq C_\epsilon e^{-t(\lambda^R(1,\kappa^*)-\epsilon)}\,.
\end{displaymath}
The proposition can now be proved by applying the inverse of
(\ref{eq:69a}).  
\end{proof}

In Section \ref{s7} we will need a stronger estimate than
(\ref{eq:104}).
\begin{proposition}
\label{lemmaunifter}
Let $\lambda^R(\mathfrak{j},\kappa)$ denote the real value of the
leftmost eigenvalue of $\LL^R(\mathfrak{j},\kappa)$. Then for any
positive $\kappa_0$, $\mathfrak{j}_0\,$ and $\mathfrak{j}_1$ there
exists $C(\mathfrak{j}_0,\mathfrak{j}_1,\kappa_0)>0$ such that, for $0
< \mathfrak{j}_0 \leq \mathfrak{j} \leq \mathfrak{j}_1$ and $\kappa\in
[0,\kappa_0]$, \\
\begin{equation}
\label{eq:68R}
\|e^{ -t\LL^R (\mathfrak{j},\kappa) }\|\leq C (\mathfrak{j}_0,\mathfrak{j}_1,\kappa_0) \, e^{-t\lambda^R(\mathfrak{j},\kappa)}\,.
\end{equation}
\end{proposition}
\begin{proof} 
As in the previous proof, we can reduce after dilation the proof to
the case $\mathfrak{j}=1$. We then need to control the uniformity of
the various estimates with respect to $\kappa$ in $[0,
\kappa_0/\jf_0^\frac 13]$ Denote by
$(\lambda_1^R(\kappa),v_1^R(\cdot,\kappa))$ the eigenpair of
$\LL^R(1,\kappa)$ for which
$\Re\lambda_1^R(\kappa)=\lambda^R(1,\kappa)$ and $\|v_1\|=1$.  In
\cite{GHH} (see also Appendix \ref{appC}) we show that for any $\kappa
\geq 0$, $\lambda_1^R(\kappa)$ is simple and unique.  Let then
$\Pi_1^R(\kappa)$ denote the projection on ${\rm span}
(v_1^R(\cdot,\kappa))$, i.e,
\begin{equation}   \label{eq:70}
   u\mapsto  \Pi_1^R (\kappa) u =  \frac{ \langle u, \bar{v}_1^R(\cdot,\kappa) \rangle}{|\langle
     v_1^R(\cdot,\kappa), \bar{v}_1^R(\cdot,\kappa) \rangle|} \, v_1^R(\cdot,\kappa)\rangle  \,. 
\end{equation}
Clearly,
\begin{equation}
  \| \Pi_1^R (\kappa)  \|_{\mathcal L (L^2_+)} = |\langle v_1^R(\cdot,\kappa), \bar{v}_1^R(\cdot,\kappa) \rangle|^{-1} \,.
\end{equation}
We refer the reader to \cite[Section 6 ]{GHH} for the derivation of
the above relation (where an explicit expression of $v_1^R$ in terms
of Airy function is provided). It can be verified \cite{GH} that
$\|\Pi_1^R\| $ is uniformly bounded when $\kappa$ belongs to any
bounded interval in $\overline{\R_+}$.
  
Let $E_1= (I-\Pi_1^R)L^2(\R_+)$, and $\LL^{R,1}=\LL^R(I-\Pi_1^R)$.  We
may define $\LL^{R,1}$ on $(I-\Pi_1^R)D(\LL^R)$, which is clearly a
dense set, in $L^2$ sense, in $E_1$.  By Riesz-Schauder theory we have
that
\begin{equation}   \label{eq:65}
  (\LL^R-\lambda)^{-1}= \frac{\Pi_1^R}{\lambda-\nu_1} +  T_1(\lambda,\kappa)\,,
\end{equation}
where $T_1$ is holomorphic for all $\lambda$ satisfying
$\Re\lambda<\lambda^{R,2}$, in which
\begin{equation}   \label{eq:91}
\lambda^{R,2}= \Re\lambda^R_2>\lambda^R\,.
\end{equation}

By applying the same techniques as in the previous proposition we can
prove that, for any $\kappa_0 >0$ and $\epsilon>0$ there exists
$C_\epsilon>0$ such that
\begin{equation}   \label{eq:66a}
  \sup_{\kappa\in[0,\kappa_0]} \sup_{\Re\lambda\leq\lambda^{R,2}-\epsilon}\| T_1(\lambda) \|\leq C_\epsilon \,.
\end{equation}
Restricting $T_1(\lambda) $ to $E_1$ (onto $D(\LL^{R,1})$) we may
write $T_1|_{E_1}=(\LL^{R,1}-\lambda)^{-1}$. By \eqref{eq:66} and the
Gearhardt-Pr\"uss Theorem we then obtain that for every $\epsilon>0$
there exists $C_\epsilon>0$ such that
\begin{equation}   \label{eq:14}
  \sup_{\kappa^*\in[0,\kappa_0^*]}\|e^{-t\LL^{R,1}}\|\leq C_\epsilon e^{-(\lambda^{R,2}-\epsilon)t} \,.
\end{equation}
We complete the proof of \eqref{eq:68R} by observing that
\begin{displaymath}
  e^{-t\LL^R}=e^{-t\LL^R}\Pi_1^R+e^{-t\LL^{R,1}}
\end{displaymath}
and setting $\epsilon<\lambda^{R,2}-\lambda^R$.
\end{proof}
\begin{remark}
The estimate \eqref{eq:68R} remains valid at $\kappa=0$, i.e., for
Neumann boundary condition.  For Dirichlet boundary conditions it is
an immediate result of \cite[Lemma 4.2]{AlHe}.
\end{remark}

We conclude this section by making the following simple observation
\begin{lemma}
\label{lem:deriv-R}
Under the previous assumptions, there exists $C(\jf_0,\jf_1,\kappa_0)$
such that if $\jf\in[\jf_0,\jf_1]$ and $\kappa\in [0,\kappa_0]$ then
\begin{equation}     \label{eq:76}
\Big|\frac{\partial \lambda^R_1}{\partial \jf}(\jf,\kappa)\Big|\leq  C(\jf_0,\jf_1,\kappa_0) \,.
\end{equation}
\end{lemma}
\begin{proof}
The proof is immediate from the Feynman-Hellman formula
\begin{displaymath}
    \Big|\frac{\partial \lambda^R_1}{\partial \jf}(\jf,\kappa)\Big| =i\frac{\langle x\, \bar{u}_1^R,v_1^R\rangle}{\langle\bar{u}_1^R,v_1^R\rangle}\,,
\end{displaymath}
and from the fact that $\lambda^R_1$ is simple and hence
$\langle\bar{u}_1^R,v_1^R\rangle\neq0\,$.
\end{proof}

\section{The complex Airy operator with a semi-permeable barrier: definition and properties}\label{s3} 

For $\kappa \geq 0\,$, $\jf \neq 0$ and $\nu \geq0\,$, we consider the
sesquilinear form $a_\nu$ defined for $u=(u_-,u_+)$ and $v=(v_-,v_+)$
by
\begin{eqnarray}
a_\nu ^{T}(u,v) &=& \int_{-\infty}^0\Big(u_-'(x)\,\bar v_-'(x) +i \,\jf\, xu_-(x)\, \bar v_-(x)+\nu\, u_-(x)\,\bar v_-(x)\Big)\,dx \nonumber \\
& & + \int_0^{+\infty}\Big(u_+'(x)\, \bar v_+'(x) +i \,\jf\, xu_+(x)\, \bar v_+(x)+\nu\, u_+(x)\, \bar v_+(x)\Big)\,dx \nonumber \\
& &  + \kappa \,  \big(u_+(0)-u_-(0)\big)
 \big(\overline{v_+(0)-v_-(0)}\big)\,, \label{defForm1d}
\end{eqnarray}
where the form domain $\mathcal V^T$ is
\begin{equation*}
 \mathcal V^T := \Big\{u=(u_-,u_+)\in H_-^1\times H_+^1 : |x|^\frac 12\, u\in L_-^2\times L_+^2\Big\}\,,
\end{equation*}
with $ L_\pm^2= L^2(\mathbb R_\pm)$, $H_-^s= H^{s}(\mathbb R_\pm)$.\\

The space $\mathcal V$ is endowed with the  Hilbertian norm
\begin{equation*}
 \|u \|_{\mathcal V} := \sqrt{ { \|u_-\|_{H_-^1}^2+\|u_+\|_{H_+^1}^2 } +\||x|^{1/2}u\|_{L^2}^2}\,.
\end{equation*}
To give a precise mathematical definition of the associated closed
operator, we cannot, due to the lack of coercivity, use the standard
version of the Lax-Milgram theorem.  In \cite{GHH} a generalization of
the Lax-Milgram theorem, introduced in \cite{AH}, is used to obtain
that
\begin{proposition}
The operator $\mathcal L^{T} (\jf,\kappa)$ acting as
\begin{equation*}
 u\mapsto   \mathcal L^{T} (\jf ,\kappa) u  = \left(-\frac{d^2}{dx^2}u_-+i\,\jf\, x\, u_-\,,\, -\frac{d^2}{dx^2}u_+ +i \, \jf \, x\, u_+\right)
\end{equation*}
on the domain
\begin{equation}   \label{defdomA1} 
 \mathcal D(\mathcal L^{T}(\jf,\kappa) ) = \big\{u\in H_-^2\times H_+^2 : x\,u\, \in
 L_-^2\times L_+^2 \text{and } u\in  \mathfrak D^T  \big\}\,,
\end{equation}
where $ \mathfrak D^T(\kappa) $ is given by \eqref{eq:98}, is a closed
operator with compact resolvent.\\ There exists some $\nu \in
[0,+\infty)$ such that the operator $\mathcal L^{T}(\jf,\kappa) +\nu $
is maximal accretive.
\end{proposition}
Maximal accretiveness of $\mathcal L^{T}(\jf,\kappa) +\lambda$ for all
$\lambda\in\R_+$ can be proved in the following manner.  Denote by
$(\mathcal L^{T}(\jf,\kappa))^*$ the adjoint of $\mathcal
L^{T}(\jf,\kappa)$.  By the above construction it is simply $\mathcal
L^{T}(-\jf,\kappa)$.  Since $\mathcal L^{T}(\pm \jf,\kappa) +\lambda$
is accretive whenever $\Re\lambda>0$, it follows by \cite[Theorem
II.3.17]{enna00} that $\mathcal L ^{T}(\jf,\kappa)+\lambda$ is maximal
accretive, and hence generates a contraction semigroup.

As in the previous section we have  
\begin{proposition}
\label{propSchatten}
For any $\lambda\in\rho(\mathcal L ^{T}(\jf,\kappa)) $,
$(\LL^T(\jf,\kappa) -\lambda)^{-1}$ belongs to the Schatten class
$\mathcal C^p$ for any $p > \frac 32$.
\end{proposition}
In contrast with the previous section, however, the numerical range of
$\LL^T(\jf,\kappa) $ is not embedded in the first quadrant of the
complex plane, but instead covers its right half. Hence, we only have
\begin{equation}  \label{eq:79}
  \|\LL^T(\jf,\kappa)  -\lambda)^{-1}\| \leq \frac{1}{|\Re \lambda|}\,,\, \mbox{ if } \Re \lambda < 0\,,
\end{equation}
Since the above bound is not enough to establish completeness of the
system of the eigenfunctions of $\LL^T(\jf,\kappa) $ in $ L_-^2\times
L_+^2$ an additional estimate is necessary. It has been established in
\cite{GHH} that there exists $M>0$ such that for all $\lambda\in\R_+$
we have
\begin{displaymath}
  \|(\LL^T(\jf,\kappa) -\lambda)^{-1}\| \leq M(1+|\lambda|)^{-\frac 14} (\log \lambda)^\frac 12 \,.
\end{displaymath}
The above, together with \eqref{eq:79}, the Phragm\'en-Lindel\"of
principle, and the fact that the resolvent is in $\mathcal C^p$, for
any $p>\frac 32$, implies, modulo the proof that all the eigenvalues
are simple,
\begin{proposition}\label{CompR}
For any $\kappa \geq 0$, the space generated by the eigenfunctions of
$\LL^T(\jf,\kappa) $ is dense in $L^2_-\times L^2_+$.
\end{proposition}

We have hence to prove the simplicity. We can reduce the proof to $\jf
=-1$.  We recall from \cite{GHH} that the eigenvalues of
$\LL^T(-1,\kappa) $ are determined by
\begin{displaymath}
  f_\kappa \overset{def}{=} f(\lambda)+\frac{\kappa}{2\pi}=0 \,, 
\end{displaymath}
where
\begin{displaymath}
  f(\lambda) = \Ai^\prime(e^{i2\pi/3} \lambda) \, \Ai^\prime(e^{-i2\pi/3} \lambda)\,,
\end{displaymath}
is entire.\\

\begin{lemma}
\label{lem:jordan}
All eigenvalues of  $\LL^T(-1,\kappa) $ are simple.
\end{lemma}
\begin{proof}
Recall that if $\mu\in\sigma(\LL^T(-1,\kappa) )$ then
\begin{displaymath}
    f(\mu) = \Ai'(e^{i2\pi/3} \mu) \, \Ai'(e^{-i2\pi/3} \mu) =- \kappa/2\pi \,.
\end{displaymath}
Suppose further that $f^\prime(\mu)=0$. It has been established in
\cite{GHH} that
\begin{equation}  \label{eq:122}
  \Ai(e^{i2\pi/3} \mu) = \frac{e^{i\frac \pi 6}}{2\kappa }\Ai^\prime(e^{i2\pi/3} \mu)
  \quad ; \quad   \Ai(e^{-i2\pi/3} \mu) = \frac{e^{- i\frac \pi 6}}{2\kappa }\Ai^\prime(e^{-i2\pi/3} \mu)\,.
\end{equation}
Let $u=(u_+,u_-)$ denote the eigenfunction associated with $\mu$. It
can be easily verified that
\begin{equation}   \label{eq:84}
  \begin{cases}
    u_-(x) = C_-\Ai(-e^{i\pi/6}(x-i\mu)) & x<0 \\
    u_+(x) = C_+\Ai (e^{-i\pi/6}(x-i\mu)) & x>0 \,.
  \end{cases}
\end{equation}
We can now rewrite \eqref{eq:122} in the following manner
\begin{displaymath}
  u_-(0) = -\frac{1}{2\kappa}u_-^\prime(0)  \quad ; \quad    u_+(0) = \frac{1}{2\kappa}u_+^\prime(0) 
\end{displaymath}
It follows that both $\mu$ and $\bar{\mu}$ are eigenvalues of
$\bar{\LL}^R(1,2\kappa)$. This is, however, a contradiction, as
$\sigma(\bar{\LL}^R(1,2\kappa))$ lies in the fourth quadrant, and
since $\mu\notin\R$. 
\end{proof}
Before providing some semigroup estimates, as in the previous section,
we need to establish another auxiliary result. 
\begin{lemma}
Let $\omega\in\R_+$. Let $Z(\kappa,\omega)\in\Z_+$ denote the number
of zeros of $f_\kappa$ for $\Re\lambda\leq \omega$.  Then, for every
$\kappa_0>0$ there exists $M(\omega,\kappa_0)$ such that
\begin{equation}     \label{eq:17}
\sup_{\kappa\in[0,\kappa_0]} Z(\kappa,\omega) \leq M (\omega,\kappa_0)\,.
  \end{equation}
\end{lemma}
\begin{proof}
Let
\begin{displaymath}
    D_\omega =\{ z\in\C \,| \, 0\leq\Re z\leq\omega \} \,. 
\end{displaymath}
The number of zeros of $f_\kappa$ in $D_\omega$ is precisely $Z$ since
there are no eigenvalues of $A_1^{+,T}$ in the left side of the
complex plane (when $\kappa \geq 0$).  \\
From  the analysis of the resolvent for $\Im \lambda$ large (see
\cite{GHH}) and the continuity of the zeros, we deduce that there
exists $L(\kappa_0,\omega)$ such that, for all $\kappa \in
[0,\kappa_0]$, the number of zeros of $f_\kappa$ in
\begin{displaymath}
  D^L_\omega = \{ z\in\C \,| \, 0\leq\Re z\leq\omega \;;\; |z|\leq L(\kappa_0,\omega)\} \,,
\end{displaymath}
is precisely $Z(\kappa,\omega)$. 

To prove \eqref{eq:17} we now argue by contradiction. Suppose that
there exists $\{\kappa_j\}_{j=1}^\infty\subset[0,\kappa_0]$ such that
$Z(\kappa_j,\omega)\to\infty$.  Without loss of generality we assume
$\kappa_j\to\kappa_\infty$, otherwise we move to a subsequence. Let
$\CC$ denote a closed path in $\C\setminus D^L_\omega$ enclosing
$D_\omega^L$ such that $f_{\kappa_\infty}\neq0$ on $\CC$. For
sufficiently large $j$, $f_{\kappa_j}\neq0$ on $\CC$ and we may use
Rouch\'e's theorem to obtain
\begin{displaymath}
  \frac{1}{2\pi i} \oint_{\CC}\frac{f^\prime}{f_{\kappa_j}} d\lambda \to +\infty\,,
\end{displaymath}
leading to a contradiction as
\begin{displaymath}
    \frac{1}{2\pi i} \oint_{\CC}\frac{f^\prime}{f_{\kappa_\infty}} d\lambda < + \infty \,.
\end{displaymath}
\end{proof}

As in the previous section we prove a semigroup estimate. 
\begin{proposition}\label{lemmaunifbisT} 
Let $\lambda^T(\mathfrak{j},\kappa)$ denote the real value of the
leftmost eigenvalue of $\LL^T (\mathfrak{j},\kappa)$. Then for any
positive $\mathfrak{j}_0<\mathfrak{j}_1$, $\kappa_0$ and $\epsilon$,
there exists $C(\mathfrak{j}_0,\mathfrak{j}_1,\kappa_0,\epsilon)>0$
such that, for any $\mathfrak{j}_0\leq \mathfrak{j} \leq
\mathfrak{j}_1$ and $\kappa \in [0,\kappa_0]$
\begin{equation}   \label{eq:103}
\|e^{ -t\LL^T (\mathfrak{j},\kappa) }\|\leq C
(\mathfrak{j}_0,\mathfrak{j}_1,\kappa_0, \epsilon) \,
e^{-t(\lambda^T(\mathfrak{j},\kappa)- \epsilon)}\,. 
\end{equation}
\end{proposition}
We skip the proof as it is identical with the proof of Proposition
\ref{lemmaunifbis}.  One can also improve the proposition in the
following way:
\begin{proposition}
\label{lemmaunifbisTimp} 
Let $\mathfrak{j}_0 < \jf_1$ and $\kappa_0$ be positive constants.
Then, there exists $C(\mathfrak{j}_0,\mathfrak{j}_1,\kappa_0)>0$ such
that, for any $0 < \mathfrak{j}_0\leq \mathfrak{j} \leq
\mathfrak{j}_1$ and $\kappa \in [0,\kappa_0]$ 
\begin{equation}   \label{eq:3.8}
\|e^{ -t\LL^T (\mathfrak{j},\kappa) }\|\leq C
(\mathfrak{j}_0,\mathfrak{j}_1,\kappa_0) \,
e^{-t\lambda^T(\mathfrak{j},\kappa)}\,. 
\end{equation}
\end{proposition}
\begin{proof}
The proof is similar to the proof of Proposition \ref{lemmaunifter},
and we provide therefore only its outlines. Let
$\lambda\in\sigma(\LL^T)$. It has been proved in \cite{GHH} that
$\lambda\notin\R$ and that there are at least two complex conjugate
eigenvalues with a real value equal to $\lambda^T$.  With the proof of
simplicity in mind, we get $K$ pairs of complex conjugate eigenvalues
with same real part $ \lambda^T(\mathfrak{j},\kappa)$ and
$K(\jf,\kappa)$ is uniformly bounded by \eqref{eq:17}.  Let $U_K={\rm
span}\{u_1,\ldots,u_{2K}\}$ denote the space spanned by all the
eigenfunctions (and generalized eigenfunctions) of $\LL^T$
corresponding to eigenvalues whose real value is equal to $\lambda^T$.
Let
\begin{displaymath}
    P_k^T = \sum_{\ell =1}^{2k} \Pi_\ell \,,
\end{displaymath}
where $\Pi_\ell $ denotes the projection on $u_\ell $ defined in
\eqref{eq:70}.  Using the same technique as in the proof of
Proposition \ref{lemmaunifbis}, we then conclude for any $\epsilon>0$,
the existence of $C_\epsilon(\jf_0,\jf_1,\kappa_0)>0$ such that
\begin{equation}   \label{eq:22}
  \|e^{-t\LL^T}(I-P_k^T)\|\leq C_\epsilon(\jf_0,\jf_1, \kappa_0)\,e^{-(\lambda^{T,2}-\epsilon)t} \,,
\end{equation}
where
\begin{equation}    \label{eq:92}
\lambda^{T,2}=\inf \Re\sigma\big(\LL^T(I-P_k^T) \big)>\lambda^T \,.
\end{equation}
The proposition now follows from the fact that
\begin{displaymath}
  \|e^{-t\LL^T}P_k^T \|\leq \left(  \sum_{\ell=1}^{2k} \|\Pi_\ell \| \right) \, e^{-\lambda^Tt}  \leq  C(\jf_0,\jf_1,\kappa_0) e^{-\lambda^Tt} \,.
\end{displaymath}
\end{proof}
We conclude this section by making the following simple observation
\begin{lemma}
\label{lem:deriv-T}
Let $0 < \jf_0 < \jf_1$ and $\kappa_0 >0$.  Then, there exists
$C(\jf_0,\jf_1,\kappa_0)>0$ such that, for $\jf\in[\jf_0,\jf_1]$ and
$\kappa \in [0, \kappa_0]$,
\begin{equation}   \label{eq:77} 
\Big|\frac{\partial \lambda^T}{\partial \jf}(\jf,\kappa) \Big|\leq C(\jf_0,\jf_1,\kappa_0)\,.
\end{equation}
\end{lemma}
The proof is identical with the proof of Lemma \ref{lem:deriv-R}.

\section{Limit problems: linear potential}\label{sModels}

In this section, we consider the simplified cases where $\Omega$ is
either the entire space $\mathbb{R}^n$, or the half-space
\begin{equation}\label{defHalfSpace}
\mathbb{R}_+^n = \{x=(x_1,\dots,x_n)\in\mathbb{R}^n : x_n>0\}.
\end{equation}
Furthermore the potential $V$ will be assumed to be a linear function. We note that these relatively simple problems naturally
arise in the semi-classical limit $h\to0$ of $(\A_h-\lambda)^{-1}$,
where $\A_h$ is given by \eqref{eq:1} or \eqref{eq:7}.

\subsection{The entire space problem}
\label{ssModels}
In this subsection, we mainly refer to \cite{Hel2}, and reformulate
the $2$-dimensional statements therein in the $n$-dimensional setting.
Let $ \vec J=(J_1,\dots,J_n)\in\mathbb{R}^n$ and
\begin{displaymath}
  \mathcal A_0 = -\Delta+i\ell 
\end{displaymath}
acting on $L^2(\mathbb{R}^n)\,$, where $ \ell(x) = \vec J\cdot
x\,$. Up to an orthogonal change of variable followed by the scale
change $ x\mapsto J^{2/3}x$, where we use the notation
\begin{equation*}
J = |\vec{J}|\,,
\end{equation*}
we can assume that $\mathcal A_0$ has the form
\begin{displaymath}
  \mathcal A_0 = -\Delta+ix_n\,.
\end{displaymath}
Let $x=(x^\prime,x_n)$. Applying a partial Fourier transform in
$x^\prime$
\begin{displaymath}
  \F(u)=(2\pi)^{-(n-1)/2}\int_{\R^{n-1}}e^{-i\omega'\cdot x^\prime}\,dx^\prime \,,
\end{displaymath}
we obtain the transformed operator on $L^2(\mathbb R_{\omega'}^{n-1}
\times \mathbb R_{x_n})$
\begin{displaymath}
  \hat{\A}_0= -\partial_{x_n}^2+ix_n + |\omega'|^2\,.
\end{displaymath}
with domain
\begin{displaymath}
 D ( \hat{\A_0}) = \big\{ u \in L^2(\mathbb R^n), |\omega'|^2 u \in L^2 (\mathbb R^n),\partial_{x_n}^2 u \in L^2(\mathbb R^n)\,,\,
 x_n u \in L^2(\mathbb R^n)\big\}\,.  
\end{displaymath}
from which we easily obtain that
\begin{equation}   \label{eq:93}
D ( \mathcal A_0) = \{ u\in  H^2(\mathbb R^n)\,|\, x_nu \in L^2(\mathbb R^n)\}\,.
\end{equation}
Recalling that the complex Airy operator
\begin{equation}
 \mathcal L  =  -\frac{d^2}{dx_n^2}+ix_n
\end{equation}
on $L^2(\mathbb{R})$ has empty spectrum, we then get as in
\cite[Proposition $7.1$]{Hel2} the following lemma.
\begin{lemma}\label{lemModelRn}
We have $\sigma(\mathcal A_0) = \emptyset$, and for all
$\omega\in\mathbb{R}$, there exists $C_\omega^0$ such that
\begin{equation}   \label{eq:12}
\sup_{\Re z\leq\omega}\|(\mathcal A_0-z)^{-1}\|\leq C_\omega^0.
\end{equation}
\end{lemma}
\begin{proof}
Set
\begin{displaymath}
    \A_0 = -\Delta_{x^\prime} + \mathcal L  
\end{displaymath}
(equivalently we may set $\overline{-\Delta_{x^\prime} \otimes I + I
\otimes \mathcal L}$).\\
As
\begin{displaymath}
  e^{-t\A_0} = e^{t\Delta_{x^\prime}}\otimes e^{-t\mathcal L } \,,
\end{displaymath}
and since (see \cite{da07})
\begin{displaymath} 
  \|e^{-t\mathcal L }\|\leq e^{-\frac{t^3}{12}} \,,
\end{displaymath}
we obtain that
\begin{equation}
\label{eq:5} 
  \|e^{-t(\A_0-z)}\| \leq e^{-\frac{t^3}{12}+t\, \Re z } \,.
\end{equation}
Recalling that (see \cite{da07})
\begin{equation}   \label{eq:10}
   (\A_0-z)^{-1} = \int_0^{+\infty}e^{-t(\A_0-z)}dt\,,
\end{equation}
we may conclude that for any $\omega>0$ we have
\begin{displaymath} 
  \sup_{\Re z\leq\omega}\|(\A_0-z)^{-1}\|\leq  \int_0^{+\infty}e^{-\frac{t^3}{12}+t\omega}\,dt \,.
\end{displaymath}
As $\omega \to +\infty$, we have, using Laplace method,
\begin{equation}   \label{asy}
 \int_0^{+\infty}e^{-\frac{t^3}{12}+t\omega}\,dt \sim C\,  \omega^{-\frac 14} \, e^{\frac 43 \omega^\frac 32} \,,
\end{equation}
which gives a more precise information on $C_\omega^0$.  A universal
upper bound is obtained in Appendix \ref{appB}.
\end{proof}

\subsection{The half-space problem: Definitions}
\label{ssModelQcq}

\subsubsection{Notation}
Let
\begin{displaymath}
  \R^n_\pm  = \{ (x_1,\ldots,x_n)\in\R^n \, | \, \pm x_n>0\}\,,
\end{displaymath}
and $\vec J=(J^\prime,J_n)\in\mathbb{R}^n\setminus\{0\}$ where
$J^\prime\in\R^{n-1}$. We study here the spectrum and the behavior of
the resolvent of the operator
\begin{displaymath}
 \A^\# =\A^\# (\vec J): =-\Delta+i \vec J\cdot x
\end{displaymath}
acting on $L^2(\mathbb{R}_+^n)$ or on $L^2(\R^n_-)\times L^2(\R^n_+)$. \\
Here the superscript $\#$ means that $\A^\#$ is defined on a subset
of $ \mathfrak D^\#_n$, where 
\begin{equation}   \label{eq:102} 
  \begin{cases}
    \mathfrak D^\#_n=\{u\in H^2_{loc}(\overline{\R_+^n}) \,| \, u(x^\prime,0)=0 \} & \#=D \\
    \mathfrak D^\#_n = \{u\in H^2_{loc}(\overline{\R_+^n}) \,| \, u_{x_n}(x^\prime,0)=0 \} & \#=N \\
    \mathfrak D^\#_n = \{u\in H^2_{loc}(\overline{\R_+^n}) \,| \, u_{x_n}(x^\prime,0)= \, \kappa \,  u(x^\prime,0) \} & \#=R \\
    \mathfrak D^\#_n = \{u\in H^2_{loc}(\overline{\R_-^n}) \times
     H^2_{loc}(\overline{ \R_+^n}) \,|  \partial_{x_n} u_+(x^\prime ,  0)= & \\
    \kern 10em  \partial_{x_n}u_-(x^\prime ,  0) = \kappa[u_+  -  u_-]{(x^\prime,0)} \} &   \#=T \,.
  \end{cases}
\end{equation}
Naturally, $\mathfrak D^\#_n$ depends on $\kappa \geq 0$, when $\#\in
\{R,T\}$. We shall occasionally, therefore, write $\A^\#(\vec J,\kappa)$
or $\A^\#(\kappa)$ to emphasize the dependence on the parameters.  
We now attempt to obtain the domain of definition of $\A^\#$ so that
$\A^\#:D(\A^\#)\to L^2(\R^n_\#)$ is surjective (where
$\R^n_\#=\R^{n-1}\times\R_\#$). This has been
accomplished for Dirichlet boundary conditions by R. Henry
\cite{Hen2}. We now employ the same technique as in \cite{Hen2} to
obtain $D(\A^\#)$ when $\#\in\{N,R,T\}$. \\

\subsubsection{Definition by using the separation of variables}

We set, as above, $x'=(x_1,\dots,x_{n-1})$ to be the $(n-1)$ first
coordinates of a vector $x\in\mathbb{R}^n$ and $\vec
J=(J^\prime,J_n)$, and then let
\begin{equation}   \label{defCA_x'}
 \mathcal A_{x'} = \mathcal A_{x'} (J')= -\Delta_{x'}+iJ'\cdot x'\,,
\end{equation}\
acting on $L^2(\R^{n-1})$ and
\begin{equation}   \label{eq:9}
  \mathcal L _{x_n}^\#  = \mathcal L _{x_n} ^\# (J_n)=  -\frac{\partial^2}{\partial x_n^2} + i J_n x_n \,,
\end{equation}
acting on $L^2_\#$, which denotes either $L^2(\R_+,\C)$ when
$\#\in\{D,N,R\}$ or $ L^2(\R_+,\C)\times L^2(\R_-,\C)$ for the
transmission problem. \\
We recall that the domains of $\A_{x'}$ and $ \mathcal L _{x_n} ^\#$
have been well identified in \cite{Alm, Hen2,GHH,GH} and that
estimates for the resolvent have been established in each case (these
depend on $J'$, $J_n$ and on the value of $\#$ in $\{D,N,R,T\}$).  In
addition, it has been established, with the aid of the Hille-Yosida
theorem, that both $\mathcal A_{x'}$ and $ \mathcal L _{x_n} ^\# $ are
generators of contraction semigroups $(e^{-t\mathcal A_{x'}})_{t>0}$
and $(e^{-t \mathcal L _{x_n} ^\#})_{t>0}$, respectively (recall that
$\kappa\geq0$ when $\#\in\{R,T\}$). It can be easily verified that the
family $(e^{-t\mathcal A_{x'}}\otimes e^{-t \mathcal L _{x_n}
^\#})_{t>0}$ is a contraction semigroup in $L^2_\#$.  Thus, we can
define
\begin{definition}  \label{Definition4.2}
$\A^{\#,sv}$ is the generator of the semigroup $(e^{-t\mathcal
A_{x'}}\otimes e^{-t \mathcal L _{x_n}^\# })_{t>0}\,$.
\end{definition}
To derive the domain of the operator $ \A^{\#,sv}$, we follow the
approach of \cite{Hen2}. Recall, then, the following result (see
\cite[Theorem X.$49$]{ReSi}):
\begin{theorem}
\label{thmcoeur}
Let $\mathcal A$ be the generator of a contraction semigroup on a
Hilbert space $\mathcal H$. Let $\mathcal C \subset\mathcal D(\mathcal
A)$ be a dense subset of $\mathcal H$, such that
\begin{equation}   \label{stabi}
  e^{-t\mathcal A}\, \mathcal C \subset\mathcal C\,,
\end{equation}
Then, $\ \mathcal C$ is a \emph{core} for $\mathcal A\,$, that is
\begin{equation*}
\mathcal A = \overline{\mathcal A_{/\mathcal C}}\,.
\end{equation*}
\end{theorem}
We now apply this theorem with $\mathcal A= \A^{\#,sv}$  and 
\begin{equation}  \label{choice1}
 \mathcal C=\mathcal D (\mathcal A_{x'})\odot\mathcal D( \mathcal L _{x_n}^\# )
\end{equation}
that is the set of all finite linear combinations of functions of the
form $f\otimes g=f(x')g(x_n)$, where $f\in\mathcal D (\mathcal
A_{x'})$ and $g\in\mathcal D ( \mathcal L _{x_n}^\# )$. \\ Then it is
clear that $\mathcal C$ satisfies the conditions of Theorem
\ref{thmcoeur}, hence
\begin{equation}  \label{defdom}
\A^{\#,sv} = \overline{{\mathcal
    A^{\#,sv}}_{/\mathcal C}}\,.
\end{equation}
\begin{remark}\label{choixdeC}~\\
Instead of $\mathcal D( \mathcal L _{x_n}^\# )$ in \eqref{choice1} we
may use the span of the eigenfunctions of $ \mathcal L _{x_n}^\#
$. Note that it has been shown in \cite{Alm,Hen2,GHH} that this space
is dense in $L^2_\#$. Furthermore, it is a subset of $\mathcal S_\#$
and \eqref{stabi} is still satisfied, and hence we may still apply to
it Theorem \ref{thmcoeur}.  In this way, it is easier to estimate any
expression involving functions in $\CC$.
\end{remark}

We may now conclude from \eqref{defdom} that the domain of $A^{\#,sv}$
is given by
\begin{eqnarray}   \label{carac}
  \mathcal D (\A^{\#,sv}) &=& \{u\in L^2_\# : \exists (u_j)_{j\geq1}\in\mathcal C^\mathbb{N}\,,~
  u_j\underset{\tiny{j\to+\infty}}{\overset{L^2_\#}{\longrightarrow}}u\,,\nonumber\\
  & &\quad\quad \quad\quad  (\A^{\#,sv}u_j)_{j\geq1}~\textrm{is
    a Cauchy sequence in } L^2_\#~\}\,.\label{caracDomA+} 
\end{eqnarray}

Clearly, $(-\Delta+i\ell)u\in L^2_\#$ for any $u\in \mathcal D
(\mathcal A^{\#,sv}) $. Nevertheless, it is not obvious that $\Delta
u\in L^2_\#$ (and hence also $\ell u\in L^2_\#$). In the following we
obtain bounds on $\|u\|_{H^2(\mathbb{R}_+^n)}$ and $\|\ell
u\|_{L^2(\mathbb{R}_+^n)}\,$ in terms of the graph norm $u\mapsto
\sqrt{ \|u\|_2^2+\|\A^{\#,sv}u\|^2_2}$, thereby allowing for a
representation of $\mathcal D (\A^{\#,sv})$, which is much more
transparent than \eqref{defdom} or \eqref{carac}.  We state the result
for the Robin case only, as the statement for the transmission
problem are similar.
\begin{proposition}\label{Proposition4.4}
If $|\vec J|=J \neq 0$, we have
\begin{equation}\label{descrDomA+}
  \mathcal D (\A^{R,sv}) =  H^2(\mathbb{R}_+^n)\cap
  L^2(\mathbb{R}_+^n ; \ell(x)^2\,dx)\cap\mathfrak D^R_n\,.
\end{equation}
Furthermore, for all $u\in\mathcal D (\A^{R,sv})\,$, we have 
\begin{equation}  \label{eq:4.13}
\| \nabla u\|^2 \leq \Re \langle \A^{R,sv}\, u\,,\, u\rangle\,,
\end{equation}
and 
\begin{equation}\label{estDomA+a}
\|\Delta u\|_{L^2(\mathbb{R}_+^n)}^2 + \|\ell u\|_{L^2(\mathbb{R}_+^n)}^2 \leq \|\A^{R,sv} u\|_{L^2(\mathbb{R}_+^n)}^2+
J \, \|\nabla u\|_{L^2(\mathbb{R}_+^n)}\|u\|_{L^2(\mathbb{R}_+^n)}\,. 
\end{equation}
\end{proposition}
\begin{proof}
The proof below is adapted from Henry \cite{Hen2}. Let $u\in\mathcal D
(\A^{R,sv})$. By \eqref{carac} there exists $(u_j)_{j\geq1}\in\mathcal
C^\mathbb{N}$ such that
$u_j\underset{\tiny{j\to+\infty}}{\overset{L^2}{\longrightarrow}}u$
and $(\mathcal A u_j)_{j\geq1}$ is a Cauchy sequence. Then, as
\begin{equation*}
\Re\sca{\mathcal
    A v}{v} = \|\nabla v\|_{L^2(\mathbb{R}_+^n)}^2 + \kappa \int_{\mathbb
    R^{n-1}} |v (x',0)|^2 dx' \geq \|\nabla v\|_{L^2(\mathbb{R}_+^n)}^2 \,,
\end{equation*} 
(note that $\kappa \geq 0$) it follows that $(\nabla u_j)_{j\geq1}$ is
a Cauchy sequence in $L^2(\mathbb{R}_+^n)$, and hence
\begin{equation}
    \label{cvNabla} u_j\underset{\tiny{j\to+\infty}}{\overset{H^1}{\longrightarrow}}    u\,, 
\end{equation}
and $u\in H^1(\mathbb{R}_+^n)\,$.\\
At this stage, we have established that
\begin{equation} 
\mathcal D (\A^{R,sv}) \subset H^1(\mathbb{R}_+^n)\,.  
\end{equation}
Moreover the first trace of $u_j$ on $\{x_n=0\}$, which is defined by
\begin{equation*}
  (\gamma_0u_j) (x')= u_j (x',0)\,, 
\end{equation*}
converges to $\gamma_0 u$ in $H^\frac 12(\mathbb R^{n-1})$. For the
second trace, which is defined by
\begin{equation*}
  (\gamma_1 u_j ) (x')= \partial_{x_n} u_j (x',0)\,, 
\end{equation*}
we use standard elliptic estimates to obtain that $u\in
H^2_{loc}(\R^n_+)$ and hence also that $\gamma_1u\in H^{-\frac
12}_{loc} (\mathbb R^{n-1})$. Consequently, the Robin boundary
condition is satisfied by $u$ as well.\\
In order to prove (\ref{estDomA+a}), we write (all the norms denoting
$L^2$ norms)
\begin{eqnarray} \label{eq:4.20}
    \|\A u_j\|^2 & = & \sca{(-\Delta+i\ell)u_j}{(-\Delta+i\ell)u_j} \nonumber\\
    & = & \|\Delta u_j\|^2+ \|\ell u_j\|^2 + 2\Im\sca{-\Delta u_j}{\ell
      u_j}\,.\label{estDomInter} 
\end{eqnarray} 
Here we use the regularity of $u_j$ which follows from the fact that
$\Dg(\mathcal A_{x'})$ is given by \eqref{eq:93} and $\Dg(\mathcal
L^R_{x_n})$ by \eqref{eq:107}.\\
We now observe that 
\begin{equation} \label{intbyparts} \sca{-\Delta
      u_j}{\ell u_j} = \int_{\mathbb{R}_+^n}\nabla u_j(x)\cdot\overline{\nabla (\ell
      u_j)(x)}dx - \int_{\mathbb R^{n-1}} \partial_{x_n} u_j (x',0) \ell (x',0)
    \bar u_j (x',0)\, dx'\,.  
\end{equation} 
Using the Robin condition, this reads
\begin{equation}\label{intbyparts1} \sca{-\Delta u_j}{\ell u_j}
    = \int_{\mathbb{R}_+^n}\nabla u_j(x)\cdot\overline{\nabla (\ell u_j)(x)}dx + \kappa
    \int_{\mathbb R^{n-1}}\ell (x',0) |u_j (x',0)|^2 \, dx'\,.
\end{equation} 
Hence we have 
\begin{eqnarray*}
    \Im \sca{-\Delta u_j}{\ell u_j} & = & \Im\int_{\mathbb{R}_+^n}\nabla u_j(x)\cdot\overline{\nabla (\ell u_j)(x)}\,dx\\
    & = & \Im\int_{\mathbb{R}_+^n} \left( \vec J\cdot\nabla
    u_j(x)\right) \, \overline{u_j(x)}\,dx\,.  
\end{eqnarray*} 
This implies 
\begin{equation*}
|\Im \sca{-\Delta u_j}{\ell u_j}|\leq J \,\|\nabla u_j\|\, \|u_j\|\,. 
\end{equation*}
Thus, the estimate (\ref{estDomA+a}) holds for $u_j$ for all $1\leq
j$.  Consequently, $(u_j)_{j\geq1}$ is a Cauchy sequence in
$L^2(\mathbb{R}_+^n ; |\ell(x)|^2dx)$ and $\Delta u_j$ is a Cauchy
sequence in $L^2 (\mathbb R^n_+)$.  Hence $u\in H^1( \mathbb{R}_+^n)$,
$\Delta u \in L^2(\mathbb{R}_+^n)$ and $u$ satisfies the Robin
condition.  We can then use the regularity of the Robin Laplacian
which follows from the regularity of the Neumann Laplacian (see
\cite{GH}) in order to show that $u\in H^2(\mathbb R_+^n)$.  We have
thus established that
\begin{displaymath} 
\mathcal D
      (\A^{R,sv}) \subseteq \hat{\Dg}\overset{def}{=}H^2(\mathbb{R}_+^n)\cap
      L^2(\mathbb{R}_+^n ; \ell(x)^2\,dx)\cap\mathfrak D^R_n\,.
\end{displaymath} 
The proof of the converse inclusion goes as follows. Let
$u\in\hat{\Dg}$. Clearly $f=(-\Delta+i\ell +1 )u\in L^2(\R^n_+)$.
Since $-\Delta+i\ell +1: \mathcal D(\A^{R,sv}) \to L^2(\R^n_+)$ is
surjective, it follows that there exists $ v\in\mathcal D(\A^{R,sv})$
such that $(-\Delta+i\ell +1 )v=f$. Thus $w=u-v\in\hat{\Dg}$ and
satisfies $(-\Delta+i\ell +1)w=0$. By the injectivity of
$(-\Delta+i\ell +1)$ on $ \hat \Dg$ we obtain that $u=v$ and hence
$\hat{\Dg}\subseteq\mathcal D(\A^{R,sv})$.
\end{proof}
\begin{remark} 
Similar arguments lead to the proof of the same statement for the
transmission case.
\end{remark}
\begin{remark}\label{RmLM}
Using \cite[Theorem 2.1]{AH}, we can also define the operator via a
generalized Lax-Milgram lemma (see also \cite{GH}). It can be shown
that the two definitions coincide. This approach has the advantage
that it is applicable in cases where the potential is not a sum of
potentials with separate variables. Nevertheless, we may apply the
Lax-Milgram approach only under the following assumption on $V$:
\begin{equation}\label{ineqV} 
|\nabla V| \leq C \sqrt{1+ V(x)^2}\,.
\end{equation} 
We note that in the next section we consider potentials that are
neither separable nor satisfy (\ref{ineqV}). However, since in Section
\ref{s7} we will consider again a separable potential, preference was
given to the separation of variables technique over the other
approach.  
\end{remark}

\subsection{Determination of the spectrum}

Once $\A^\#$ has been defined in Definition \ref{Definition4.2} (we
occasionally write it as either $\A^\# (\vec J)$ or $\A^\#
(J,\theta,\vec v)$ with $\theta \in (-\pi, +\pi]$ to emphasize the
dependence on the parameters), we can focus our attention on its
spectrum $\sigma(\A^\# ((J,\theta,\vec v)) )$ and on its resolvent.
\begin{proposition} 
\label{propSpecA+} 
Let $\vec J:= (J^\prime,J_n)= J\, (\sin \theta \,
\vec{v},\cos\theta)\,$, where $\vec v\in \mathbb{S}^{n-2}$ and $J >
0$.  Then:
\begin{itemize} 
\item 
If $\theta\neq 0,\pi\,$, we have $ \sigma(\A^\#(J,\theta,\vec v) )=
\emptyset\,$.  
\item 
Otherwise, we have
\begin{equation} \label{eq:8}
\sigma(\A^\# (J,0,\vec v)) = \Union_{r\geq0}\{ \sigma(\mathcal L ^\# (J) )+r \}\,, 
\end{equation} 
and 
\begin{equation} \label{eq:8a} 
\sigma(\A^\# (J,\pi,\vec v)) = \overline{\sigma(\A^\# (J,0,\vec v))} \,. 
\end{equation}
\end{itemize} 
\end{proposition} 
\begin{proof} 
To prove \eqref{eq:8} we apply the partial Fourier transform $\mathcal
F$ defined by
\begin{equation} \label{eq:31}
( \F f) (\omega',x_n) = \int_{\R^{n-1}} e^{i\omega'\cdot
        x^\prime}f(x^\prime,x_n) \, dx^\prime \,, 
\end{equation} 
to $\A^\#$ to obtain (recall that $\theta=0$)
\begin{equation}\label{deffourier}
    \widehat\A^\# := \F\,\A^\#\, \F^{-1}  =
      \mathcal L _{x_n} ^\# + |\omega'|^2 \,.  
\end{equation} 
As $ \sigma(\A^\#)=\sigma\big(\widehat\A^\# \big)$ we readily obtain
\eqref{eq:8}.
\end{proof} 
Henry's lemma \cite[Lemma 2.8]{Hen2}
  can be extended to any realization $\#$ in the following way:
\begin{proposition}\label{lemAngle}~ 
\begin{enumerate} 
\item 
For every $\omega\in\mathbb{R}\,$, there exists $C_\omega>0$ such that
\begin{equation} \label{ResA+a} 
\sup_{\tiny{\Re z \leq \omega}}\|(\A^\#(J,\theta, \vec v) -z)^{-1}\| \leq C_\omega \,
        \exp \Big\{\frac{C_\omega}{J |\sin\theta|}\Big\}\,.  
\end{equation}
\item 
Let $0 < J_0 < J_1$ and $\kappa_0 >0$. There exists
$K(J_0,J_1,\varepsilon, \kappa_0) >0$ such that for all
$\varepsilon>0\,$, $\vec v\in \mathbb{S}^{n-2}\,$, $J_0\leq J \leq
J_1$, and $\kappa \in [0, \kappa_0]$
\begin{equation} \label{resAngle1}
        \sup_{\tiny{\begin{array}{c}\theta\in[-\pi,\pi]\\ \Re z\leq \lambda^\#(J \cos
              \theta,\kappa) - \varepsilon \end{array}}}\|(\A^\# (J,\theta,\vec
        v)- z)^{-1}\|\leq K(J_0,J_1,\varepsilon, \kappa_0) \,, 
\end{equation} 
where $\lambda^\#(J\cos\theta,\kappa)=\min \Re\sigma\big(\mathcal L
^\#(J\cos\theta,\kappa)\big) $ (when $\#\in \{D,N\}$ the dependence on
$\kappa$ should be omitted).
\end{enumerate} 
\end{proposition}
\begin{proof} 
Let $ \mathcal L _{x_n}^\# $ be given by \eqref{eq:9} and
$\A_{x^\prime}$ be given by \eqref{defCA_x'}.  When $\# \in \{R,T\}$
with $\#$-constant $\kappa$ we write $\mathcal L_{x_n}^{\#} (J\cos
\theta,\kappa)$.  As both $\A_{x^\prime}$ and $\mathcal L_{x_n}^{\#}(J
\cos \theta)$ are maximal accretive, we may write
\begin{equation} \label{OtimesAJthetav} 
  e^{-t\A^\#
      (J,\theta,\vec v)} = e^{-t\A_{x^\prime}(J,\theta,\vec v)}\otimes e^{-t\mathcal
      L_{x_n}^\# (J \cos \theta)}\,.  
\end{equation} 
Note that, for $\pm \theta\in(0,\pi)\,$, we may apply the rescaling
$x\mapsto (J\, |\sin\theta| )^{1/3}x\,$, to obtain
\begin{equation*}
e^{-t\A_{x^\prime}(J,\theta,\vec v)} = e^{-t( J
    |\sin\theta| )^{2/3}\A_{x^\prime}(1, \pm \pi/2,\vec v)}\,,
\end{equation*} 
Hence, by \eqref{eq:5},
\begin{equation} \label{eq:121} 
|| e^{-t\A_{x^\prime}(J,\theta,\vec v)} || \leq e^{- \frac{t^3}{12} ( J
      |\sin\theta| )^2}\,.  
\end{equation} 
From  (\ref{OtimesAJthetav}) and the fact that $ e^{-t\mathcal
L_{x_n}^{\#} (J \cos \theta)}$ is a contraction semi-group, we get
\begin{displaymath} \|e^{-t\mathcal
      A^\#(J,\theta,\vec v)}\|\leq e^{-\frac{1}{12}J ^2\sin^2\theta \, t^3}\,.
\end{displaymath} 
Suppose first that $\omega\geq1$. Using \eqref{eq:10}, this time for
$\A^\#$, i.e.,
\begin{equation} \label{formResSG}
 (\A^\# -z)^{-1} = \int_0^{+\infty}e^{-t(\A^\# -z)}dt\,, 
\end{equation} 
we obtain from \eqref{eq:64a} that, for all $\Re z \leq \omega\,$,
\begin{equation} 
\label{eq:11} 
\|(\A^\#
    -z)^{-1} \|\leq \int_0^{+\infty}e^{-\frac{1}{12}J^2\sin \theta^2\, t^{3}+\omega
      t}\,dt\leq \frac{\sqrt{2\pi}}{[J |\sin\theta|]^{1/2}\omega^{1/4}}\exp
    \biggl(\frac{2\omega^{3/2}}{J |\sin\theta|}\biggr) \,, 
\end{equation} 
which proves (\ref{ResA+a}) for $\omega\geq1$. For $\omega<1$ we write
\begin{displaymath} 
\|(\A^\# -z)^{-1} \|\leq
    \int_0^{+\infty}e^{-\frac{1}{12}J^2\sin \theta^2\, t^{3}+t}\,dt\leq
    \frac{\sqrt{2\pi}}{[J |\sin\theta|]^{1/2}}\exp \biggl(\frac{2}{J
      |\sin\theta|}\biggr) \,, 
\end{displaymath} 
which completes the proof of (\ref{ResA+a}). \\
Let $0<\theta_0<\pi/2$. Since (\ref{resAngle1}) follows immediately
from \eqref{ResA+a} whenever $\theta_0\leq|\theta|\leq\pi-\theta_0$,
we suppose first that $|\theta|<\theta_0$.  We then use
\eqref{OtimesAJthetav} and \eqref{eq:121} to obtain
\begin{equation*} 
\|e^{-t\mathcal A^{\#}(J,\theta,\vec v)}\| \leq
    e^{-\frac{t^3}{12}(J|\sin\theta|)^2}  
    \|e^{-t \LL_{x_n} ^\#(J\cos\theta,\kappa)}\|\leq\|e^{-t \LL_{x_n} ^\#(J\cos\theta,\kappa)}\|   \,.  
\end{equation*} 
  We may now use
either \eqref{eq:104} (for $\#\in\{N,R\}$) or
\eqref{eq:103} (when $\#=T$) to obtain that, for every $\epsilon>0$,
there exists $C(J_0, J_1, \theta_0,\kappa_0,\epsilon)$ such that
\begin{displaymath} 
\|e^{-t\mathcal A^\#(J,\theta,\vec v)}\| \leq
C(J_0, J_1, \theta_0,\kappa_0,\epsilon)e^{-t(\lambda^\#( J\cos\theta,\kappa)-\epsilon/2)}\,.  
\end{displaymath} 
We now use \eqref{formResSG} to obtain that 
\begin{displaymath} 
\|(\A^\#(J,\theta,\vec v) -z)^{-1} \|\leq
  C(J_0, J_1, \theta_0, \kappa_0,\epsilon)\int_0^\infty { e^{-t(\lambda^\#(J\cos \theta,\kappa)-\Re
    z-\epsilon/2)}\,}dt\,, 
\end{displaymath} 
from which \eqref{resAngle1} is easily verified for
$|\theta|\leq\theta_0$. For $\pi-\theta_0<|\theta|\leq\pi$ we use the
fact that
$\LL_{x_n}^\#(J\cos\theta,\kappa)=\overline{\LL_{x_n}^\#(J\cos(\pi-\theta),\kappa)}$
to complete the proof.
\end{proof}
\begin{remark}
\label{remunif} 
Note that $C_\omega$ in \eqref{ResA+a} is  independent of $\kappa $ (as long
as $\kappa$ is non-negative) for  $\#\in \{R,T\}$. Furthermore, let
$\theta_0\in(0,\pi/2)\,$ and $0 < J_0 <J_1$ . Then, (\ref{ResA+a})  implies
that for all $\omega\in\mathbb{R}\,$, $\kappa \geq 0$, and , there exists
$K_\omega'>0$ such that 
\begin{equation}\label{resAngle2}
     \sup_{
       \begin{subarray}{c}
         J_0\leq J\leq J_1\\
         \vec v\in\mathbb S^{n-2}
       \end{subarray}}   
\sup_{\tiny{\begin{array}{c}|\theta|\kern -2pt \in\kern -2pt [\theta_0\kern
      -1pt , \kern -1pt \pi\kern -2pt -\kern -2pt \theta_0] \\ \Re z\leq\omega
            \end{array}}}\|(\A^\# (J,\theta,\vec v)-z )^{-1}\|\leq
        K_\omega' \,.  
\end{equation} 
\end{remark}

For the sake of the upper bound derived in Section \ref{s7}, we now
provide an extension of \eqref{resAngle1} to the case where $\Re z$ is
slightly larger than $\lambda^\#$.  To this end we use the improved
estimates of $||e^{- t \mathcal L_{x_n}^\#}||$ provided in
Propositions \ref{lemmaunifter} and \ref{lemmaunifbisTimp}.
\begin{proposition}~\\ 
\label{prop:semigroup-upper}
Let $\theta_0\in(0,\pi/2)$, $0 < J_0 < J_1$, and $\kappa_0 >0$ (for
$\#\in\{R,T\}$). Then, there exists $C(\kappa_0,J_0,J_1,\theta_0)>0$
such that, for all $J \in [J_0,J_1]$, $\kappa \in [0,\kappa_0]$ (for
$\#\in\{R,T\}$), $t>0$ and $\theta\in [-\theta_0,\theta_0]\cup [\pi
-\theta_0,\pi + \theta_0]$\,, we have
\begin{equation} \label{decSGA+a} 
  \,~\|e^{-t\A^\# (J,\theta,\vec v)}\|\leq C(J_0,J_1,\theta_0,\kappa_0)\,
      e^{-\frac{1}{12} J^2\sin^2\theta \, t^3 - t \lambda^\#(J\cos\theta,\kappa)}\,.
\end{equation} 
Furthermore, we have that 
\begin{equation} 
\label{eq:28} 
\sup_{\Re z\leq\lambda^\#(J\cos \theta ) +[J\, |\sin \theta|]^{2/3} } 
\|(\mathcal A^\#(J,\theta,\vec v) -z)^{-1} \|\leq \frac{C(J_0,J_1,\theta_0,\kappa_0) }{[\,J\, |\sin \theta|\, ]^{2/3}} \,.
\end{equation} 
\end{proposition} 
\begin{proof} 
From  \eqref{eq:68} and \eqref{eq:3.8} it follows that for any
$\theta\in[-\theta_0,\theta_0]$ and $J_0\leq J\leq J_1$, there exists
$C_1(\kappa_0,J_0,J_1,\theta_0) >0$ such that
\begin{equation} 
\|e^{-t\mathcal
        L _{x_n}^\# (J\cos \theta)} \| \leq C_1(J_0,J_1,\theta_0,\kappa_0) \, e^{-t\lambda^\# (J\cos \theta,\kappa)}\,.  
\end{equation} 

As
\begin{displaymath} 
\|e^{-t\A^\#(J,\theta,\vec v)}\| \leq \| e^{-t\mathcal
        A_{x'}}\|\,\| e^{-t \mathcal L _{x_n} ^\#}\| 
\end{displaymath} 
we readily obtain \eqref{decSGA+a}. \\
Combining \eqref{formResSG} and \eqref{decSGA+a} we obtain that
\begin{equation} \label{eq:69}
      \|(\A^\#(J,\theta,\vec v) -z)^{-1} \|\leq C
      \int_0^{+\infty} e^{-\frac{1}{12} J^2\sin
        \theta^2\, t^{3}+(\Re z-\lambda^\# (J\cos \theta,\kappa))\, t}\,dt\,. 
\end{equation} 
To obtain \eqref{eq:28} we use \eqref{eq:69} in conjunction with
\eqref{eq:64a}.
\end{proof}

We conclude this section with the following straightforward estimate
which will become useful in Section \ref{s7}.
\begin{lemma} 
Let $x=(x^\prime,x_n)\in\R^{n-1}\times\R^\#$, and 
\begin{displaymath}
   \A^\#  = -\Delta + ix_n \,, 
\end{displaymath}
be defined on $D(\A^\#)$ given by  \eqref{carac}.  
For any $\mu \in \mathbb R$, there exist positive $C(\mu)$ such that,
for any $\lambda=\mu+i\nu$ with with $|\nu|>\mu+4\,$,
\begin{equation}
  \label{eq:30}
\|(\A^\# -\lambda)^{-1}\| \leq  C(\mu) \,.
\end{equation}
\end{lemma}
\begin{proof}
Applying a partial Fourier transform \eqref{eq:31} in the $x^\prime$
direction yields (see also \eqref{deffourier}) 
\begin{displaymath}
  \widehat\A^\# = \int^{\oplus}  \widehat\A^\#(\omega') d \omega'   \,,
\end{displaymath}
where, for $\omega'\in \mathbb R^{n-1}$,
\begin{equation*}
  \widehat\A^\#(\omega') := \mathcal L_{x_n}^\# + |\omega'|^2\,,
\end{equation*}
is considered as an unbounded one variable operator on $L^2_\#$.\\
We may now use the same technique as in \cite{Hel2} to establish that
\begin{displaymath}
   \| (\widehat\A^\#(\omega')  -\lambda)^{-1}\|\leq C(\mu_\omega')\,,
\end{displaymath}
where $\mu_{\omega'}=\mu-|\omega'|^2$.\\
Since for sufficiently large $|\omega'|$ we have (see \cite{Hel2})
\begin{displaymath}
  \| (\widehat\A^\#(\omega')   -\lambda)^{-1}\|\leq \frac{C(\mu)}{|\omega'|^2}\,.
\end{displaymath}
we easily obtain \eqref{eq:30}.
\end{proof}

\section{Limit problems: Quadratic potential} \label{s5}

\subsection{The quadratic model}
In this section, we consider the $\#$ realization of the operator 
\begin{equation}    \label{eq:62} 
\PP^\# = -\Delta +i\,\biggl(\vec J^\cdot x +\sum_{k=1}^{n-1}\alpha_j x_j^2\biggr)\,, 
\end{equation}
acting in $\mathbb R^n_\#$.  When $\#\in \{ R,T\}$, there is the
additional Robin or Transmission parameter $\kappa\geq 0$ which is not
always mentioned in the notation. In the above, the $\alpha_j\neq0$
for all $1\leq j\leq n-1$ and $J_n\neq 0$.  By applying an appropriate
translation in the $x^\prime= (x_1,\ldots,x_{n-1})$ direction and a
shift of the spectrum, which does not modify its real part, we obtain
the case $J'=0$ and consider from now on the reduced form
\begin{equation}    \label{eq:62a} 
\PP^\# = \PP^\# (\vec \alpha)= -\Delta +i\, \biggl(\alpha_n x_n
    +\sum_{k=1}^{n-1}\alpha_j x_j^2 \biggr)\,, 
\end{equation}
with $\alpha_n:=J_n$.\\
Setting 
\begin{equation}   \label{eq:78}
V (x) = \alpha_n x_n +\sum_{k=1}^{n-1}\alpha_jx_j^2 \,, 
\end{equation} 
we shall also use the notation 
\begin{equation*}
  \PP^\# = \mathcal P_{ V}^\#\,.  
\end{equation*}
We adopt here an approach which can be applied to a much wider class
of operators. To this end it proves useful to consider first the
problem in $\R^n$, and only then to introduce the effect of the
boundary.

\subsection{The entire space problem}

The problem we address here has already been treated for the
selfadjoint case in \cite{HM}, for the polynomial case in \cite{HNo},
for the case of Fokker-Planck operators in \cite{HN}, and for the
complex Schr\"odinger magnetic operator in \cite{AH}. The complex harmonic oscillator was first analyzed in \cite{Dav} (see  also \cite{H}). We consider
here a different class of operators, which includes the operator
\eqref{eq:62} acting on a dense set in $L^2(\R^n)$.  More precisely,
our goal is to establish compactness of the resolvent and to provide a
transparent description, when possible, for the domain of operators of
the type $\PP_V:=-\Delta + i V$. Here $\PP_V$ is defined as the
closure of $P_V/{C_0^\infty(\mathbb R^n)}$. \\
We note first that $D (\PP_V) \subset D ((\PP_{-V})^*)$, since
$(\PP_{-V})^*$ is a closed extension of $\PP_V$. Moreover by
\cite[Exercise 13.7]{H} $\PP_V$ and $\PP_{-V}$ are maximal
accretive. It follows immediately that $(\PP_{-V})^* +1$ is injective
and $(\PP_{V}) +1$ is surjective.  This implies $ D ((\PP_{-V})^*)
\subset D (\PP_V) $. If indeed, $u\in D ((\PP_{-V})^*)$, there exists,
by the surjectivity of $(\PP_{V}) +1$, $v \in D(\PP_V)$ such that
$((\PP_{-V})^* +1) u = (\PP_V+ 1) v = ((\PP_{-V})^* +1) v$.  We
conclude then that $u=v$ by the injectivity of $(\PP_{-V})^* +1$, so
$u\in D (\PP_V)$.  Hence, we have proved 
\begin{equation} 
  D (\PP_V):=  \{ u\in L^2 (\mathbb R^n)\,|\, (-\Delta + i V) u \in L^2(\mathbb R^n)\}\,.  
\end{equation}
 
Following \cite{HM}, we now introduce a rather general class of
potentials extending polynomials of degree $r$.
\begin{definition}
\label{def:Tr}
For $r\in \mathbb N$, we say that $V\in\Tg_r$ if
\begin{enumerate}
   \item $V\in C^{r+1}(\mathbb R^n,\mathbb R)$
   \item There exists $C_0$ such that for all
$x\in\mathbb R^n$
\begin{equation}  \label{A.17}
\max_{|\beta|=r+1} |D_x^\beta V(x)| \leq C_0\, m(V,r,x)\,,
\end{equation}
where
\begin{equation}  \label{A.15}
m:= m(V,r,x)= \sqrt{ \sum _{|\alpha| \leq r}  |D_x^\alpha V(x)|^2 +1}\, .
\end{equation}
\end{enumerate}
\end{definition}
In particular, we have
\begin{example}
\label{ExampleA1}~ 
\begin{enumerate}
\item 
The potential $V = J\cdot x$ is of class $\Tg_0$.
\item 
The potential $V$ defined by
\begin{equation*}
V(x',x_n):=  \alpha_n x_n + \sum_{j=1}^{n-1} \alpha_j x_j^2 \,,
\end{equation*}
with $\alpha_j \neq 0$ ($j=1,\cdots,n$), is of class $\Tg_1$.
\end{enumerate}
\end{example}

Note also that, for the case $r=0$, \eqref{A.17} reduces to
\eqref{ineqV} which is precisely the type of potentials considered in
\cite{AH} in the absence of magnetic field. \\
 
The following auxiliary lemma is an adaptation of a similar result in
\cite{HM} to our needs. 
\begin{lemma}
\label{LemmaA4}~\\
Let $T\in C^2(\R^n)$. Then, for $k=1\,\,\dots,n$, $u\in
C_0^\infty(\R^n)$ and $0\leq s\leq1$\,,
\begin{align}
\label{eq:86}
&\| m^{- 1+ \frac s2} (\partial_{x_k}T) u \|^2 = - \langle m^{-1+s} Tu \;|\; m^{-1}(\partial_{x_k}T)\partial_{x_k} u\rangle 
 -\notag\\  &\langle \partial_{x_k} u \;|\; m^{-1} (\partial_{x_k}T) m^{-1+s} Tu \rangle
\kern -3pt - \kern -3pt \langle m^{-1+s} T u \;| \; m^{1-s}\partial_{x_k}(m^{-2 +s}\partial_{x_k}T)u \rangle \,,
\end{align} 
where $m$ is given by (\ref{A.15}).
\end{lemma}
\begin{proof}
Let $U\in C^1$. To make the following integrations by parts more
transparent to the reader, we represent the multiplication operator by
$\partial_{x_k}U$, for each $1\leq k\leq n\,$, as the bracket
$[X_k,U]$ where $X_k=\partial_{x_k}$.  We then write:
\begin{equation*}
\begin{array}{ll}
\| m^{- 1+ \frac s2} [X_k,T] u \|^2&=
 \langle  m^{-1+s} (X_k T -T X_k)u \;|\;  m^{-1} [X_k,T] u \rangle \\
&=
 \langle  m^{-1+s} X_k T u \;|\;  m^{-1} [X_k,T] u \rangle \\
&\quad  - 
 \langle  m^{-1+s} T X_k u \;|\;  m^{-1} [X_k,T] u \rangle \\
& 
=  - \langle m^{-1+s} Tu \;|\; m^{-1}[X_k,T] X_k u\rangle \\
  &\quad - \langle X_k u \;|\; m^{-1} [X_k,T] m^{-1+s} Tu \rangle \\
& \quad - \langle  T u \;|\; [X_k, m^{-2 +s} [X_k,T]] u \rangle  \,,
\end{array}
\end{equation*}
from which \eqref{eq:86} easily follows.
\end{proof}

The following weighted estimate is useful when proving compactness of
the resolvent $(\PP_V-\lambda)^{-1}$ and for describing $D(\PP_V)$
when $r=1$.
\begin{proposition}
\label{prop:improve}
Let $V$ be such that \eqref{A.17} is satisfied for some $r\geq1$. Then
\begin{equation}
    \label{eq:108}
\|m^{\frac{2}{2^{r+1}-1}} u||^2 +
||m^{-2\frac{2^{r-1}-1}{2^{r+1}-1}}Vu||^2  \leq C \left( ||P_V u||^2 +
  ||u||^2 \right)\,,\, \forall u \in C_0^\infty(\mathbb R^n)\,. 
\end{equation}
\end{proposition}
\begin{proof}~\\
{\bf Step 1:} For $\beta \geq 0$, we prove that for every $\epsilon>0$
there exists $C_\epsilon>0$ such that
\begin{equation}
\label{eq:110}
\|m^{\beta/2}\nabla u\|_2^2 \leq C_\epsilon\, \left(\|P_Vu \|_2^2+\|u\|_2^2\right)+ \epsilon^2 \|m^\beta u\|_2^2\,,
\end{equation}
for all $u\in C_0^\infty(\R^n)$. 

To prove \eqref{eq:110} we first observe that
\begin{equation}
  \Re\langle m^\beta u, P_V u\rangle = \|m^{\beta/2}\nabla u\|_2^2 +  \Re\langle\nabla(m^\beta)u,\nabla u\rangle \,. 
\end{equation}
It then follows, using Cauchy-Schwarz's inequality,  that
\begin{equation}
\label{eq:111} 
  \|m^{\beta/2}\nabla u\|_2^2  \leq \frac{2}{\epsilon^2}\| P_V  u \|^2_2 + \frac{\epsilon^2}{2}||m^\beta u||^2_2+
  \|m^{-\beta/2}\nabla(m^\beta)u\|_2^2\,.
\end{equation}
We now use the fact that by Assumption \eqref{A.17} we have  
\begin{equation*}
m^{-\beta/2}|\nabla(m^\beta)|\leq C \, m^{\beta/2}\,,
\end{equation*}
to obtain
\begin{displaymath}
  \|m^{-\beta/2}\nabla(m^\beta)u\|_2^2 \leq  \frac{C^2}{2\epsilon^2}\| u \|^2_2 + \frac{\epsilon^2}{2}||m^\beta u||^2_2\,,
\end{displaymath}
which, combined with \eqref{eq:111}, yields \eqref{eq:110}.\\
 
{\bf Step 2:} We now prove that, for $0\leq \beta \leq 2/3\,$, there
exists $C>0$ such that for every $\epsilon>0$ we have, for some
$C_\epsilon>0$,
\begin{equation}
\label{eq:112}
 \|m^{(3\beta-2)/4}Vu\|_2^2\leq C_\epsilon \, (  \| P_V u \|^2_2 + ||u||^2_2)+ C \, \epsilon^2 \, \|m^\beta u\|_2^2\,,
\end{equation}
for all $u\in C_0^\infty(\R^n)$. \\
Let $0\leq \alpha\leq 1$. An integration by parts yields
\begin{displaymath}
  \Im\langle m^{-\alpha}Vu,P_Vu\rangle= \langle m^{-\alpha} V^2u,u\rangle + \Im\langle\nabla(m^{-\alpha}V)u,\nabla u\rangle\,.
\end{displaymath}
We can then conclude that
\begin{displaymath}
  \|m^{-\alpha/2} Vu\|_2^2 \leq \|m^{-\alpha}Vu\|_2\|P_Vu\|_2+
  \|m^{-(1-\alpha)}\nabla(m^{-\alpha}V)\|_\infty \|m^\beta u\|_2\|m^{1-\alpha-\beta}\nabla u\|_2\,.
\end{displaymath}
Since $m^{-\alpha}\leq m^{-\alpha/2}$ and since by \eqref{A.17}, $|
m^{-(1-\alpha)}\nabla(m^{-\alpha}V)|$ is bounded we obtain
\begin{displaymath} 
    \|m^{-\alpha/2} Vu\|_2^2 \leq   \|P_Vu\|_2^2+ 2 C\, \|m^\beta u\|_2\, \|m^{1-\alpha-\beta}\nabla u\|_2\,.
\end{displaymath}
Setting $\alpha=1-3\beta/2$, we obtain \eqref{eq:112} from \eqref{eq:110}. \\

{\bf Step 3:} For $\beta \geq 0$ and $\sigma\leq0$, we prove that
there exists $C>0$ such that for every $\epsilon>0$ there exists
$C_\epsilon>0$ such that,
\begin{equation}
\label{eq:113}
 \|m^{\sigma/2-(2-\beta)/4}(\partial_{x_k} T)u\|_2^2\leq \|m^\sigma Tu\|_2\, \big[C_\epsilon(  \|P_Vu\|_2 + ||u||_2) + C\, \epsilon \,   \|m^\beta u\|_2\big]\,,
\end{equation}
for all $u\in C_0^\infty(\R^n)$, and $T\in C^2(\R^n)$ satisfying
\begin{equation}
\label{eq:67}
   \sup_{1\leq k,j\leq N} \{|\partial_{x_k} T|+|\partial_{x_j}\partial_{x_k} T| \} \leq C\, m\,,
\end{equation}
for some positive $C>0\,$.

We begin by rewriting \eqref{eq:86} in the form:
\begin{multline*}
  || m^{- 1+ \frac s2} (\partial_{x_k} T)u ||^2 =- \langle m^{-1+s} Tu \;|\;
  m^{-1}(\partial_{x_k} T)\partial_{x_k}  u\rangle \\
 - \langle \partial_{x_k}  u \;|\; m^{-1} (\partial_{x_k} T) m^{-1+s} Tu \rangle 
+ \langle m^{-1+s} Tu \;|\; m^{1-s}\partial_{x_k} (m^{-2 +s}
 (\partial_{x_k} T)) \,  u \rangle\;. 
\end{multline*}
 
For the first and second terms on the right-hand-side, which are
complex conjugate,  we have by \eqref{eq:110} and \eqref{A.17} 
that 
\begin{multline*}
 |  \langle m^{-1+s} Tu \;|\;
  m^{-1}(\partial_{x_k} T)\partial_{x_k}  u\rangle \,| 
   \leq \hat C \, \|m^{\beta/2}\partial_{x_k}  u\|_2\,\|m^{-1+s-\beta/2} Tu
  \|_2 \,  \\   \leq \hat C\, \|m^{-1+s-\beta/2} Tu\|_2\,  [C_\epsilon(\|P_Vu \|_2 + ||u||_2)+ C\, \epsilon \,  \|m^\beta u\|_2]\,.
\end{multline*}
Finally, for the last term, we have 
\begin{multline*}
  |\langle m^{-1+s} Tu \;|\; m^{1-s}\partial_{x_k} (m^{-2 +s} (\partial_{x_k} T)) u \rangle|\leq C\, \|m^{\beta/2} u\|_2\,\|m^{-1+s-\beta/2} Tu
  \|_2 \\\leq C\,  \|m^{-1+s-\beta/2} Tu\|_2 \, (C_\epsilon\|u \|_2 + \epsilon \, \|m^\beta u\|_2)\,.
\end{multline*}
Consequently,
\begin{displaymath}
   \| m^{- 1+ \frac s2} (\partial_{x_k} T)u \|^2_2 \leq C\,  \|m^{-1+s-\beta/2} Tu\|_2\, [C_\epsilon(\|P_Vu \|_2 + ||u||_2)+ \epsilon \|m^\beta u\|_2]\,.
\end{displaymath}
We now choose $s=\sigma+1+\beta/2$  to obtain \eqref{eq:113}.\\
 
{\bf Step 4:} We prove that for all $u\in C_0^\infty(\R^n)$, $1\leq
k\leq r$, $0\leq \beta\leq1$, and $\epsilon>0\,$, there exists
$C_\epsilon>0$ such that
\begin{multline}
  \label{eq:109}
\|m^{\beta/2-1+(\beta/2+1)2^{-k-1}}(\partial_x^kV)u\|_2^2\leq\\
\leq \|m^{(3\beta-2)/4}Vu\|_2^{2^{-k+1}}\big[C_\epsilon(\|P_Vu \|_2 + ||u||_2)+ \epsilon^{1/2}
\|m^\beta u\|_2\big]^{\frac{2^k-1}{2^{k-1}}}\,. 
\end{multline}
Applying \eqref{eq:113} recursively $1\leq k\leq r$ times yields
(using the fact that $T=\partial_x^jV$ satisfies \eqref{eq:67} for all
$0\leq j\leq r-1$)
\begin{equation}
\label{eq:114}
  \|m^{\sigma_k}(\partial_x^rV)u\|_2^2\leq
\|m^{\sigma_0}Vu\|_2^{2^{-k+1}}[C_\epsilon(\|P_Vu \|_2 + ||u||_2)+ \epsilon^{1/2}
\|m^\beta u\|_2]^{\frac{2^k-1}{2^{k-1}}}\,,
\end{equation}
where
\begin{equation}
\label{eq:115}
  \sigma_k = \frac{1}{2}\sigma_{k-1} + \frac{\beta}{4} -  \frac{1}{2} \,.
\end{equation}
The solution for the above recurrence relation is given by
\begin{displaymath}
  \sigma_k= 2^{-k}\sigma_0 + \Big(\frac{\beta}{2}-1\Big)(1-2^{-k}) \,.
\end{displaymath}
Setting $\sigma_0=(3\beta-2)/4$ in the above and in \eqref{eq:114} yields
\eqref{eq:109}.

{\bf Step 5:} We finally prove \eqref{eq:108}. \\
Let $\epsilon>0$ and $u\in C_0^\infty(\R^n)$. As
$\{\sigma_k\}_{k=0}^r$ is monotone decreasing for $\beta \geq -2$, we
obtain from \eqref{eq:109}, \eqref{eq:112} and \eqref{eq:113}, for
$T=V$, that
\begin{displaymath}
  \Big\|m^{\beta/2-1+(\beta/2+1)2^{-(r+1)}}\sum_{|\gamma|=0}^r|(D^\gamma V) u|\Big\|_2^2\leq C_\epsilon\, (  \|P_V u\|^2_2 + ||u||^2_2) + C\epsilon  \|m^\beta u\|_2^2\,.
\end{displaymath}
The above, with the aid of \eqref{A.17}, yields
\begin{displaymath}
    \|m^{\beta/2+(\beta/2+1)2^{-(r+1)}}u\|_2^2\leq C_\epsilon\, (  \| P_V u \|^2_2 + \|u\|^2_2)+C \, \epsilon^2 \|m^\beta u\|_2^2\,.
\end{displaymath}
Choosing $\beta=[2^r-1/2]^{-1}$ and $\epsilon$ which is sufficiently
small in the above and in \eqref{eq:112} yields \eqref{eq:108}.
\end{proof}

\begin{corollary}
\label{cor:entirespacer=1}
Suppose that
\begin{displaymath}
    \lim_{|x|\to + \infty}m(V,r,x)=+\infty \,. 
\end{displaymath}
Then, the resolvent of $\PP_V$ is compact.
\end{corollary}

\begin{remark}
For the case $r=1$, \eqref{eq:108} and \eqref{eq:110} together yield
\begin{equation}
    \label{eq:82}
\|m^{2/3}u\|_2^2 + \|m^{1/3}\nabla u\|_2^2+\|Vu\|_2^2 \leq C\, (  \| P_V  u \|^2_2 + \|u\|^2_2)\,.
\end{equation}
Consequently we have 
\begin{displaymath} 
  \|u\|_{2,2}^2 \leq C(\|u\|_2^2+ \|\Delta u\|_2^2) \leq C\, (\| P_V  u \|^2_2 + \|u\|^2_2)\,.
\end{displaymath}
We conclude from the above that
\begin{equation}
\label{cardom}
    D(P_V)=\{u\in H^2(\R^n) \, | \, Vu\in L^2(\R^n)\} \,.
\end{equation}
Note that we do not use in the proof the assumption that $m (V,r,x)
\ar +\infty$ as $|x| \ar +\infty$\,.
\end{remark}

\subsection{Half space problems}

We consider in the following the operator \eqref{eq:62} acting on
$\R^n_\#$. The boundary condition satisfied on $\partial\R^n_\#$ are
given by \eqref{eq:102} and depend on the value of $\#$ in
$\{D,N,R,T\}$ and on the additional parameter $\kappa \geq 0$ when
$\#\in \{R,T\}$.  We begin by defining $\PP^\#$, given by
\eqref{eq:62}, using the separation of variables technique presented
in the previous section.  Recall that the ``generalized Lax-Milgram''
approach, mentioned in Remark \ref{RmLM}, which relies on
\eqref{ineqV} (a condition which is {\em not} satisfied when all the
$\alpha_j$'s do not have the same sign in \eqref{eq:62}) is
inapplicable in this case.

In view of the foregoing discussion we now state 
\begin{proposition} 
The domain of $\PP^\#$ is given by \eqref{carac}, i.e.,
\begin{multline}
\label{eq:73} \
 D (\PP^\#):= \{u\in L^2_\# : \exists (u_j)_{j\geq1}\subset\mathcal C^\#\,,~
  u_j\underset{\tiny{j\to+\infty}}{\overset{L^2_\#}{\longrightarrow}}u\,,\\
   (\PP^\#u_j)_{j\geq1}~\textrm{is
    a Cauchy sequence in } L^2_\#\}\,,
\end{multline}
where $\mathcal C^\#$ is the core for $\mathcal P^\#$. \\
Moreover, we have
\begin{equation}  
D (\PP^\#)\subset H^1_\#\,,
\end{equation}
and  the $\#$-condition is satisfied for $u$ in the domain of $\PP^\#$.
\end{proposition} 
\begin{proof}
The proof of \eqref{eq:73} is identical with that of
\eqref{carac}.  It can also be easily verified that 
\begin{equation}\label{controlH1} 
\sum_j \|D_{x_j}  u\|^2 \leq \frac 12 (\|\mathcal P^\# u\|^2 + \|u\|^2)\,,\, \forall u \in
    \mathcal C^\#\,, 
\end{equation} 
and hence $D (\PP^\#)\subset H^1_\#$. Similarly to the proof of
\eqref{carac}, we can show that either the boundary or transmission
condition are well defined and satisfied for $u$ in the domain of
$\PP^\#$.
\end{proof}

We now obtain an estimate, similar to \eqref{eq:108}, in the presence
of boundaries. We derive it for operators of the type $\PP^\#$ as
defined in \eqref{eq:62a} and \eqref{eq:73} and let $V$ be defined,
through the remainder of this section, by \eqref{eq:78} (note that
$V\in\Tg_1$).  We begin by introducing a partition of unity on
$\overline{\mathbb R_+}$ corresponding to the normal variable $x_n$
\begin{equation}   \label{eq:80}
  1 = \phi_1^2+\phi_2^2\,,\, \phi_1 =1 \mbox{ on } [0,1]\,,\, \supp \phi_1
\in [0,2]\,. 
\end{equation}
The support of $ x \mapsto \phi_1(x_n) u (x)$ belongs to $\mathcal
B^n_+:= \mathbb R^{n-1}\times [0,2]$.  We then define 
\begin{equation}   \label{eq:116}
m_T(V,r,x',x_n) =\sqrt{1+ \sum_{|\alpha'| \leq 1}
        |\partial_{x'}^{\alpha'} V (x',x_n)|^2} 
\end{equation} 
which satisfies in $\mathcal B^n_+$
\begin{equation} 
|\nabla V (x) | \leq C\, m_T \,.
\end{equation}
We note the explicit expression for $m_T$  
\begin{equation*}
m_T (x',x_n) = \sqrt{ 1 + V (x',x_n)^2 + |x'|^2}\,.
\end{equation*}
We can now state   
\begin{proposition}
Then for any $u\in\CC^\#$ we have
\begin{equation}    \label{eq:117}
\|m_T^{2/3}\phi_1u\|_2^2 + \|m^{2/3}\phi_2u\|_2^2 + \|Vu\|_2^2 \leq C\, (  \| \PP^\#u \|^2_2 + \|u\|^2_2)\,.
\end{equation}
\end{proposition}
\begin{proof}
The proof is an adaptation of Proposition \ref{prop:improve} to the
$\#$ realization of $\PP_V$ in the case $r=1$. In view of the
similarities we provide only its outlines.

It can be easily verified that (recall that $\kappa\geq0$) for any
$v\in \mathcal S (\overline{\mathbb R^n_+})\cap \mathfrak D^\#_n$ with
support in $\overline{\mathcal B^n_+}$
\begin{displaymath}
   \Re\langle m_T^{2/3} v, \PP^\#v\rangle \geq  \|m_T^{1/3}\nabla v\|_2^2 +  \Re\langle\nabla(m_T^{2/3})v,\nabla v\rangle \,. 
\end{displaymath}
Hence, we obtain that for every $\epsilon>0$ there exists
$C_\epsilon>0$ such that
\begin{equation}   \label{eq:119}
  \|m_T^{1/3}\nabla u\|_2^2 \leq C_\epsilon\, \left(\|\PP^\#u \|_2^2+\|u\|_2^2\right)+ \epsilon^2 \|m_T^{2/3} u\|_2^2\,,
\end{equation}
Next we observe that
\begin{displaymath}
  \Im\langle Vv,\PP^\#v\rangle= \langle V^2v,v\rangle + \Im\langle v\nabla V,\nabla v\rangle\,,
\end{displaymath}
which together with \eqref{eq:119} leads to 
\begin{equation}
\label{eq:118}
 \|Vv\|_2^2\leq C_\epsilon \, (  \| \PP^\#v\|^2_2 + \|v\|^2_2)+ C \, \epsilon \, \|m_T^{1/3}v\|_2^2\,.
\end{equation}

We next repeat the argument leading to \eqref{eq:113}, and as the
integration by parts does not involve any boundary terms we obtain in
the same manner
\begin{displaymath}
 \|m_T^{-1/3}(\partial_{x_k} V)u\|_2^2\leq C_\epsilon(  \|\PP^\#  u\|^2_2 + ||u||^2_2) + C\, \epsilon\,   \|m_T^{2/3} u\|_2^2\,.
\end{displaymath}
We can now combine the above with \eqref{eq:118} to obtain
\begin{displaymath}
  \|m_T^{2/3}v\|_2^2  + \|Vv\|_2^2 \leq C\, (  \| \PP^\#v \|^2_2 + \|v\|^2_2)\,.
\end{displaymath}
Since $\phi_1u\in \mathcal S(\overline{\mathbb R^n_+})\cap{\mathfrak
D}^\#_n$ with support in $\overline{\mathcal B^n_+}$, and
$\phi_2u\in\Sg(\R^n)$ we may use the above and \eqref{eq:82} to obtain
that
\begin{equation}   \label{eq:120}
  \|m_T^{2/3}\phi_1u\|_2^2 + \|m^{2/3}\phi_2u\|_2^2 + \|Vu\|_2^2 \leq C\, (  \|
  \PP^\#(\phi_1u) \|^2_2 + \|\PP^\#(\phi_2u) \|^2_2+ \|u\|^2_2)\,.
\end{equation}
It can be easily verified, using \eqref{controlH1}, that
\begin{displaymath}
   \|\PP^\#(\phi_1u) \|^2_2 + \|\PP^\#(\phi_2u) \|^2_2\leq C\, (\|\PP^\#(\phi_2u) \|^2_2+\|u\|_2^2)\,,
\end{displaymath}
which together with \eqref{eq:120} yields \eqref{eq:117}.
\end{proof}
As a corollary, we get
\begin{corollary}
\label{lem:truncation}
Let $V$ be given by \eqref{eq:78}. Then, the domain of $\PP^\#$ is
\begin{displaymath}
   D(\PP^\#) = \hat{\Dg}_\PP\overset{def}{=}\{u\in H^2(\R^n_\#)\cap \mathfrak
   D^\#_n \,|\,Vu\in L^2( \R^n_\#)\}\,. 
\end{displaymath}
\end{corollary}
\begin{proof}
The inclusion $D(\PP^\#)\subseteq\hat{\Dg}_\PP$ follows immediately
from \eqref{eq:120}. To prove the other direction we may use the same
argument as in the conclusion of the proof of Proposition
\ref{Proposition4.4}.
\end{proof}

Since $m_T (x)$ tends to $+\infty$ as $|x|\ar +\infty$, we get another
corollary:
\begin{corollary} 
The resolvent of $\PP^\#$ given by \eqref{eq:62} is compact.
\end{corollary}

\section{A lower bound}
\label{s6}

In this section we prove \eqref{limSpect1N} for either the Robin
boundary condition $\#=R$ or the transmission problem $\# = T\flat$
with $\flat = D$ or $N$.  We walk along the same steps used in
\cite{Hen2} to obtain the same lower bound as for the Dirichlet
boundary condition. The Dirichlet and Neumann cases can appear as
particular cases of the Robin case for $\kappa=0$ and
$\kappa=+\infty$.

For some $1/3<\varrho<2/3$ and for every $h\in(0,h_0]$, we choose two
sets of indices $\Jg_{i}(h)\,$, $ \Jg_{\partial}(h) \,$, and a set of
points (for a transmission problem we have $\Omega^\# = \Omega^T:=
\Omega_-\cup\Omega_+$ and
$\partial\Omega^\#=\partial\Omega_-\cup\partial\Omega$)
\begin{subequations}   \label{eq:25}
\begin{equation}
\big\{a_j(h)\in \Omega^\# : j\in \mathcal J_i(h)\big\}\cup\big\{b_k(h)\in\pa {\Omega^\#} : k\in \mathcal J_\partial (h)\big\}\,,
\end{equation}
such that $B(a_j(h),h^\varrho)\subset\,{\Omega^\#} $\,,
\begin{equation}
    \bar\Omega\subset\bigcup_{j\in \Jg_{i}(h)}B(a_j(h),h^{\varrho})~\cup\bigcup_{k\in \Jg_{\partial }(h)}B(b_k(h),h^{\varrho})\,,
\end{equation}
and such that the closed balls $\bar B(a_j(h),h^{\varrho}/2)\,$, $\bar
B(b_k(h),h^{\varrho}/2)$ are all disjoint. \\
Note that $ \sharp \mathcal J_i(h)\propto h^{-n{\varrho}}$ and $
\sharp \mathcal J_\partial (h)\propto h^{-(n-1){\varrho}}\,.$\\
Now we construct in $\mathbb R^n$ two families of functions 
\begin{equation}\label{part1}
 (\chi_{j,h})_{j\in \mathcal J_i(h)} \mbox{ and } (\zeta_{j,h})_{j\in \mathcal J_\partial (h)}\,,
\end{equation}
such that, for every $x\in\bar\Omega\,$,
\begin{equation}
\sum_{j\in \mathcal J_i(h)}\chi_{j,h}(x)^2+\sum_{k\in \mathcal J_\partial (h)}\zeta_{k,h}(x)^2=1\,,
\end{equation}
\end{subequations}
and such that $\Supp \chi_{j,h}\subset B(a_j(h),h^{\varrho})$ for
$j\in \mathcal J_i(h)$, $\Supp\zeta_{j,h}\subset
B(b_j(h),h^{\varrho})$ for $j\in \Jg_\partial\,$, and
$\chi_{j,h}\equiv 1$ (respectively $\zeta_{j,h}\equiv1$) on $\bar
B(a_j(h),h^{\varrho}/2)$ (respectively $\bar
B(b_j(h),h^{\varrho}/2)$)\,.\\

To verify that the approximate resolvent constructed in the sequel
(see \eqref{defResApp1}) satisfies the boundary conditions on
$\partial\Omega^\#$, we require in addition that
\begin{equation}
 \label{eq:90}
\frac{\partial\zeta_{k,h}}{\partial\nu}\Big|_{\partial\Omega^\#}=0 
\end{equation}
for $\#\in \{ N, R, T\flat\}$.  Note that, for all
$\alpha\in\mathbb{N}^n\,$, we can assume that there exist positive
$h_0$ and $C_\alpha$, such that, $\forall h \in (0,h_0]$, $ \forall
x\in \overline{\Omega}$,
\begin{equation}
\label{supPaCutoff} 
\sum_j |\pa^\alpha\chi_{j,h}(x)|^2\leq C_\alpha \, h^{- 2 |\alpha|{\varrho}} ~~~\textrm{ and }~~~
\sum_j |\pa^\alpha\zeta_{j,h}(x) |^2 \leq C_\alpha  \, h^{-2 |\alpha|{\varrho}}\,.
\end{equation}
We introduce also $\eta_{j,h} = 1_{\Omega^\#}\, \zeta_{j,h}\,. $\\

We next introduce for each $j$ an approximate operator. \\ 
{\bf Interior balls.\\} For $j\in \mathcal J_i(h)\,$, we use a
linearization of $V$ at $a_j$ and set
\begin{equation}   \label{defCAjh}
\begin{cases}
  \A_{j,h} = -h^2\Delta+i\, \bigl( V(a_j(h))+\nabla V(a_j(h))\cdot(x-a_j(h))\bigr)\,,\\
\Dg(\A_{j,h}) = H^2(\mathbb{R}^n)\cap L^2(\mathbb{R}^n ; |\nabla V(a_j) \cdot x|^2dx)\,.
\end{cases}
\end{equation}
In this case, $\sigma(\A_{j,h}) = \emptyset\,$ and, after the
rescaling $x\mapsto h^{-2/3}x\,$, we may use \eqref{eq:12} to obtain
that that for all $\omega\in\mathbb{R}\,$, there exist $C_\omega>0$
{and $h_0 >0$ such that, for any $h\in (0,h_0]$} and any $j \in
\mathcal J_i(h)$,
\begin{equation}    \label{resCAjh}
\sup_{\Re z\leq \, \omega \, h^{2/3}}\|(\A_{j,h}-z)^{-1}\|\leq \frac{C_\omega}{h^{2/3}}\,.
\end{equation}
To obtain that $C_\omega$ is independent of $j$ and $h$, we use
Assumption \ref{notdeg}. \\
{\bf Boundary balls.}\\
In order to define the approximating operators at the boundary, we
denote by\break  $\mathcal{F}_{b_j} = \mathcal{F}_{b_j(h)}$ the local
diffeomorphism defined by 
\begin{equation}   \label{eq:38}
  \mathcal{F}(s,\rho) = \varphi(s) -  \rho \, \vec\nu(s)\,,
\end{equation}
where $|\rho |<\delta$, $s \in \mathbb R^{n-1}$ and
$\varphi(s)\in\partial\Omega^\# $ for all $s \in B(0,\delta)$ for some
sufficiently small $\delta>0$. We choose $b = b_j(h)$ as the origin,
so that $\varphi(0) = b_j(h)\,$. Let further
\begin{equation}   \label{defvecJ}
  \vec{J} = \nabla V (b_j(h))\cdot D \mathcal{F}(0,0) \quad ; \quad \vec{J}^\prime= \vec{J}-(\vec{J} \cdot \vec \nu)\vec \nu \,,
\end{equation} 
and
\begin{displaymath}
  \cos \theta =    \frac{\vec{J}\cdot \vec{\nu}}{J} \,, 
\end{displaymath}
where $J=|\vec{J}|$.  

In these coordinates, we define our local approximation for $\A_h^\#$
near $b_j(h)$ as
\begin{equation}
\label{eq:13}
\begin{cases}
{\tilde\A_{j,h}^\#} = -h^2\Delta_{s,\rho}+i\left(V(b_j(h))+\vec{J}'(b_j)\cdot s +
  (\vec J(b_j)\cdot \vec \nu)\, \rho \right)\,,\\
\Dg(\tilde \A_{j,h}^\#) = \{ u\in L^2(D^\#)\cap \mathfrak D^\#_n \, | \,
\tilde\A_{j,h}u\in L^2 \} \,,
\end{cases}
\end{equation}
where $D^R=\R^n_+$, $D^{T,i}=\R^n_-\cup\R^n_+$ and $D^{T,o}=\mathbb
R^n_+$ (the superscripts $i$ and $o$ respectively refer to the inner
transmission condition $\partial\Omega_-$ and the outer Neumann
condition or outer Dirichlet condition on $\partial\Omega$). Let
$\epsilon>0$. We split the boundary into two disjoint subsets
\begin{displaymath} 
  \partial\Omega_s^\# = \{x\in\partial\Omega^\# \, | \, \lambda^\#(\vec J(x)\cdot \vec \nu (x) ) \geq \Lambda_{m}^\# -  \epsilon/2 \,\}
  \quad ; \quad \partial\Omega_p ^\#= \partial\Omega^\#\setminus\partial\Omega_s^\# \,,
\end{displaymath}
where $\Lambda_{m}^\#$  is introduced in \eqref{eq:97}.\\
Consider first the case where $b_j\in\partial\Omega_s^\# $. In this
case, a straightforward dilation argument (taking \eqref{eq:94} into
account) and \eqref{resAngle1} lead to
\begin{equation}   \label{eq:15}
\sup_{\Re z\leq(\Lambda_{m}^\#-\epsilon)h^{2/3}}\|(\tilde{A}_{j,h}^\#-z)^{-1}\|\leq C_\epsilon \, h^{-2/3}\,.
\end{equation}
Consider now the case when $b_j\in\partial\Omega_p^\#$.\\
From  the definition of $\partial\Omega_p^\#$, and from either
\eqref{eq:76} or \eqref{eq:77} it follows that there exists
$C_\epsilon >0$, such that
\begin{displaymath}
  \sup_{(x,y)\in\partial\Omega_p^\#\times\partial\Omega_\perp^\#} |\vec J (x) \cdot \vec \nu(x) - \vec J (y) \cdot\vec \nu(y)|\geq C_\epsilon. 
\end{displaymath} 
By Assumption \ref{nondeg2} we therefore get the existence of
$C_\epsilon >0$ such that
\begin{displaymath}
  d(\partial\Omega_p^\#,\partial\Omega_\perp^\#)\geq C_\epsilon \,.
\end{displaymath}
Using Assumption \ref{nondeg2} once again, we get the existence of
$C_\epsilon >0$ and $h_0>0$ such that, for any $h\in (0,h_0]$ and any
$ b_j\in\partial\Omega_p^\#$
\begin{equation*}
  |\vec{J^\prime}(b_j)|\geq C_\epsilon\,.
\end{equation*}
Consequently, for $|\theta|<\theta_0<\pi/2$, we obtain by
\eqref{resAngle2}
\begin{displaymath}
\sup_{\Re z\leq(\Lambda_{m}^\#-\epsilon)h^{2/3}}\|(\tilde{\A}_{j,h}^\# -z)^{-1}\|\leq
\frac{C}{C_\epsilon\, h^{2/3}}\,.
\end{displaymath}
For $0< \theta_0\leq |\theta|\leq\pi/2$ we may use \eqref{ResA+a} to
obtain \eqref{eq:15} once again.  Combining the above with
\eqref{eq:15} and \eqref{resCAjh} then yields that for any $\epsilon
0$, there exists $C(\epsilon)>0$ such that
\begin{equation}    \label{eq:16}
\sup_{
  \begin{subarray}{c}
    \Re z\leq(\Lambda_m^\#-\epsilon)h^{2/3} \\
    j\in \Jg_i\cup \Jg_\partial 
  \end{subarray}}\|(\tilde{\A}_{j,h}^\#-z)^{-1}\|\leq
\frac{C(\epsilon)}{h^{2/3}}\,.
\end{equation}

Following the same arguments as in \cite{Hen2} we can establish that
the linear approximation of $V$ in the local coordinates $s$ and $\rho
$ is a good approximation of $V$ in a closed neighborhood of
$b_j$. More precisely, we have
\begin{equation}   \label{approxV}
  V\circ\mathcal{F}_{b_j} = V(b_j(h))+ \sum_{i=1}^{n-1}J_i^{(j)}s_i+J_n^{(j)}\rho+\mathcal{O}(|(s,\rho)|^2)\,,
\end{equation}
where $\vec{J}^{(j)}=(J_1^{(j)},\ldots,J_n^{(j)})$.

For any $\nu \in\mathbb{R}$, we set
\begin{equation}  \label{defLambdah1}
\lambda(h,\nu ) = \Lambda_0h^{2/3}+i\nu \,,~~~\Lambda_0 =\Lambda_m^\#-\varepsilon\,.
\end{equation}
We now construct an approximation of $(\A_h^\#-\lambda(h,\nu
))^{-1}$. To this end we first define
\begin{equation}    \label{transfoF}
T_{\mathcal{F}}^x : L^2(\Omega\cap B(x,\delta))  \longrightarrow L^2(\mathcal{U})
\quad \text{s.t.}\quad      T_{\mathcal{F}}^x (u) = u\circ\mathcal{F}_x \,,
\end{equation}
where $\Ug=\mathcal{F}_x(B(x,\delta))$.  
Then we set
\begin{equation}  \label{defResBord} 
R_{j,h}=
T_{\mathcal{F}_{b_j}}^{-1}(\tilde\A_{j,h}^\#-\lambda(h,\nu ))^{-1}T_{\mathcal{F}_{b_j}}\,,
\end{equation}
and
\begin{displaymath}
  \widehat {\A}_{j,h}^\# =T_{\mathcal{F}_{b_j}} \A_h^\# T_{\mathcal{F}_{b_j}}^{-1}\,,
\end{displaymath}
which is defined on $H^2(\Ug,\C)$. 

We can now introduce for $h\in(0,h_0]\,$ the following global
approximate resolvent
\begin{equation}   \label{defResApp1}
 \mathcal{R}(h,\nu ) = \sum_{j\in \mathcal J_i(h)}\chi_{j,h}(\A_{j,h}-\lambda(h,\nu ))^{-1}\chi_{j,h} 
+  \sum_{j\in \mathcal J_\partial (h)}\eta_{j,h}  R_{j,h}\eta_{j,h}\,.
\end{equation}
In view of \eqref{eq:90}, $ \mathcal{R}(h,\nu )$ maps $L^2(\Omega^\#)$
into $D(\A_h^\#)$. Hence, we may apply to it $(\A_h^\# -\lambda(h,\nu
))$ to obtain
\begin{equation}    \label{decompRlambda}
 (\A_h^\# -\lambda(h,\nu ))\, \mathcal{R}(h,\nu ) = I  +  \sum_{j\in \mathcal J_i(h)}\B_j \chi_{j,h}  
+  \sum_{j\in \mathcal J_\partial (h)}\B_j\eta_{j,h} \,.
\end{equation}
In the above
\begin{subequations}  \label{eq:6}
\begin{multline}
  \mathcal B_j:= \chi_{j,h}(\A_h^\# -\A_{j,h})(\A_{j,h}-\lambda(h,\nu ))^{-1}\hat \chi_{j,h} \\
  + [\A_h^\# ,\chi_{j,h}](\A_{j,h}-\lambda(h,\nu ))^{-1}\hat \chi_{j,h} \,, \mbox{ for } j\in \mathcal J_i\,,
\end{multline}
and
\begin{multline}
      \mathcal B_j:=   \eta_{j,h} T_{\mathcal{F}_{b_j}}^{-1}
  (\widehat {\A}_{j,h}^\# -\tilde{\A}_{j,h}^\#) (\tilde{\A}_{j,h}^\#
  -\lambda (h,\nu ) )^{-1} T_{\mathcal{F}_{b_j}} \widehat \eta_{j,h} \\ +
  1_{\Omega^\#}  [\A_h^\#,\eta_{j,h}]\, R_{j,h}\,\widehat \eta_{j,h}\,,  \mbox{ for } j\in \mathcal J_\partial \,,
\end{multline}
\end{subequations}
where $\widehat \chi_{j,h}$ and $\widehat \eta_{j,h}$ are such that
\begin{itemize}
\item
$\Supp \widehat \chi_{j,h}\subset B(a_j(h),2 h^{\varrho})$ for $j\in \mathcal J_i(h)\,$,
\item
$\Supp \widehat \eta_{j,h}\subset B(b_j(h),2h^{\varrho})$ for $j\in \Jg_\partial\,$, 
\item  
$\widehat \chi_{j,h} \chi_{j,h}  =\chi_{j,h}$ and $\widehat \eta_{j,h}
  \eta_{j,h}=\eta_{j,h}\,$. 
\end{itemize}
We note that the approximate resolvent $\Rg$, the remainder $\Eg$ and
their various components all depend on $\nu $ through $\lambda (h,\nu
)$.  Hence, our estimates in the sequel must be uniform with respect
to $\nu $.

Following precisely the same steps as in \cite{Hen2} we obtain that
there exists $C>0$ and $h_0 >0$ such that for any $f\in
L^2(\Omega,\C)$, $h\in (0,h_0]$, $\nu \in \mathbb R$ and $
j\in\Jg_i(h)$
\begin{equation}    \label{estCommut1}
\|\B_j\| \leq C \, 
(h^{2/3-{\varrho}}+ h^{2({\varrho}-1/3)}) \,.
\end{equation}
The boundary terms in (\ref{decompRlambda}) can be estimated once
again following the same procedure detailed in \cite{Hen2}. To this
end we shall need the identity
\begin{displaymath}
  \eta_{j,h} T_{\mathcal{F}_{b_j}}^{-1}\Tg =
  T_{\mathcal{F}_{b_j}}^{-1}\tilde{\eta}_{j,h}\Tg \,,
\end{displaymath}
where $\tilde{\eta}_{j,h} =T_{\mathcal{F}_{b_j}}(\eta_{j,h})$ and
$\Tg:L^2(\Ug,\C)\to L^2(\Ug,\C)$. We obtain that there exist $C>0$ and
$h_0 >0$ such that for any $f\in L^2(\Omega^\#,\C)$, $h\in (0,h_0]$,
$\nu \in \mathbb R$ and $j\in\Jg_\partial(h)$
\begin{equation}   \label{estRHS_R1bdry}
\|\B_j\| \leq C\, (h^{\varrho} + h^{2/3-{\varrho}})\,.
\end{equation}
Let
\begin{displaymath}
  \Eg(h,\nu ) = (\A_h^\#-\lambda(h,\nu ))\, \mathcal{R}(h,\nu ) - I\,.
\end{displaymath}
The remainder $\mathcal E (h,\nu )$ has the form 
\begin{equation*}
\mathcal E (h,\nu ):= \sum_j \mathcal C_j  \,,
\end{equation*}
with $\mathcal C_j:= \mathcal B_j \chi_{j,h} $, if $j\in
\mathcal J_i(h)$ or  $\mathcal C_j:= \mathcal B_j \eta_{j,h} $ if  $j\in
\mathcal J_\pa(h)$.

We seek an estimate for
\begin{equation}    \label{eq:72}
    \Big\|\sum_j  \mathcal C_j f\Big\|^2 = \sum_{j,k} \langle \mathcal C_j f, \mathcal C_k f\rangle\,.
\end{equation}
To this end, let $j\in\Jg_i$. If $B(a_k,h^{\varrho})\cap
B(a_j,h^{\varrho})\neq\emptyset$ (or $B(b_j,h^{\varrho})\cap
B(a_j,h^{\varrho})\neq\emptyset$) for some $k\neq j$, then
$B(a_k,h^{\varrho})\subset B(a_j,2h^{\varrho})$ (or
$B(b_k,h^{\varrho})\subset B(a_j,2h^{\varrho})$). For a given
$j\in\Jg_i$,
\begin{multline*}
      \text{card}\{ k\in\Jg_i\cup\Jg_\partial  \,|\, B(a_k,h^{\varrho})\subset
      B(a_j,2h^{\varrho}) \\ \text{ or }
    B(b_k,h^{\varrho})\subset B(a_j,2h^{\varrho})\,\} 
     \leq \frac{B(0,2h^{\varrho})}{B(0,h^{\varrho}/2)}  \leq 4^n\,.
  \end{multline*}
Hence the cardinality of the set of $k\in\Jg_i\cup\Jg_\partial$ such
that $\mathcal C_k^*\mathcal C_j\not\equiv0$ is bounded from above by
$N_0= 4^n$. A similar argument applies for $j\in\Jg_\partial$ as well.
Applying \eqref{eq:72} together with the inequality
\begin{equation}   \label{eq:74}
    | \langle \mathcal C_j f, \mathcal C_k f\rangle| \leq \frac 12 (\| \mathcal C_j
 f\|^2 + \| \mathcal C_k f\|^2) \,, 
\end{equation}
yields
\begin{multline}    \label{eq:59}
   \Big\| \sum_j  \mathcal C_j f\Big\|^2_2 \leq N_0  \sum_j \|\mathcal C _j
 f\|^2_2 \\ \leq N_0 \,  \left(  \sum_{j\in \mathcal J_i}  \|\mathcal B_j\|^2 \| \chi_j f\|^2_2\,
   + \sum_{j\in \mathcal J_\partial}  \|\mathcal B_j\|^2 \| \eta_j f\|^2_2\right)\,. \quad  
\end{multline}
With the aid of \eqref{estCommut1} and \eqref{estRHS_R1bdry} we then
get the existence of $C >0$ and $h_0 > 0$ such that, for any $\nu \in
\mathbb R$, any $h\in (0,h_0]$, any $f\in L^2(\Omega^\#)$,
\begin{equation*}
  \Big\| \sum_j \mathcal \mathcal C_j
 f\Big\|^2 \leq C\, h^{2\inf \{\varrho, \frac 23 -\varrho, 2 (\varrho - \frac
   13)\}} \| f\|^2\,. 
\end{equation*}
Hence, we obtain
\begin{equation*}
  \| \Eg (h,\nu ) \|_{\mathcal L (L^2)}  \leq 
   C^\frac 12\,  h^{\inf \{ \varrho , \frac 2 3 - \varrho  , 2  (\varrho - \frac 13)\}}\,.
\end{equation*}
From  the foregoing discussion we may conclude that
\begin{equation}   \label{eq:106}
  \|\Eg(h,\nu )\|\xrightarrow[h\to0]{}0 \,, 
\end{equation}
uniformly with respect to $\nu \in \mathbb R$.\\
Consequently, $I+\Eg(h,\nu )$ is invertible for sufficiently small
$h$. Hence
\begin{displaymath}
  (\A_h^\# -\lambda(h,\nu ))^{-1} = \mathcal{R}(h,\nu )\, (I+\Eg(h,\nu ))^{-1}\,.
\end{displaymath}
Following the same steps, as in the proof of \eqref{eq:106} (with
differently defined $\CC_j$'s), we obtain, with the aid of
\eqref{resCAjh} and \eqref{eq:16}, that
\begin{displaymath}
 \sup_{\nu  \in \mathbb R}  \|\Rg(h,\nu )\| \leq C_\epsilon \,  h^{-2/3}\,,
\end{displaymath}
and hence
\begin{displaymath}
  \sup_{\nu \in \mathbb R} \|(\A_h^\#-\lambda(h,\nu ))^{-1}\| \leq C_\epsilon \, h^{-2/3}\,.
\end{displaymath}
Since $\A_h^\# $ is accretive (in the case of $\#\in \{R,T\flat\}$ we
assume that $\K \geq 0$), it follows that $\inf
\Re\sigma(\A_h^\#)\geq0\,$.  Thus, we conclude from the above (applied
for any $0<\epsilon<2\Lambda_m^\#$) and \eqref{defLambdah1} that for
all $\epsilon>0$, there exists $h_0 >0$ such that, for all $h\in
(0,h_0]$,
\begin{displaymath}
  \inf \Re \sigma(\mathcal A_h^\#) \geq (\Lambda_m^\# -\epsilon)\,  h^{2/3}\,.
\end{displaymath}
which is equivalent to \eqref{limSpect1} for $\#=D$,
\eqref{limSpect1N} for $\#=N$, \eqref{limSpect1R} for $\#=R$ and
\eqref{limSpect1T} for $\#=T\flat$ with $\flat = D$ or $N$. Furthermore, we
may also prove in  the same manner \eqref{estRes1N},
\eqref{estRes1R}, and \eqref{estRes1T}.

\section{Upper bound}\label{s7}

In this section we prove \eqref{limSpect1Nat} for the Dirichlet case
under a weaker assumption than in \cite{AlHe}. We also prove its
extension to the Neumann, the Robin and the transmission problem. 
The proof is obtained by establishing the existence of an eigenvalue $
\lambda(h)\in\sigma(\A_h^\#)$ such that
\begin{equation}    \label{eq:18}
h^{-2/3}\Re\lambda(h)\xrightarrow[h\to0]{}\Lambda_m^\#  \,,
\end{equation}
where $\Lambda_m^\#$ is defined in \eqref{eq:97}, and the superscript
$\#\in\{D,N,R,T\flat\}$ respectively stands for either the Dirichlet,
or the Neumann, or the Robin, or one of the transmission boundary
conditions, detailed in \eqref{eq:7} or \eqref{eq:7D}. The proof, as
in \cite{AlHe}, can be splitted into two steps: quasimode
constructions and resolvent estimates.

\subsection{Quasimodes}

The relevant quasimodes have already been obtained in \cite{GH}. We
recall here a weaker version of these results which is sufficient for
our application.  Let $x_0\in \pa \Omega^\#_\perp$ (see
\eqref{defPaPerp}). For $n = 2$, let $(s,\rho)$ denote a
curvilinear system whose origin is situated at $x_0$, in a close
neighborhood of $\partial\Omega^\#$, $s$ denotes the arclength in the
positive trigonometric direction, and $\rho$ denotes the signed
distance from $\partial\Omega$: positive inside $\Omega$ (or
$\Omega^+$ in the transmission case when $x_0$ lies on
$\partial\Omega_-$) and negative outside. When $n>2$, we choose
instead of the arclength a local system of coordinates of $\partial
\Omega^\#$ in the neighborhood of $x_0$.


Let further $(s,\rho)=(h^{1/2}\sigma,h^{2/3}\tau)$. For the Robin case
the approximate eigenpair is given by \cite{GH} (where only the case
$n=2$ is treated, and where higher order terms are provided as well):
\begin{subequations}   \label{eq:125}
  \begin{empheq}[left={\empheqlbrace}]{alignat=2} 
&\Lambda^{R,1}_{\K(h)}(h) \kern -3pt = \kern -3pt iV(x_0) + h^{2/3}{\lambda^R_1(\kappa^* (x_0)  )}\, \jf_0 ^{2/3} +
h\sum_{j=1}^{n-1}|\alpha_j|^{1/2}e^{i\pi/4\, \sign \alpha_j} \,, & & \\
&U^{R,1}_{\K(h)}(x,h) =
C^R_1(h)\, \chi_h(|x- x_0|)\,w_1^R(\tau)  \prod_{j=1}^{n-1} \exp(-|\alpha_j|^{1/2}e^{i(\pi/4)\, \sign
    \alpha_j}\sigma_j^2/2 )  \,, & &
\end{empheq}
\end{subequations}
where
\begin{displaymath}
  w_1^R(\tau)=\Ai \bigl(\lambda^R_1(\kappa^* (x_0)) + \tau \jf_0^{1/3}\, e^{i\pi/6} \bigr)
\end{displaymath}
and
\begin{displaymath} 
  x = x_0 + (s,\rho) = x_0 + (h^{1/2} \sigma,h^{2/3}\tau) \,,\,
\K(h)= h^{\frac 43} \kappa\,, \jf_0 = |\nabla V (x_0)| \,, \kappa^* (x_0) = \kappa \, \jf_0^{-\frac
    13}\,.
\end{displaymath}
It is implicitly assumed in the above and in the sequel, without any
loss of generality, that
\begin{equation*}
\frac{\partial V}{\partial\nu}(x_0)<0\,,
\end{equation*} 
(otherwise we consider $\bar{\A}_h$ instead of $\A_h$).  In the above,
the $\{\alpha_j\}_{j=1}^{n-1}$ are given by (\ref{eq:95}b), computed
at $x_0$.  $\chi_h\in C^\infty(\R_+,[0,1])$ is a cutoff function
satisfying
\begin{displaymath}
  \chi_h(r) =
  \begin{cases}
    1 &\mbox{ for }  r<h^\gamma \\
    0 & \mbox{ for } r>2h^\gamma 
  \end{cases}\,,
\quad |\chi^\prime_h (r)| \leq  \frac{2}{h^\gamma} \,,
\end{displaymath}
where $0<\gamma<1/2$. The constant $C^{R,1}(h) $ is chosen such that
$\|U^{R,1}\|_2=1$, and $\lambda^R_1(\kappa^*)$ is the leftmost
eigenvalue of $\LL^R(1,\kappa^*)$ defined in \eqref{eq:99}. \\
It is shown in \cite{GH} that, there exist $C$ and $h_0 >0$ such that
for $h\in (0,h_0]$
\begin{equation}   \label{eq:85}
  \|( \mathcal A_{h,\K(h)}^R-\Lambda^{R,1}_{\K(h)}(h))\,  U^{R,1}_{\K(h)} (\cdot,h)\|_2 \leq C\, h^{7/6} \,.
\end{equation}
We note that the eigenvalues and eigenmodes for the Dirichlet and
Neumann cases are respectively obtained by letting $\kappa\to +\infty$
and $\kappa=0$ respectively.  Hence, we have similar quasimodes for
the Dirichlet and Neumann cases (see
\cite{Hen2} and \cite{AH}). 

For the transmission case $\#= T\flat$, we have (cf. \cite{GH}), if
$x_0 \in \partial \Omega_- \cap \partial \Omega_\perp^\#$,
\begin{subequations}    \label{eq:19}
  \begin{empheq}[left={\empheqlbrace}]{alignat=2}
&\Lambda^{T\flat,1}_{\K(h)} (h)= iV(x_0) + h^{2/3}\lambda^T_1(\kappa^* (x_0))\, \jf_0^{2/3} +
h\sum_{j=1}^{n-1}|\alpha_j|^{1/2}e^{i(\pi/4)\, \sign \alpha_j} \,, & & \\
&U^{T\flat,1}_{\K(h)}(x,h)=C^{T,1}(h) \chi_h(|x-x_0|) w_1^T(\tau) \prod_{j=1}^{n-1}
\exp(-|\alpha_j|^{1/2}e^{i(\pi/4) \sign \alpha_j}\sigma_j^2/2)  \,, & &
\end{empheq}
where
\begin{equation}
w_1^T(\tau)= 
\begin{cases}
\Ai (\lambda^T_1(\kappa^*(x_0))+\tau \jf_0 ^{1/3}e^{i\pi/6}) & \tau>0\,, \\
  \Ai (\bar{\lambda}^T_1(\kappa^* (x_0))+|\tau| \jf_0^{1/3}e^{-i\pi/6})&\tau<0 \,.
\end{cases}
\end{equation}
\end{subequations} 
In the above $\lambda^T_1(\kappa^*)$ denotes one of leftmost
eigenvalues of $\LL^T(1,\kappa^*)$ (we can indeed choose a relabeling
in which this is $\lambda^T_1$) and $C^{T,1}(h)$ is chosen such that
$\|U^{T,1}(\cdot,h)\|_2=1$.  Note that this quasimode is independent
of the boundary condition $\flat$ on the exterior boundary $\partial
\Omega$.  We note also that \eqref{eq:19} will be used whenever
$\Lambda_m^\#\,$, defined in \eqref{eq:97}, is obtained on
$\partial\Omega_-$. Otherwise, if it is obtained on $\partial\Omega$
we need to use the quasimode for either the Dirichlet or the Neumann
condition depending on whether we consider
\eqref{eq:7D} or \eqref{eq:7}, i.e. $\flat =D$ or $N$.

We borrow from \cite{GH}, as we did in \eqref{eq:85}, the estimate
\begin{equation}   \label{eq:20}
\|(\mathcal A_{h,\K(h)}^{T\flat } -\Lambda^{T\flat ,1}_{\K(h)}(h)) \, U^{T\flat,1} _{\K(h)}(\cdot,h)\|_2 \leq C \, h^{7/6} \,.  
\end{equation}

We now turn to prove the existence of an eigenvalue of $\A_h^\#$
satisfying \eqref{eq:18}. To this end we shall obtain, in the next
subsection, a resolvent estimate for $\A_h^\#$ in the vicinity of some
$x_0 \in \Sg^\#$. More precisely, we seek an estimate for
$\chi_h(\A_h^\#-\lambda)^{-1}\chi_h$ for $\lambda$ which is close to
$\Lambda^{\#,1}$.

\subsection{Resolvent estimates for simplified models}

For $\kappa^* \geq 0$, let $\LL_\tau^\#:= \LL_\tau^\#(1,\kappa^*) $ as
before.  Let $\{\lambda_k^\#\}_{k=1}^{+\infty}$ and
$\{v_k^\#\}_{k=1}^{+ \infty}$ respectively denote the eigenvalues of
$\LL_\tau^\#$ and their associated eigenfunctions normalized by
\begin{equation*}
\langle v_k^\#\,,\, \bar{v}_k^\# \rangle=1\,.
\end{equation*}
Since the eigenvalues of $\LL_\tau^\#$ are all simple, this
normalization is indeed possible (see \cite{AD}) and we denote by
\begin{displaymath}
  \Pi_k^\#  = v_k^\#\, \langle \cdot\,,\, \bar{v}_k^\#\rangle \,,
\end{displaymath}
the spectral projection on ${\rm span}\, \{v^\#_k\}$.

Recall the definition of $\lambda^\#=\lambda^\#(1)$ from
\eqref{eq:75}. Let $K^\#$ denote the number of eigenvalues whose real
value is given by $\lambda^\#$. When $\#\in\{D,N,R\}$ we have $K^\#=1$
whereas for $\#=T$, $K^\#\geq2$.  We thus set, as in Section \ref{s3},
\begin{equation}    \label{eq:81}
P_1^\# = \sum_{k=1}^{K^\#}\Pi_k^\# \,. 
\end{equation}
Let $\lambda^{\#,2}=\Re\lambda_2^\#$ for $\#\in\{D,N,R\}$, and let
$\lambda^{T,2}$ be defined as in the proof of Proposition
\ref{lemmaunifbisTimp}.  Combining \eqref{eq:14} and \eqref{eq:22}
yields that for any $\epsilon>0$ and $\kappa_0^*>0$, there exists
$C_\epsilon^\# (\kappa_0^*)>0$ such that for all
$\kappa^*\in[0,\kappa_0^*]$ we have
\begin{equation}    \label{eq:21}
\| e^{-t\LL_\tau^\#}(I-P_1^\# )\|\leq C_\epsilon^\#(\kappa^*) \, e^{-t \, (\lambda^{\#,2}-\epsilon)} \,.
\end{equation}
We now consider the tangential operator
\begin{displaymath}
   \LL_\sigma = -\Delta_\sigma  + i\, V_\sigma \,,
\end{displaymath}
where
\begin{displaymath}
  V_\sigma=\sum_{j=1}^{n-1}\alpha_j\sigma_j^2
\end{displaymath}
be defined on
\begin{displaymath} 
   D(\LL_\sigma)= \{ u\in H^2(\R^{n-1},\C) \,| \, V_\sigma \, u\in L^2(\R^{n-1},\C) \,\} 
\end{displaymath}
(see Corollary \ref{cor:entirespacer=1} which holds for $\LL_\sigma$
as well).  Let
\begin{equation}    \label{eq:39}
   \LL_j = -\partial^2_{\sigma_j}  + i\alpha_j^2\sigma_j^2\,,
\end{equation}
where $\alpha_j\neq0$, be defined on 
\begin{displaymath}
   D(\LL_j)= \{ u\in H^2(\R,\C) \,| \,
  |\sigma_j|^2u\in L^2(\R ,\C) \,\} \,.
\end{displaymath}
It is well known (see for example \cite[Corollary 14.5.2]{da07}) that
there exists $C>0$, independent of $\alpha_j$, such that, for any
$t\geq 0\,$,
\begin{displaymath}
  \|e^{-t\LL_j}\| \leq C\, e^{-|\alpha_j/2|^{1/2}t} \,.
\end{displaymath}
As
\begin{displaymath}
  e^{-t\LL_\sigma} = e^{-t\LL_1}\otimes\cdots\otimes e^{-t\LL_{n-1}}\,,
\end{displaymath}
we obtain that
\begin{equation}    \label{eq:26}
\|e^{-t\LL_\sigma}\|\leq C^{\,n-1} \,\exp( -t \,\mu_1^r ) \,,
\end{equation}
where
\begin{equation}    \label{eq:29}
  \mu_1^r=\sum_{j=1}^{n-1}\left(\frac{|\alpha_j|}{2}\right)^{1/2}\,.
\end{equation}
For $\varepsilon>0$, we define the operator
\begin{subequations}  \label{eq:123}
  \begin{equation}
  \B_\varepsilon^\# = \LL_\tau^\# +\varepsilon^{1/2}\LL_\sigma 
\end{equation}
on
\begin{equation}
   D(\B_\varepsilon^\#) = \{u\in L^2(\R^{n}_\#,\C) \,| \, \B_\varepsilon^\#u\in L^2(\R^n_\#,\C) \,\}\,.
\end{equation}
\end{subequations}
Set
\begin{displaymath}
  \mu_1 = \sum_{j=1}^{n-1}|\alpha_j|^{1/2}e^{i(\pi/4)\, \sign \alpha_j}\,,
\end{displaymath}
and then let 
\begin{equation}    \label{eq:124}
\Lambda(\varepsilon) =\lambda_1^\#+\varepsilon^{1/2}\mu_1 \,.
\end{equation}
In the case $\#=T$ , we may choose any $\lambda_\ell^\# $ (with $1\leq
\ell\leq K^\#$) satisfying $\Re\lambda_\ell^\#=\lambda^\#$.  We set
$\ell=1$ without any loss of generality.  When
$\{\alpha_j\}_{j=1}^{n-1}\in\R_+^{n-1}$, it has been established in
\cite{AlHe} that there exist positive $r_0 $, $\varepsilon_0$, and
$C$, such that, for $r\in (0,r_0]\,$, $\varepsilon \in
(0,\varepsilon_0]$ \, and $\lambda\in\partial
B(\Lambda(\varepsilon),r\varepsilon^{1/2})$,
\begin{displaymath}
\| (\B_\varepsilon^D-\lambda)^{-1}\|\leq \frac{C}{r\varepsilon^{1/2}} \,.
\end{displaymath}

Since the technique used in \cite{AlHe} cannot be applied neither to
the case when there exists $1\leq j\leq n-1$ for which $\alpha_j<0$
when $\# \in \{D,N,R \}$, nor to $\B_\epsilon^T$, we apply here a more
general technique that could be applied in all cases.
\begin{lemma} 
\label{lemma7.1} 
There exist positive $r_0 $, $\varepsilon_0$, and $C$, such that, for
all $\varepsilon\in (0,\varepsilon_0]$ and $r \in (0,r_0]$,
$B(\Lambda(\varepsilon),r\varepsilon^{1/2})$ must belong to the
resolvent set of $\B_\varepsilon^\#$, and such that for any
$\lambda\in B(\Lambda(\varepsilon),r\varepsilon^{1/2})$,
\begin{equation}    \label{eq:3}
\| (\B_\varepsilon^\#-\lambda)^{-1}\|\leq \frac{C}{r\varepsilon^{1/2}} \,.
\end{equation}
Furthermore, for every $a>0$ there exist $C_a>0$ and $\varepsilon_0$,
such that, for $\varepsilon \in (0, \varepsilon_0]$,
\begin{equation}     \label{eq:24}
\sup_{\Re\lambda\leq\Re\Lambda(\varepsilon) -a\,\varepsilon^{1/2}}\| (\B_\varepsilon^\#-\lambda)^{-1}\| \leq \frac{C_a}{\varepsilon^{1/2}} \,.
\end{equation}
\end{lemma}
\begin{proof}
Let $P_1^\#$ be given by \eqref{eq:81}. Since $\LL_\tau^\#$ and
$\LL_\sigma$ commute we have
\begin{displaymath}
 (I-P_1^\#) e^{-t\B_\varepsilon^\#}=
 (I-P_1^\#)e^{-t\LL_\tau^\#}\otimes e^{-t\varepsilon^{1/2}\LL_\sigma}=e^{-t\LL_\tau^{\#,2}}\otimes
 e^{-t\varepsilon^{1/2}\LL_\sigma}\,.
\end{displaymath}
Hence, by \eqref{eq:21} (with $\epsilon = \frac 12 (\lambda^{\#,2} -
\lambda^{\#})$) and \eqref{eq:26}, there exists $C>0$ such that, for
all $t\geq 0\,$, $\varepsilon >0\,$,
\begin{displaymath}
  \|e^{-t\B_\varepsilon^\#}(I-P_1^\#) \|\leq C \, e^{-\frac t 2 (\lambda^{\#} + \lambda^{\#,2})} \,e^{-t\varepsilon^{1/2}\mu_1^r} \,,
\end{displaymath}
where $\lambda^{\#,2}$ is given by \eqref{eq:91} and
\eqref{eq:92}. Since
\begin{displaymath}
  (\B_\varepsilon^\#-\lambda)^{-1}(I-P_1^\#) = \int_0^\infty
  e^{-t\B_\varepsilon^\#}(I-P_1^\#) \,dt \,,
\end{displaymath}
there exist positive $\hat C$, $r_0$ and $\varepsilon_0$ such that,
for $\varepsilon\in (0,\varepsilon_0]$, $r \in (0,r_0]$, and
$\lambda\in\partial B(\Lambda(\varepsilon),r\varepsilon^{1/2})$ or $\Re\lambda\leq\Lambda(\varepsilon)-a\varepsilon^{1/2}\,$,
\begin{multline}\label{eq:4}
  \| (\B_\varepsilon^\# -\lambda)^{-1}(I-P_1^\#)\| \leq C \int_0^\infty \exp \bigl\{-t(( \lambda^\#+ \lambda^{\#,2})/2 
  +\varepsilon^{1/2}\Re \mu_1-\Re\lambda) \bigr\} \,dt \\
   \leq C \int_0^\infty \exp \bigl\{-t(( \lambda^{\#,2} - \lambda^\# )/2
  - r \varepsilon^{1/2}) \bigr\} \,dt
    \leq  \hat C\,.
\end{multline}
To complete the estimate of $\| (\B_\varepsilon^\#-\lambda)^{-1}\|$ we
first observe that
\begin{displaymath}
  \|(\B_\varepsilon^\#-\lambda)^{-1}P_1^\#\|\leq \sum_{k=1}^{K^\#}
  \|(\varepsilon^{1/2}\LL_\sigma-\lambda+\lambda_k^\#)^{-1}\|\,\|\Pi_k^\#\|
 \,.
\end{displaymath}
By the Riesz-Schauder theory there exist $C>0$ and some fixed
neighborhood $\Ug$ of $\mu_1$, such that, for any $\mu\in\Ug$,
\begin{displaymath}
    \|(\LL_\sigma-\mu)^{-1}\|\leq \frac{C}{|\mu-\mu_1|} \,.
\end{displaymath}
Consequently
\begin{displaymath}
  \|(\B_\varepsilon^\#-\lambda)^{-1}P_1^\#\|\leq\frac{C}{|\lambda-\Lambda (\varepsilon)|}\,, 
\end{displaymath}
which together with \eqref{eq:4} readily yields both \eqref{eq:3} and
\eqref{eq:24}.
\end{proof}

Unlike in \cite{AlHe}, the assumptions made on $V$ in this case do not
guarantee that $V\neq V(x_0)$ in some neighborhood of $x_0$. This
makes the resolvent estimates more complicated, since it is much more
difficult to prove exponential decay of the error terms away from
$x_0$.  We thus need the following auxiliary lemma, where we introduce
the characteristic function in the normal variable $ \mathbf 1_{
|\tau|\geq\varepsilon^{-a}}$ of the set $\{
|\tau|\geq\varepsilon^{-a}\}$.
\begin{lemma} 
Let $a>0$ and $b>0$. Then there exist positive $C$, $r_0$ and
$\varepsilon_0$ such that, for $r\in (0,r_0]$ and $\varepsilon \in
(0,\varepsilon_0]$, if either $\lambda\in\partial B(\Lambda
(\varepsilon),r\varepsilon^{1/2})$, or $\Re\lambda\leq
\Re\Lambda(\varepsilon) -b\, \varepsilon^{1/2}$, we have
\begin{subequations}  \label{eq:23}
\begin{equation}
\big\|{\mathbf 1}_{|\tau|\geq\varepsilon^{-a}}(\B_\varepsilon^\#-\lambda)^{-1}\big\| \leq C\,,
\end{equation}
\begin{equation} 
    \big\|\partial_\tau(\B_\varepsilon^\#-\lambda)^{-1}\big\|
    +\epsilon^{1/4}\big\|\nabla_\sigma(\B_\varepsilon^\#-\lambda)^{-1}\big\| \leq
    C\,(1+\|(\B_\varepsilon^\#-\lambda)^{-1}\|)\,,
\end{equation}
and, for any $f\in L^2(\R^n_\#,\C)$ such that $\sigma f \in
L^2(\R^n_\#,\C^{n-1})$
\begin{equation} 
   \big\|(\partial_\tau^2+\varepsilon^{1/2}\Delta_\sigma)(\B_\varepsilon^\#-\lambda)^{-1}f\big\| \leq
    C \left( \|(\B_\varepsilon^\#-\lambda)^{-1}\| +1 \right) \big(\|f\|_2 +
  \varepsilon^{1/2}\|\sigma f\|_2\big)\,.
\end{equation}
Finally, we have
\begin{equation} 
    \big\|{\mathbf 1}_{|\tau|\geq\varepsilon^{-a}}\partial_\tau(\B_\varepsilon^\#-\lambda)^{-1}\big\|
    +\varepsilon^{1/4}\big\|{\mathbf 1}_{|\tau|\geq\varepsilon^{-a}}\nabla_\sigma(\B_\varepsilon^\#-\lambda)^{-1}\big\| \leq C\,.
\end{equation}
\end{subequations}
\end{lemma}
\begin{proof}
Obviously,
\begin{displaymath}
    \big\|{\mathbf 1}_{|\tau|\geq\varepsilon^{-a}}(\B_\varepsilon^\#-\lambda)^{-1}\big\| \leq
    \big\|{\mathbf 1}_{|\tau|\geq\varepsilon^{-a}}(\B_\varepsilon^\#-\lambda)^{-1}P_1^\# \big\|  +
    \big\|{\mathbf 1}_{|\tau|\geq\varepsilon^{-a}}(\B_\varepsilon^\#-\lambda)^{-1}(I-P_1^\#)\big\|  \,.
\end{displaymath}
By \eqref{eq:4} there exists $C>0$ such that 
\begin{equation}   \label{eq:40}
 \big\|{\mathbf 1}_{|\tau|\geq\varepsilon^{-a}}(\B_\varepsilon^\#-\lambda)^{-1}(I-P_1^\# )\big\| \leq C\,.
\end{equation}
In view of the asymptotic behavior of the $\{v_k\}_{k=1}^{K^\#}$
(given by \eqref{eq:83} for $\#\in\{D,N,R\}$, and \eqref{eq:84} for
$\#=T$), which can be derived from the behavior of $\Ai(\tau)$ as
$\tau\to\infty$ (see \eqref{5}), we have
\begin{equation}   \label{eq:43}
   \big\|{\mathbf 1}_{|\tau|\geq\varepsilon^{-a}}P_1^\#\|\leq C\, \exp\biggl( -\frac{\sqrt{2}}{3}\varepsilon^{-3a/2} \biggr)\,.
\end{equation}
Consequently, as $(\B_\varepsilon^\#-\lambda)^{-1}$ and $P_1^\# $
commute,
\begin{displaymath}
    \big\|{\mathbf 1}_{|\tau|\geq\varepsilon^{-a}}(\B_\varepsilon^\#-\lambda)^{-1}P_1^\# \big\| 
\leq C\, \exp\biggl(-\frac{\sqrt{2}}{3}\varepsilon^{-3a/2}\biggr)\,.
\end{displaymath}
The above, combined together with \eqref{eq:40}, yields (\ref{eq:23}a).

To prove (\ref{eq:23}b) we let $w=(\B_\varepsilon^\#-\lambda)^{-1}f$ for  $f\in
L^2(\R^n_\#,\C)$. As
\begin{displaymath}
  \Re \langle w,(\B_\varepsilon^\#-\lambda)w\rangle \geq   \|\pa_\tau w \|_2^2+
  \varepsilon^{1/2}\,\|\nabla_\sigma w \|_2^2 - \Re\lambda \, \|w\|_2^2 \
\end{displaymath}
(where the inequality reflects the missing boundary term for
$\#\in\{R,T\}$), we easily obtain the existence of a constant $C>0$
such that, for $f\in L^2(\R^n_\#,\C)$, such that
\begin{displaymath} 
  \|\pa_\tau w \|_2^2+\varepsilon^{1/2}\| \nabla_\sigma w \|_2^2 \leq C\, ( \|w\|_2^2+\|f\|_2^2) \,.
\end{displaymath}
From   this  (\ref{eq:23}b) readily follows.  \\
To prove (\ref{eq:23}c) we begin by writing
\begin{displaymath}
  (\B_\varepsilon^\#-\lambda)(\sigma_jw) = \sigma_jf +  2\frac{\partial w}{\partial\sigma_j} \,.
\end{displaymath}
This implies that $\sigma_j w$ (which satisfies the $\#$-boundary
condition) belongs to $\mathcal D(B_\epsilon^\#)$ when $\sigma_j f \in L^2$.
Hence,
\begin{equation}    \label{eq:48} 
  \|\sigma_jw\|_2\leq  \|(\B_\varepsilon^\#-\lambda)^{-1}\|\Big(\|\sigma_jf\|_2 + 2\,
  \Big\|\frac{\partial w}{\partial\sigma_j}\Big\|_2 \Big)\,.  
\end{equation}
We now write 
\begin{multline*}
  \Re \langle-(\partial_\tau^2+\varepsilon^{1/2}\Delta_\sigma )w, (\B_\varepsilon^\#-\lambda)w\rangle \geq 
  \|(\partial_\tau^2+\varepsilon^{1/2}\Delta_\sigma )w\|_2^2  \\
 \quad - \Re\lambda(\|\pa_\tau w \|_2^2+\varepsilon^{1/2}\|\nabla_\sigma w \|_2^2) +2\, \Im \langle \pa_\tau w ,w\rangle +
  2\varepsilon\sum_{j=1}^{n-1}\alpha_j\Im\langle \pa_{\sigma_j} w ,\sigma_jw\rangle \,.
\end{multline*}
The inequality reflects, once again, the missing boundary term for
$\#\in\{R,T\}$, which is nonnegative under the assumption that $\kappa
\geq 0\,$. Consequently, by (\ref{eq:23}b) and \eqref{eq:48} we have
\begin{displaymath} 
\|  (\partial_\tau^2+\varepsilon^{1/2}\Delta_\sigma )w\|_2 \leq C  \left( \|(\B_\varepsilon^\#-\lambda)^{-1}\| +1\right) \, \Big(\|f\|_2 +
  \varepsilon^{1/2}\|\sigma f\|_2\Big)\,.
\end{displaymath}
To prove (\ref{eq:23}d) we introduce a cutoff function
$\zeta_\varepsilon\in C^\infty(\R_\#,[0,1])$ satisfying
\begin{displaymath}
  \zeta_\epsilon (\tau)  =
  \begin{cases}
    1 & |\tau|\geq\varepsilon^{-a} \\
    0 & |\tau|\leq \varepsilon^{-a}-1
  \end{cases}
\quad ; \quad |\zeta_\varepsilon^\prime(\tau) |\leq 2 \,.
\end{displaymath}
For $f\in L^2(\R^n_\#,\C)$ we write as before
$w=(\B_\varepsilon^\#-\lambda)^{-1}f$. As
\begin{displaymath}
  \Re \langle \zeta_\varepsilon^2w,(\B_\varepsilon^\#-\lambda)w\rangle \geq   \|\pa_\tau (\zeta_\varepsilon w) \|_2^2+
  \varepsilon^{1/2}\, \|\zeta_\varepsilon \nabla_\sigma w \|_2^2 - \|w \, \zeta'_\varepsilon\|_2^2 - \Re\lambda\|\zeta_\varepsilon w\|_2^2 \,,
\end{displaymath}
we obtain that
\begin{displaymath}
  \|{\mathbf 1}_{|\tau|\geq\varepsilon^{-a}} \pa_\tau w \|_2^2+\varepsilon^{1/2}\,\|{\mathbf
    1}_{|\tau|\geq\varepsilon^{-a}} \nabla_\sigma w \|_2^2 \leq C\, ( \|{\mathbf
    1}_{|\tau|\geq (\varepsilon^{-a}-1)} \, w\|_2^2+\|f\|_2^2) \,.
\end{displaymath}
For sufficiently small $\varepsilon$ we may now use (\ref{eq:23}a)
(with $\varepsilon /2$ instead of $\varepsilon$) to obtain
(\ref{eq:23}d).
\end{proof}
\begin{remark}
Using a standard regularity theorem for the $\#$ realization of the
Laplacian in $\mathbb R^n_\#$ we may conclude from (\ref{eq:23}c) that
\begin{equation*}
 \sum_{|\alpha|+|\beta| =2}  \big\| \partial_\tau^\alpha (\varepsilon^{1/4}\pa_\sigma)^\beta\, (\B_\varepsilon^\#-\lambda)^{-1}f\big\| \leq
    C \left( \|(\B_\varepsilon^\#-\lambda)^{-1}\| +1 \right) \big(\|f\|_2 +
  \varepsilon^{1/2}\|\sigma f\|_2\big)\,.
\end{equation*}
\end{remark}
\vspace{1ex}

For later reference we need yet the following result
\begin{lemma}
  Let $0<a<3/4$. Then there exist positive $C$, $r_0$ and
$\varepsilon_0$ such that:
\begin{itemize}
\item  For $r\in (0,r_0]$ and $\varepsilon \in
(0,\varepsilon_0]$,  $\lambda\in\partial B(\Lambda
(\varepsilon),r\varepsilon^{1/2})$, we have
\begin{subequations}  
\label{eq:88}
\begin{equation}
 \big\|{\mathbf 1}_{|\sigma|\geq\varepsilon^{-a}}\partial_\tau(\B_\varepsilon^\#-\lambda)^{-1}\big\|+\big\|{\mathbf 1}_{|\sigma|\geq\varepsilon^{-a}}(\B_\varepsilon^\#-\lambda)^{-1}\big\| \leq
\frac{C}{r\varepsilon^{1/2-2a/3}}\,,
\end{equation}
\begin{equation} 
   \big\|{\mathbf 1}_{|\sigma|\geq\varepsilon^{-a}}\nabla_\sigma(\B_\varepsilon^\#-\lambda)^{-1}\big\| \leq
    \frac{C}{r\varepsilon^{1/2-a/3}}\,. 
\end{equation}
\end{subequations}
\item For $\varepsilon \in
(0,\varepsilon_0]$, $b>0$  and $\Re\lambda\leq
\Re\Lambda(\varepsilon) -b\, \varepsilon^{1/2}$, we have
\begin{subequations}  
\label{eq:136}
\begin{equation}
 \big\|{\mathbf 1}_{|\sigma|\geq\varepsilon^{-a}}\partial_\tau(\B_\varepsilon^\#-\lambda)^{-1}\big\|+\big\|{\mathbf 1}_{|\sigma|\geq\varepsilon^{-a}}(\B_\varepsilon^\#-\lambda)^{-1}\big\| \leq
\frac{C}{(\min(1,b))\varepsilon^{1/2-2a/3}}\,,
\end{equation}
\begin{equation} 
   \big\|{\mathbf 1}_{|\sigma|\geq\varepsilon^{-a}}\nabla_\sigma(\B_\varepsilon^\#-\lambda)^{-1}\big\| \leq
    \frac{C}{(\min(1,b))\varepsilon^{1/2-a/3}}\,. 
\end{equation}
\end{subequations}
\end{itemize}
\end{lemma}
\begin{proof}
By (\ref{eq:4}) 
  \begin{displaymath}
    \|(\B_\varepsilon^\#-\lambda)^{-1}(I-P_1^\#)\|\leq C \,,
  \end{displaymath}
Furthermore since, for every $w\in D(\B_\varepsilon^\#)$, 
\begin{displaymath}
 \Re \langle(I-P_1^\#)w,(\B_\varepsilon^\#-\lambda)(I-P_1^\#)w\rangle = \|\partial_\tau(I-P_1^\#)w\|_2^2
  +  \varepsilon^{1/2}\|\partial_\sigma(I-P_1^\#)w\|_2^2- \Re\lambda\|(I-P_1^\#)w\|_2^2 \,,
\end{displaymath}
we easily get
\begin{equation}
   \|\partial_\tau(\B_\varepsilon^\#-\lambda)^{-1}(I-P_1^\#)\|\leq C \,,
\end{equation}
and
\begin{equation}\label{eq:136bis} 
   \|\nabla_\sigma (\B_\varepsilon^\#-\lambda)^{-1}(I-P_1^\#)\|\leq C \varepsilon^{-\frac 14} \,.
\end{equation}

Consequently, it suffices to  prove (\ref{eq:88}a) for ${\mathbf
  1}_{|\sigma|\geq\varepsilon^{-a}}(\B_\varepsilon^\#-\lambda)^{-1}P_1^\#$.\\
   Let $w\in D(\B_\varepsilon^\#)$
and $g\in L^2(\R^n_\#)$ satisfy
\begin{displaymath}
  (\B_\varepsilon^\#-\lambda)w=g \,.
\end{displaymath}
We further introduce $P_1^\#w=w_1(\sigma)v_1^\#(\tau)$ (or
$w_1=\langle w\,,\,  \overline{v_1^\#}\, \rangle_\tau$) and $g_1=\langle g \,,\,  \overline{v_1^\#}\, \rangle_\tau$). \\
Then,
\begin{equation}
\label{eq:128}
  (\LL_\sigma-\nu_1)w_1= \frac{g_1}{\varepsilon^{1/2}}\,,
\end{equation}
where $\nu_1=\varepsilon^{-1/2}(\lambda-\lambda_1^\#)$, and $\lambda_1^\#$ is the leftmost
eigenvalue of $\LL_\tau^\#$. \\
If $-1\leq\Re\nu_1\leq\mu_1^r+1$, where $\mu_1^r$
is given by (\ref{eq:29}), we  may write
\begin{displaymath} 
  (\LL_\sigma -i \, \Im\nu_1)w_1 = \frac{g_1}{\varepsilon^{1/2}}+\Re\nu_1w_1 \,,
\end{displaymath}
to obtain from \eqref{eq:82}, with
\begin{displaymath}
  V=\sum_{j=0}^{n-1}\alpha_j\sigma_j^2-\Im\nu_1\,,
\end{displaymath}
that
\begin{equation}
\label{eq:89}
  \||\sigma|^{2/3}w_1\|_2 +  \||\sigma|^{1/3}\nabla w_1\|_2  \leq C\Big(\frac{|| g_1|| }{\varepsilon^{1/2}}+ |\Re\nu_1|\,||w_1|| \Big)\,.
\end{equation}
Note that $C$ is independent of $\Im\nu_1$ as
$$|D^2V|\leq C_0[1+|\nabla V|^2+|V|^2]^{1/2}$$
 where $C_0$ is independent of
$\Im\nu_1$. \\
In the first case, i.e. for $\lambda\in\partial B(\Lambda(\varepsilon),r\varepsilon^{1/2})$ we have by (\ref{eq:3}) 
\begin{displaymath}
  \|w_1\|_2\leq\frac{C}{r\varepsilon^{1/2}}||g_1||_2\,,
\end{displaymath}
which, when substituted into (\ref{eq:89}) easily yields
both (\ref{eq:88}a) and (\ref{eq:88}b) via the inequality  $|\sigma|^p\geq\varepsilon^{-pa}{\mathbf
  1}_{|\sigma|\geq\varepsilon^{-a}}$.\\
We now consider the second case: $\Re\lambda\leq
\Re\Lambda(\varepsilon) -b\, \varepsilon^{1/2}$.  For 
$$
\Re\Lambda(\varepsilon) -(\mu_1^r+1)\, \varepsilon^{1/2}\leq\Re\lambda\leq\Re\Lambda(\varepsilon) -b\, \varepsilon^{1/2}
$$ 
we obtain (\ref{eq:136}a) in the
same manner using this time \eqref{eq:24}.
For 
$$
\Re\lambda<\Re\Lambda(\varepsilon) -(\mu_1^r+1)\varepsilon^{1/2}
$$ 
(i.e. for $\Re\nu_1<-1$) we may use the accretiveness of $\LL_\sigma$ to get
\begin{displaymath}
 -\Re\nu_1\|w_1\|_2\leq \varepsilon^{-1/2}\|g_1\|_2 \,.
\end{displaymath}
Substituting the above into \eqref{eq:89} yields (\ref{eq:136}a) 
  and (\ref{eq:136}b) for
$\Re\nu_1\leq -1$ as well.
\end{proof}

The last auxiliary result we present here is useful for the estimate
of the resolvent in regions where $\nabla V$ is nearly perpendicular
to $\partial\Omega$.
\begin{lemma}
\label{lem:oblique}
For $\ell \in \{1,\ldots,K^\#\}$, let $0<r_0\leq
d(\lambda_\ell^\#,\sigma(\LL^\#_\tau)\setminus\lambda_\ell^\# )/2$ and
\break $\lambda\in\partial B(\lambda_\ell^\#,r_0)$.  Let further $\vec
J\in\R^n$ satisfy $|\vec J\times\hat{i}_\tau|>0$, i.e. $J'\neq 0$, and
\begin{displaymath} 
    \B_{\vec{J}}^\#  = -\partial^2_\tau- \Delta_\sigma + i\,[J^\prime\cdot \sigma+(1+\tilde{J}_n)\tau]\,,
\end{displaymath}
be defined on~\eqref{descrDomA+}, with $J_n =1 + \tilde J_n\,$.  Then,
for any $a>0$, $\kappa_0>0$, there exist positive $C(\kappa_0)$ and
$\delta$ such that if
\begin{equation}\label{asb}
  \exp\bigl(-\varepsilon^{-a/2}\bigr)\leq |J'|\leq \delta\,, \quad
|\tilde{J}_n|\leq\delta\,, \quad \kappa\in[0,\kappa_0] \,, \mbox{ and } \quad 0<\varepsilon<1\,,
\end{equation}
then
\begin{subequations} \label{eq:42}
\begin{equation}
\big\|{\mathbf 1}_{|\tau|\geq\varepsilon^{-a}}( \B_{\vec J}^\#  -\lambda)^{-1}\big\| \leq C
\end{equation}
and
\begin{equation}
\big\|{\mathbf 1}_{|\tau|\geq\varepsilon^{-a}}\nabla_{\sigma,\tau}( \B_{\vec J} ^\#  -\lambda)^{-1}\big\| \leq C\,.
\end{equation}
\end{subequations}

Finally, for sufficiently small $\delta$ and for
$|\Re(\lambda-\lambda^\#_\ell (1+\tilde{J}_n)^{1/3})|\leq|J^\prime|^{2/3}$ we have 
\begin{equation}
  \label{eq:130}
\big\|\nabla_\sigma( \B_{\vec J}^\#  -\lambda)^{-1}\big\| \leq \frac{C}{|J^\prime|^{1/3}}\,.
\end{equation}
\end{lemma}
\begin{proof}
The proof is very similar to the proof of \eqref{eq:23}. Without any
loss of generality we assume $\tilde{J}_n=0$, otherwise we rescale
$\tau$ and $\sigma$ by $(1+\tilde{J}_n)^{1/3}$ and proceed
similarly. There exists $\delta>0$ such that, under assumption
\eqref{asb},
\begin{displaymath} 
   \| (I-P_1^\#) e^{-t\B_{\vec J} ^\#}\|\leq C\, e^{-t (\lambda^\#+\lambda^{\#,2})/2}\,,
\end{displaymath}
where $P_1^\#$ is given by \eqref{eq:81}. \\
Hence,
\begin{equation}
\label{eq:131}
  \|( \B_{\vec J}^\#  -\lambda)^{-1}(I-P_1^\#)\|\leq C \,,
\end{equation}
and since by \eqref{eq:43} and \eqref{eq:28} we have 
\begin{displaymath}
  \|{\mathbf 1}_{|\tau|\geq\varepsilon^{-a}}( \B_{\vec J}^\#  -\lambda)^{-1}P_1^\#\|\leq
  \frac{C}{ |J^\prime|^{2/3}} \, \exp\biggl(-\frac{\sqrt{2}}{3}\varepsilon^{-3a/2} \biggr)\,,
\end{displaymath}
(\ref{eq:42}a) readily follows under the lower bound assumption
appearing in \eqref{asb}. The proof of (\ref{eq:42}b) is obtained in a
similar manner to the proof of (\ref{eq:23}d).

It can be easily verified that for every $w\in D(\B_{\vec J}^\# )$ 
\begin{displaymath}
  \|\nabla_\sigma(I-P_1^\#)w\|_2^2 \leq \Re\lambda\|(I-P_1^\#)w\|_2^2 +
  \|(I-P_1^\#)w\|_2\|( \B_{\vec J}^\#  -\lambda)(I-P_1^\#)w\|_2 
\end{displaymath}
By \eqref{eq:131} we thus have
\begin{displaymath}
  \|\nabla_\sigma(I-P_1^\#)w\|_2 \leq C\,.
\end{displaymath}
Consequently, once we establish \eqref{eq:130} for $P_1^\#w$ it will
hold for $w$ as well. Let then $w_1(\sigma)=\langle v_l^\#,w\rangle_\tau$. As
\begin{displaymath}
  (-\Delta_\sigma+iJ^\prime\cdot \sigma)w_1 = (\lambda-\lambda_\ell^\#)w_1 + g_1\,,
\end{displaymath}
where $g_1=\langle v_l^\#,( \B_{\vec J}^\#  -\lambda)w\rangle_\tau$,
we easily obtain, via integration by parts that
\begin{displaymath}
  \|\nabla_\sigma w_1\|_2^2 \leq |\Re(\lambda-\lambda_\ell^\#)|\|w_1\|_2^2 +  \|w_1\|_2\|g_1\|_2\,.
\end{displaymath}
Since, $\|w_1\|_2\leq C|J^\prime|^{-2/3}$,  we can easily establish
\eqref{eq:130} keeping in mind that 
$$
|\Re(\lambda-\lambda^\#_\ell)|\leq|J^\prime|^{2/3}\,.
$$
\end{proof}

\subsection{Approximate resolvent}

Let $x_0\in \mathcal S^\#$, where $\Sg^\#$ is given by \eqref{eq:2},
and $\Lambda^{\#,1}_{\mathcal K (h)}(x_0,h)$ be given by either
\eqref{eq:125} or \eqref{eq:19}. We seek an estimate for the resolvent
along a circle  in the complex plane, centered at
$\Lambda^{\#,1}_{\mathcal K (h)}(h)$, whose radius is of some suitably
chosen $o(h)$ size.

We now construct once again the partition of unity as in
\eqref{eq:25}. In contrast with the previous section we need to use
two different scales (or ball sizes).  We first split $\Jg_\partial$ into
two disjoint subsets:
\begin{displaymath}
  \Jg_\partial^\perp = \{ j\in \Jg_\partial \,| \, b_j\in\partial\Omega_\perp^\#\}\quad ;\quad \Jg_\partial^o =
  \Jg_\partial\setminus \Jg_\partial^\perp  \,.
\end{displaymath}
We note here that, since $\partial\Omega_\perp^\#$ is a finite set, we
could have easily constructed $\Jg_\partial$ so that
$\Jg_\partial^\perp =\emptyset$. However, since it simplifies the
construction of the approximate resolvent in the vicinity of
$\partial\Omega_\perp^\#$\,, we prefer to select a partition of unity
for which
\begin{equation}
\label{eq:126}
  \Union_{j\in \Jg_\partial^\perp}\{b_j\} = \partial\Omega_\perp^\#  \,. 
\end{equation}
We use different ball sizes for $j\in \Jg_\partial^\perp$ and
  $j\in\Jg_i\cup\Jg_\partial^o$. We proceed in two steps.  We first construct a finite (independent of $h$) partition  of unity (of size $h^{\varrho_\perp}$) $\tilde \xi_{h}, \zeta_{j,h}$ with $j\in \mathcal \Jg^\perp$
   such that
   \begin{equation}
   \tilde \xi_h^2 + \sum_{j\in \mathcal \Jg_\pa^\perp} \zeta_{j,h}^2 = 1 \mbox{ in } \Omega^\#\,,
   \end{equation}
 with $$\zeta_{j,h}\equiv1 \mbox{ in } B(b_j,h^{\varrho_\perp}/2)\,,\,
\zeta_{j,h}\equiv0 \mbox{ in } \Omega\setminus B(b_j,h^{\varrho_\perp})\,,
$$
\begin{equation}
\label{eq:138}
  |\nabla\zeta_{j,h}| + h^{\varrho_\perp}|D^2\zeta_{j,h}|\leq C\,h^{-\varrho_\perp} \,, \;
  \forall j\in\Jg_\partial^\perp\,.
\end{equation} 
and
\begin{equation}
  |\nabla\tilde \xi_{h}| + h^{\varrho_\perp}|D^2 \tilde \xi_{h}|\leq C\, h^{-\varrho_\perp} \,.
  \end{equation}
As in \eqref{eq:90} we introduce an additional condition
  \begin{equation}\label{eq:90a}
  \frac{\pa \zeta_{j,h}}{\pa \nu} \big |_{\pa \Omega^\#}=0 \,,\, \forall j \in
  \mathcal \Jg_\pa^\perp \mbox{ and }  \frac{\pa \tilde \xi_{h}}{\pa \nu} \big |_{\pa \Omega^\#}=0 \,.
  \end{equation}

We now combine the above with the partition of unity (of size $h^\varrho$ with $\varrho < \varrho_{\perp}$) given in
(\ref{eq:25}) and \eqref{part1}: $(\chi_{k,h},\zeta_{k,h})$. Note that
$\chi_{k,h}$ and $\zeta_{k,h}$ are respectively supported on $B(a_j,h^\varrho)$
or $B(b_j,h^\varrho)$. We set
\begin{displaymath}
  \tilde{\zeta}_{k,h} = \zeta_{k,h}\, \tilde \xi_h \quad   \tilde{\chi}_{k,h} = \chi_{k,h}\, \tilde \xi_h \,.
\end{displaymath}
As $\varrho<\varrho_\perp$, for sufficiently small $h$ we must have for some
$k\in\Jg_i$ and $j\in\Jg_\partial^\perp$ that $B(a_k,h^\varrho)\subset B(b_j,h^{\varrho_\perp}/2)$
and hence $\tilde{\chi}_{k,h}\equiv0$. A similar observation can be made for
$k\in\Jg_\partial^o$. We thus define
\begin{displaymath}
  \tilde{\Jg}_i= \{k\in\Jg_i \,| \, \tilde{\chi}_{k,h}\not\equiv0 \} \quad ;
  \quad  \tilde{\Jg}_\partial^o= \{k\in\Jg_\partial^o \,| \, \tilde{\zeta}_{k,h}\not\equiv0 \} \,.
\end{displaymath}
Clearly,
\begin{displaymath}
  \sum_{j\in\Jg_\partial^\perp} \zeta_{j,h}^2 + \sum_{j\in\tilde{\Jg}_\partial^o}
  \tilde{\zeta}_{j,h}^2 +  \sum_{j\in\tilde{\Jg}_i}
  \tilde{\chi}_{j,h}^2 =1 \mbox{ in } \Omega^\# \,.  
\end{displaymath}
For simplicity of notation we drop the tilde accent  in the sequel
and use $(\chi_{j,h},\zeta_{j,h})$ instead of
$(\tilde{\chi}_{j,h},\tilde{\zeta}_{j,h})$ and $(\Jg_i,\Jg_\partial^o)$ instead
of $(\tilde{\Jg}_i,\tilde{\Jg}_\partial^o)$. \\
Note that we deduce of the previous construction that
\begin{subequations}
  \label{eq:137}
  \begin{equation}
    \begin{cases}
   |\nabla\chi_{j,h}|+h^{\varrho_\perp}|D^2\chi_{j,h}|\leq Ch^{-\varrho_\perp} & \text{ in }
B(a_j,h^{\varrho}/2) \\ |\nabla\chi_{j,h}|+h^\varrho|D^2\chi_{j,h}|\leq Ch^{-\varrho}&  \text{ in }
B(a_j,h^{\varrho})  
    \end{cases}
 \; \forall j\in\Jg_i   
\end{equation}
and
\begin{equation}
  \begin{cases}
    |\nabla\zeta_{j,h}|+h^{\varrho_\perp}|D^2\zeta_{j,h}|\leq Ch^{-\varrho_\perp}& \text{ in }
B(b_j,h^{\varrho}/2) \\ |\nabla\zeta_{j,h}|+h^{\varrho}|D^2\zeta_{j,h}|\leq Ch^{-\varrho} &
\text{ in }
B(b_j,h^{\varrho}) 
  \end{cases}
\; \forall j\in\Jg_\partial^o\,.
\end{equation}
\end{subequations}
As in the previous section, we keep the property that each point of
$\Omega^\#$ belongs to at most $N_0$ balls with $N_0$ independent of $h$
and the inequality  
\begin{equation}
\sum_j |\pa^\alpha\chi_{j,h}(x)|^2 + 
\sum_j |\pa^\alpha\zeta_{j,h}(x) |^2 \leq C_\alpha  \, h^{-2 |\alpha|{\varrho}}\,,
\end{equation}
and introduce $\eta_{j,h} = 1_{\Omega^\#}\, \zeta_{j,h}\,. $\\
Note that, as a result of \eqref{eq:90} and \eqref{eq:90a}, we have
\begin{equation}\label{eq:139}
  \frac{\pa \zeta_{j,h}}{\pa \nu} \big |_{\pa \Omega^\#}=0 \,,\, \forall j \in \mathcal \Jg_\pa\,.
  \end{equation}

As a consequence of future constraints we impose from now on the following conditions:
\begin{equation}
\label{eq:87}
0 < \frac 13 < \varrho_{\perp}  < \frac 25 \mbox{ and } \frac 12 <  \varrho < \frac 59 \,.
\end{equation}

We further split $\Jg_\partial^\perp$ into the subsets
\begin{displaymath} 
  \Jg_\partial^m = \{ j\in \Jg_\partial^\perp \,| \, b_j\in\partial\Omega_\perp^\# \;;\;
  V(b_j)=V(x_0) \; ; \; \lambda^\#(|\nabla
  V(b_j)|)=\Lambda_m^\#\; ; \;  \lambda_1^\#(b_j)=\lambda_1^\#(x_0)\}
\end{displaymath}
and
\begin{displaymath}
  \Jg_\partial^c =   \Jg_\partial^\perp\setminus \Jg_\partial^m   \,.
\end{displaymath}
Recall again that 
\begin{displaymath}
  \Union_{j\in \Jg_\partial^m}\{b_j\} = \Sg^\#  \,.
\end{displaymath}

We now construct the approximate resolvent in a similar, though
slightly different, manner to the one used in the previous section.
We obtain bounds for the approximate operators $\A_{j,h}$ defined in
\eqref{defCAjh}, and $\tilde{A}_{j,h}^\#$, which for
$j\in\Jg_\partial^0$ is defined by \eqref{eq:13}.  A separate
definition of $\tilde{A}_{j,h}^\#$ is necessary for $j\in
\Jg_\partial^m$. For $j\in\Jg_i$, \eqref{resCAjh} still holds, whereas
for $j\in\Jg_\partial^m$, we need to refine the estimates of Section
\ref{s6}. Thus, for $j\in\Jg_\partial^m$,  we set 
\begin{equation}   \label{eq:34}
\begin{cases}
 \tilde\A_{j,h}^\#= -h^2\Delta_{s,\rho}+i\big(V(b_j)+
  J(b_j)\rho+ s\cdot D^2_sV_\partial(b_j)s\big)\,,\\
\Dg(\tilde \A_{j,h}^\#) = \{ u\in L^2(\R^n_\#)\cap \mathfrak D^\#_n\, | \, \tilde\A_{j,h}u
\in L^2(\R^n_\#)  \} \,,
\end{cases}
\end{equation}
where the curvilinear coordinates $(s,\rho)$ are given by
\eqref{eq:38} and $\mathfrak D^\#_n$ is given by \eqref{eq:102}. Note
that $\tilde\A_{j,h}^\#$ depends on $\K (h) =\kappa h^\frac 43$ when
$\# \in \{R, T\flat\}$ and $b_j \in \partial \Omega_-$.\\
We now set 
\begin{equation*}
r(h)=h^q\,,
\end{equation*} 
for some $q$ satisfying
\begin{equation}\label{condq}
0 < q < \frac{1}{15}\,.
\end{equation}
Applying to \eqref{eq:34} the dilation
\begin{equation}     \label{eq:96}
\rho=\Big[\frac{h^2}{J(b_j)}\Big]^{1/3}\tau \quad ; \quad s=h^{1/2}\sigma \,,
\end{equation}
we obtain the operator $$ i V (b_j(h)) + h^{2/3} J(b_j)^{-\frac 13}
\B_{\varepsilon_j}^\#\,,$$
where $\B_\varepsilon^\#$ is defined in
\eqref{eq:123}, $$\varepsilon_j =h^\frac 23 J(b_j) ^\frac 23$$ and (if
appropriate) the Robin or Transmission parameter   $$\kappa_j^*= \kappa
J(b_j)^{-\frac 13}\,.$$
  By \eqref{eq:3} we then have
\begin{equation}     \label{eq:27} 
\sup_{j\in \Jg_\partial^m}\sup_{\lambda\in\partial
  B(\Lambda^{\#,1},rh)}\|(\tilde{\A}_{j,h}^\# -\lambda)^{-1}\|\leq \frac{C}{rh}=
\frac{C}{h^{1+q}}\,,
\end{equation}
where $\Lambda^{\#,1}$ is given by \eqref{eq:125} and \eqref{eq:19}. \\
{\bf The case $j\in \Jg_\partial^o$.\\}
Here we use \eqref{eq:13} for the definition of $\tilde{\A}_{j,h}$ in
this case. By the smoothness of $\partial\Omega$ and $V$, and by
Assumption \ref{nondeg2}, there exists $C>0$ such that
\begin{displaymath} 
  | \nabla V - (\vec \nu\cdot \nabla V)\vec \nu |(b_j) \geq C\, d(b_j,\partial\Omega_\perp^\#)\,.
\end{displaymath}
Furthermore, as $\lambda^\#_1$ is a simple eigenvalue $\LL^\#$, we
have, by either \eqref{eq:76} or \eqref{eq:77}, for
$\lambda\in\partial B(\Lambda^{\#,1},rh)$,
\begin{displaymath}
  |\lambda^\#(J(b_j)\cdot\nu)-\lambda| \leq C\, d(b_j,\partial\Omega_\perp^\#) 
\leq \hat C \,   | \nabla V - (\vec \nu\cdot \nabla V)\vec \nu |(b_j) \,.
\end{displaymath}
Hence, we may use \eqref{eq:11} together with \eqref{eq:96} to
establish that
\begin{equation}   \label{eq:36}
  \sup_{\lambda\in\partial B(\Lambda^{\#,1},rh)}\|(\tilde{\A}_{j,h}^\# -\lambda)^{-1}\|\leq
  \frac{C}{[d(b_j,\partial\Omega_\perp^\#)h]^{\frac{2}{3}}}\leq
  \frac{\widehat C}{h^{\frac{2}{3}(1+\varrho_\perp)}}\,. 
\end{equation}
Note that by \eqref{eq:126} there exists $C>0$ such that
$d(b_j,\partial\Omega_\perp^\#) \geq \frac 1C h^{\varrho_\perp}$ for all
$j\in \Jg_\partial^o$ and sufficiently small $h$. \\

{\bf The case  $j\in \Jg_\partial^c$.\\}
In this case we need to consider three different subsets:
\begin{enumerate}
\item $\Jg_\partial^{c,1}=\{j\in \Jg_\partial^c \, | \, V(b_j)\neq V(x_0)\,\}$, 
\item $\Jg_\partial^{c,2}=\{j\in \Jg_\partial^c \, | \,
  V(b_j)=V(x_0)\,,\;\lambda^\#(J(b_j))>\Lambda^\#_m,\}$,
\item $\Jg_\partial^{c,3}=\{j\in \Jg_\partial^c \, | \,
  V(b_j)=V(x_0)\,,\;\lambda^\#(J(b_j))=\Lambda^\#_m,\,, \; 
    \mu_1^r(b_j)>\mu_1^r(x_0)\,\}$,
\end{enumerate}
where $\mu_1^r$ is given by \eqref{eq:29}, in which the
$\{\alpha_j\}_{j=1}^{n-1}$ are the eigenvalues of $D^2V_\partial$, where $V_\partial$
denotes the restriction of
$V$ to $\partial\Omega^\#$ (see Assumption \ref{nondeg2}).

{\bf For the first subset} we use \eqref{eq:13} as the definition of
$\tilde{\A}_{j,h}^\#$.  Using \eqref{eq:30} upon applying
\eqref{eq:96} leads for $h$ small enough to the following estimate
\begin{equation}  \label{eq:32}
\max_{j\in \Jg_\partial^{c,1}}\sup_{\lambda\in\partial B(\Lambda^{\#,1},\, h^{1+q})}\|(\tilde{\A}_{j,h}^\# -\lambda)^{-1}\| \leq \frac{C}{h^{2/3}} \,.
\end{equation}
{\bf For the second subset, } i.e. for $j\in \Jg_\partial^{c,2}$ we
may use \eqref{eq:15}, with
\begin{equation*}
\epsilon= \epsilon^{c,2}=\min_{j\in \Jg_\partial^{c,2}} \{ J(b_j)-J(x_0) \}\,,
\end{equation*}
to obtain that there exists $C>0$ such that
\begin{equation}    \label{eq:33}
\max_{j\in \Jg_\partial^{c,2}}\sup_{\lambda\in\partial B(\Lambda^{\#,1},\,h^{1+q})}\|(\tilde{\A}_{j,h}^\# -\lambda)^{-1}\| 
\leq \frac{C}{\epsilon^{c,2}\, h^{2/3}} \,.
\end{equation}
{\bf Finally, for the third subset}, i.e. for $j\in
\Jg_\partial^{c,3}$, we use \eqref{eq:34} as the definition of
$\tilde{A}_{j,h}^\# $, together with \eqref{eq:24} and a dilation to
obtain
\begin{equation}   \label{eq:32a}
\sup_{j\in \Jg_\partial^{c,3}}\sup_{\lambda\in\partial B(\Lambda^{\#,1},\,h^{1+q})}\|(\tilde{\A}_{j,h}^\# -\lambda)^{-1}\|\leq \frac{C}{h}\,.
\end{equation}

We use \eqref{defResApp1} for the approximate resolvent, whose
definition is repeated here for the convenience of the reader,
\begin{displaymath}
 \mathcal{R}(h,\lambda) = \sum_{j\in \mathcal \Jg_i(h)}\chi_{j,h}(\A_{j,h}-\lambda)^{-1}\chi_{j,h} 
+  \sum_{j\in \mathcal \Jg_\partial (h)}\eta_{j,h}  R_{j,h}\eta_{j,h}\,,
\end{displaymath}
where the definition of $R_{j,h}$ is given in \eqref{defResBord} with
$\lambda(h,\nu)$ replaced by some $\lambda\in\partial B(\Lambda^{\#,1},h^{1+q})\,$. The
boundary transformation $\mathcal{F}_{b_j}$ and its associated
operator $T_{\mathcal{F}_{b_j}}$, needed in \eqref{defResBord}, are
respectively given by \eqref{eq:38} and \eqref{transfoF}. Let, as in
the previous section,
\begin{displaymath}
  \Eg(h,\lambda) = (\A_h^\# -\lambda)\mathcal{R}(h,\lambda) - I\,.
\end{displaymath}
Then, as in \eqref{decompRlambda}, we get
\begin{multline}    \label{eq:44}
 \Eg(h,\lambda) = \sum_{j\in \mathcal \Jg_i(h)}\chi_{j,h}(\A_h-\A_{j,h})(\A_{j,h}-\lambda)^{-1}\chi_{j,h} 
 + [\A_h,\chi_{j,h}](\A_{j,h}-\lambda )^{-1}\chi_{j,h}  \\
\quad +  \sum_{j\in \mathcal \Jg_\partial (h)} \left( \eta_{j,h} T_{\mathcal{F}_{b_j}}^{-1}
(\widehat {\A}_{j,h}-\tilde{\A}_{j,h}) (\tilde{\A}_{j,h}-\lambda)^{-1} T_{\mathcal{F}_{b_j}}\eta_{j,h}  
 + [\A_h,\eta_{j,h}]R_{j,h}\eta_{j,h} \right)\,.
\end{multline}
In the sequel we  prove the following
\begin{lemma}
Under Assumptions \ref{notdeg} and \ref{nondeg2}, for $q$, $\varrho_\perp$, and
$\varrho$ satisfying  \eqref{eq:87} and \eqref{condq},  we have with  $r=h^q$, 
\begin{equation}
\label{eq:46}
    \lim_{h\to0} \sup_{\lambda\in\partial B(\Lambda^{\#,1},rh)} \| \Eg(h,\lambda)\| =0\,.
\end{equation}
\end{lemma}
\begin{proof}
As in \eqref{eq:59} we may write
\begin{equation}
  \label{eq:61}
\|\Eg(h,\lambda)f\|_2^2\leq C_0 \, \left(  \sum_{j\in \mathcal J_i}  \|\mathcal B_j\|^2 \| \chi_{j,h}  f\|^2_2\,
   + \sum_{j\in \mathcal J_\partial}  \|\mathcal B_j\|^2 \| \eta_{j,h}  f\|^2_2\right)\,,
\end{equation}
where $\B_j$ is given by \eqref{eq:6} with $\lambda(h,\nu)$ replaced
by some $\lambda\in\partial B(\Lambda^{\#,1},h^{1+q})\,$.\\

For $j\in \Jg_i$, the estimates are unchanged in this new partition and  we can still write
\begin{displaymath}
  \|\chi_{j,h}(\A_h-\A_{j,h})(\A_{j,h}-\lambda)^{-1}\chi_{j,h}\|\leq Ch^{2(\varrho-1/3)}
\end{displaymath}
and 
\begin{displaymath}
  \|[\A_h,\chi_{j,h}](\A_{j,h}-\lambda )^{-1}\chi_{j,h}\|\leq Ch^{2/3-\varrho}\,.
\end{displaymath}
Consequently,
\begin{equation}   \label{eq:41}
\sum_{j\in \mathcal J_i}  \|\mathcal B_j\|^2 \| \chi_{j,h} f\|^2_2\, \leq C \,
(h^{4/3-2\varrho}+
h^{4(\varrho-1/3)})\|f\|_2^2
\end{equation}

For $j\in \Jg_\partial^{c,1}\cup \Jg_\partial^{c,2}$, we
follow the same steps as in the proof of \eqref{estRHS_R1bdry} to
obtain, with the aid of \eqref{eq:32}, \eqref{eq:33}, and
\eqref{eq:138} that
\begin{equation}   \label{eq:35} 
\sum_{j\in \Jg_\partial^{c,1}\cup \Jg_\partial^{c,2}} \|\mathcal B_j\|^2 \| \eta_{j,h}
f\|^2_2\leq C \,  (h^{2(\varrho_\perp-1/3)} + h^{4/3-2\varrho_\perp}) \|f\|_2^2\,.
\end{equation}
To complete the proof of the lemma we need yet an estimate for the
terms in \eqref{eq:44} for which $j\in
\Jg_\partial^{c,3}\cup\Jg_\partial^m\cup\Jg_\partial^o$.
 
{\bf The case $j \in \Jg_\partial^m$. } \\
For $j\in \Jg_\partial^m$ we write (using the methods in \cite{Hen2}
while taking account the regularity of the $\#$-realization for the
Laplacian in $\mathbb R^n_\#$)
\begin{multline}    \label{eq:49}
  \big\|\eta_{j,h}T_{\mathcal{F}_{b_j}}^{-1}
  (\widehat{\A}_{j,h}^\#-\tilde{\A}_{j,h}^\#) (\tilde{\A}_{j,h}-\lambda)^{-1}
  T_{\mathcal{F}_{b_j}}\hat{\eta}_{j,h}\big\|\\
  \leq C \Big(h^{\varrho_\perp} \|h^2\Delta_{(s,\rho)}(\tilde\A_{j,h}^\#-\lambda)^{-1} \breve{\eta}_{j,h}\| +
   \|h^2\nabla_{(s,\rho)}(\tilde\A_{j,h}^\#-\lambda(h))^{-1}\| \\ + 
 \| (V\circ \mathcal F_{b_j} -V_{b_j} ^{(2))\,}\tilde{\eta}_{j,h} (\tilde\A_{j,h}^\#-\lambda)^{-1}\breve{\eta}_{j,h}\|\Big) \,,
\end{multline}
where   
\begin{equation*}
V_{b_j} ^{(2)}(s,\rho))=V(b_j)+J(b_j)\rho+s\cdot
D^2_sV_\partial(b_j)s\,,
\end{equation*}
and 
\begin{displaymath} 
  \breve{\eta}_{j,h} = T_{\mathcal{F}_{b_j}}\hat{\eta}_{j,h}\,,\quad \tilde \eta_{j,h} = T_{\mathcal{F}_{b_j}} \eta_{j,h}\,.
\end{displaymath}
For the first term on the right-hand-side we use dilation together
with (\ref{eq:23}c), Lemma \ref{lemma7.1}, and the localization of the
support of $\breve{\eta}_{j,h}$, 
\begin{equation}  \label{eq:50} 
h^{\varrho_\perp}\|h^2\Delta_{(s,\rho)}(\tilde\A_{j,h}^\#-\lambda)^{-1} \breve{\eta}_{j,h} \| \leq C\, h^{\varrho_\perp-1/3-q} \,.
\end{equation}
The second term can be bounded via dilation and (\ref{eq:23}b),
yielding
\begin{equation}  \label{eq:51} 
\|h^2\nabla_{(s,\rho)}(\tilde\A_{j,h}^\#-\lambda(h))^{-1}\| \leq C\, h^{1/6-q} \,.
\end{equation}
Finally, to bound the last term on the right-hand-side of
\eqref{eq:49} we first write, for some $a>0$,
\begin{multline}   \label{eq:53}
  \| (V\circ \mathcal F_{b_j} -V_{b_j} ^{(2)}) \tilde{\eta}_{j,h}(\tilde\A_{j,h}^\#-\lambda)^{-1}\breve{\eta}_{j,h} \| \leq   
\| (V\circ \mathcal F_{b_j} -V_{b_j} ^{(2)}) {\mathbf
  1_{\rho\geq h^{2/3-a}}}\tilde{\eta}_{j,h}(\tilde\A_{j,h}^\#-\lambda)^{-1}\breve{\eta}_{j,h}\| \\ + 
\| (V\circ \mathcal F_{b_j} -V_{b_j} ^{(2)}) \mathbf
 1_{\rho<h^{2/3-a}}\tilde{\eta}_{j,h}(\tilde\A_{j,h}^\#-\lambda)^{-1}\breve{\eta}_{j,h}\|\,.
\end{multline}
We use \eqref{eq:3} together with (\ref{eq:23}a), to obtain via
dilation,
\begin{equation}   \label{eq:52} 
  \| (V\circ \mathcal F_{b_j} -V_{b_j} ^{(2)})\, {\mathbf
  1_{\rho\geq h^{2/3-a}}}\tilde{\eta}_{j,h}(\tilde\A_{j,h}^\#-\lambda)^{-1}\breve{\eta}_{j,h}\|_2 \leq C\, h^{2(\varrho_\perp-1/3)}\,.
\end{equation}
For the second term on the right-hand-side,
\begin{displaymath} 
   \| (V\circ \mathcal F_{b_j} -V_{b_j} ^{(2)} ){\mathbf
  1_{\rho<h^{2/3-a}}}\breve{\eta}_{j,h}\|_\infty \leq C \, (h^{3\varrho_\perp}+h^{2/3+\varrho_\perp-a}+ h^{\frac 43 -2a})\,.
\end{displaymath}
Hence, using Lemma \ref{lemma7.1},
\begin{displaymath}
  \| (V\circ \mathcal F_{b_j} -V_{b_j} ^{(2)} ){\mathbf
  1_{\rho<h^{2/3-a}}}\tilde{\eta}_{j,h}(\tilde\A_{j,h}^\#-\lambda)^{-1}\breve{\eta}_{j,h}\|\leq 
C(h^{3\varrho_\perp-1-q}+h^{\varrho_\perp-1/3-a-q} + h^{\frac 13 -2a -q}) \,.  
\end{displaymath}\
Substituting the above together with \eqref{eq:52} into \eqref{eq:53}
then yields, for any $(a,\varrho_\perp,q)$ such that  \eqref{eq:87}, \eqref{condq},  and
 \begin{equation}\label{condarhoq}
 0<a<\varrho_\perp-1/3-q\,,
\end{equation}
are satisfied, 
the upper bound
\begin{displaymath}
\| (V\circ \mathcal F_{b_j} -V_{b_j} ^{(2)}) \, \tilde{\eta}_{j,h}(\tilde\A_{j,h}^\#-\lambda)^{-1}\breve{\eta}_{j,h}\|
\leq  C\,h^{\varrho_\perp-1/3-a-q}\,.
\end{displaymath}
The above together with \eqref{eq:51} and \eqref{eq:50} yield, when
substituted into \eqref{eq:49},
\begin{equation}   \label{eq:55} 
 \sup_{j\in \Jg_\partial^m}\Big\|\eta_{j,h} T_{\mathcal{F}_{b_j}}^{-1}
  (\widehat{\A}_{j,h}-\tilde{\A}_{j,h}) (\tilde{\A}_{j,h}-\lambda)^{-1}
  \breve{\eta}_{j,h}\Big\|
  \leq C\, h^{\varrho_\perp-1/3-a-q}\,.
\end{equation}

To estimate the entire contribution of $\Jg_\partial^m$ to the second sum in
\eqref{eq:44} we need yet the following bound, which is obtained with
the aid of (\ref{eq:23}a), (\ref{eq:23}d), \eqref{eq:88}, and a dilation:
\begin{multline*}
   \sup_{j\in \Jg_\partial^m} \Big\|[\A_h,\eta_{j,h}]R_{j,h}\eta_{j,h}  \Big\| \leq
 C \sup_{j\in \Jg_\partial^m} \,\big(h^{2(1-\varrho_\perp)} \|(\tilde\A_{j,h}^\#-\lambda)^{-1}\|
    \\ +h^{2-\varrho_\perp} \big[\|{\mathbf
     1_{\rho>h^{2/3-a}}}\nabla_{s,\rho}(\tilde\A_{j,h}^\#-\lambda)^{-1}\|+  \|{\mathbf
     1_{s>\frac{h^{\rho_\perp}}{2}}}\nabla_s(\tilde\A_{j,h}^\#-\lambda)^{-1}\|\big] \\+h^{8/3-a-2\varrho_\perp}\|{\mathbf
     1_{s>\frac{h^{\rho_\perp}}{2}}}\nabla_\rho(\tilde\A_{j,h}^\#-\lambda)^{-1}\| \big) \leq
   C\, h^{\frac{2}{3}(1-2\varrho_\perp)-q} \,,
\end{multline*}
where use has been made of the fact that by
\eqref{eq:138} and \eqref{eq:90a}, we have,  for $\rho<h^{2/3-a}$, the inequality
$$
 |\partial\eta_{j,h}/\partial\rho|\leq Ch^{2/3-a-\varrho_\perp}\,.
$$ 

Combining the above with \eqref{eq:55} then yields that for
$(a,q,\varrho_\perp)$ satisfying \eqref{condarhoq}, we have
\begin{equation}     \label{eq:56}
\sum_{j\in \Jg_\partial^m}\|\B_j\|^2 \| \eta_{j,h}
f\|^2_2\leq C\, \big(h^{\frac{2}{3}(1-2\varrho_\perp)-q}+h^{\varrho_\perp-1/3-a-q}\big)\|f\|_2^2 \,.
\end{equation}

In a similar manner, with the aid of \eqref{eq:32a}, we establish
\begin{equation}   \label{eq:57}
\sum_{j\in\Jg_\partial^{c,3}}\|\B_j\|^2 \| \eta_j
f\|^2_2\leq C\, \big(h^{\frac{2}{3}(1-2\varrho_\perp)}+h^{\varrho_\perp-1/3-a}\big)\|f\|_2^2  \\.
\end{equation}

{\bf The case $j\in \Jg_\partial^o$.}\\
Using the fact that
\begin{displaymath} 
  \| (V\circ \mathcal F_{b_j} -V_{b_j} ^{(2)}) \breve{\eta}_{j,h}\|_\infty\leq C\,h^{2\varrho}\,,
\end{displaymath}
we get, as in \eqref{eq:49},
\begin{multline}  \label{eq:45}
  \big\|\eta_{j,h} T_{\mathcal{F}_{b_j}}^{-1}
  (\widehat {\A}_{j,h}^\#-\tilde{\A}_{j,h}^\# ) (\tilde{\A}_{j,h}^\# -\lambda)^{-1}
  \breve{\eta}_{j,h}\big\|
  \leq C\, \Big(h^\varrho\|h^2\Delta_{(s,\rho)}(\tilde\A_{j,h}^\#-\lambda)^{-1} \breve{\eta}_{j,h} \| \\+
  \|h^2{\nabla_{(s,\rho)}}(\tilde\A_{j,h}^\#-\lambda(h))^{-1}\| + 
h^{2\varrho}\|\tilde{\eta}_{j,h}(\tilde\A_{j,h}^\#-\lambda)^{-1}\breve{\eta}_{j,h}\|\Big) \,.
\end{multline}
 We may now use \eqref{eq:36} to obtain that
\begin{equation}  \label{eq:47}
h^{2\varrho}\|\tilde{\eta}_{j,h}(\tilde\A_{j,h}^\#-\lambda)^{-1}\breve{\eta}_{j,h}\|\leq
C\, h^{2\varrho-\frac{2}{3}(\varrho_\perp+1)}\,,
\end{equation}
which is small as $h\to0$, in view of \eqref{eq:87}.\\
As in \cite[Eq. (4.24)]{Hen2} (or see the proof of \eqref{eq:50} with
\eqref{eq:3} replaced by \eqref{eq:36}) we then obtain, with the aid
of \eqref{eq:87}
\begin{equation}  \label{eq:54} 
h^\varrho\|h^2 {\Delta_{(s,\rho)}}(\tilde\A_{j,h}^\# -\lambda)^{-1}  \breve{\eta}_{j,h}
\| \leq C\, h^{\varrho-2\varrho_\perp/3} \leq C\, h^{(1-\varrho_\perp)/3} \,.
\end{equation}
Similarly, 
\begin{displaymath}
  \|h^2{\nabla_{(s,\rho)}}(\tilde\A_{j,h}^\# -\lambda)^{-1}\| \leq C\, h^{
    \frac{4}{3}(1-\varrho_\perp)}\,.
\end{displaymath}
Substituting the above, together with \eqref{eq:54} and \eqref{eq:47},
into \eqref{eq:45} yields
\begin{equation}  \label{eq:58}
 \sup_{j\in \Jg_\partial^o}\big\|\eta_{j,h} T_{\mathcal{F}_{b_j}}^{-1}
  (\widehat {\A}_{j,h}^\#-\tilde{\A}_{j,h}^\#) (\tilde{\A}_{j,h}^\# -\lambda)^{-1}
  \breve{\eta}_{j,h}\big\| \leq C\, h^{2\varrho-\frac{2}{3}(\varrho_\perp+1)}\,. 
\end{equation}

We now estimate the rest of the contribution of $\Jg_\partial^o$, as in
\eqref{eq:56}, by writing first
\begin{multline}
\label{eq:129}
  \big\|[\A_h,\eta_{j,h}]R_{j,h}\eta_{j,h}  \big\| \leq
   C\Big(h^{2(1-\varrho)}\|(\tilde\A_{j,h}^\# -\lambda)^{-1}\|+ h^{2-\varrho_\perp}\|{\mathbf
     1}_{[s^2+\rho^2]^{1/2}\leq h^\varrho/4}\nabla_{s,\rho}(\tilde\A_{j,h}^\# -\lambda)^{-1}\|
    \\  + h^{2-\varrho}\|{\mathbf
     1_{\rho>h^{2/3-a}}}\nabla_{s,\rho}(\tilde\A_{j,h}^\# -\lambda)^{-1}\|+
 h^{2-\varrho}\|{\mathbf
     1_{\rho<h^{2/3-a}}}\nabla_s(\tilde\A_{j,h}^\# -\lambda)^{-1}\| \\ + h^{8/3-a-2\varrho}\|{\mathbf
     1}_{[s^2+\rho^2]^{1/2}>h^\varrho/4}{\mathbf
     1_{\rho<h^{2/3-a}}}\partial_\rho(\tilde\A_{j,h}^\# -\lambda)^{-1}\| \Big)  \,,
\end{multline}
where use have been made  again of the fact that, by \eqref{eq:138} and \eqref{eq:139}, for
$\rho<h^{2/3-a}$ and $[s^2+\rho^2]^{1/2}>h^\varrho/4$ in coordinates $(s,\rho)$ centered at $b_j$, we must have
$$
|\partial\eta_{j,h}/\partial\rho|\leq Ch^{2/3-a-2\varrho}\,.
$$
Note that for sufficiently small $h$,  
\begin{displaymath}
  [s^2+\rho^2]^{1/2}\leq h^\varrho/4\Rightarrow x(s,\rho)\in B(b_j,h^\varrho/2)
\end{displaymath}
and hence, by (\ref{eq:137}a), 
$$
|{\mathbf  1}_{[s^2+\rho^2]^{1/2}\leq h^\varrho/4}\nabla\eta_{j,h}|\leq Ch^{-\varrho_\perp}\,.
$$

We begin the estimate of the right-hand-side of \eqref{eq:129}
  by observing that, in view of \eqref{eq:36},
  \begin{equation}
    \label{eq:132}
h^{2(1-\varrho)}\|(\tilde\A_{j,h}^\#
-\lambda)^{-1}\|\leq C\,h^{2\Big[\frac{2-\varrho_\perp}{3}-\varrho\Big]}\,.
  \end{equation}
The second term on the right-hand-side can be estimated using
\eqref{eq:36} and dilation
\begin{equation}
  \label{eq:133}
h^{2-\varrho_\perp}\|{\mathbf
  1}_{[s^2+\rho^2]^{1/2}\leq h^{\varrho}/4}\nabla_{s,\rho}(\tilde\A_{j,h}^\#
-\lambda)^{-1}\|\leq C\,h^{2/3-5\varrho_\perp/3}\,.
\end{equation}

Let $\delta>0$ be as in Lemma \ref{lem:oblique}. Then, there exists
$C_0>0$ such that if \break $d(b_j,\partial\Omega_\perp^\#)\leq
C_0\delta$, we may use \eqref{eq:42} to obtain
that
\begin{equation}
  \label{eq:134}
h^{2-\varrho}\|{\mathbf
     1_{\rho>h^{2/3-a}}}\nabla_{s,\rho}(\tilde\A_{j,h}^\# -\lambda)^{-1}\|\leq C\, h^{4/3-\varrho}\,.
\end{equation}
Furthermore, by \eqref{eq:130} we have, as $|J^\prime|\geq C\, h^{\varrho_\perp}$,
\begin{equation}
  \label{eq:135}
h^{2-\varrho}\|{\mathbf
     1_{\rho<h^{2/3-a}}}\nabla_s(\tilde\A_{j,h}^\# -\lambda)^{-1}\| \leq C\, h^{\frac{2-\varrho_\perp}{3}-\varrho}\,.
\end{equation}
Finally, for the last term on the right-hand-side of \eqref{eq:129} we
have
\begin{displaymath}
  h^{8/3-a-2\varrho}\|{\mathbf
     1_{\rho<h^{2/3-a}}}\partial_\rho(\tilde\A_{j,h}^\# -\lambda)^{-1}\| \leq C\, h^{2\Big[\frac{2-\varrho_\perp}{3}-\varrho\Big]-a}\,.
\end{displaymath}
Substituting the above together with \eqref{eq:135}, \eqref{eq:134},
\eqref{eq:133}, and \eqref{eq:132} into \eqref{eq:129} yields,
choosing $a<(2-\varrho_\perp)/3-\varrho$
\begin{displaymath}
 \sup_{  \begin{subarray}{c}
    j\in \Jg_\partial^o \\
   d(b_j,\partial\Omega_\perp^\#)\leq C_0\, \delta 
  \end{subarray}}\Big\|[\A_h,\eta_{j,h}]R_{j,h}\eta_{j,h}  \Big\|   \leq C\,
  h^{\frac{2-\varrho_\perp}{3}-\varrho} \,.
\end{displaymath}
Otherwise, if $d(b_j,\partial\Omega_\perp^\#)>C_0\,\delta$, then, since
in that case $|\nabla V - (\vec \nu \cdot\nabla V)\vec \nu |(b_j)\geq
C\, \delta$, we may use \eqref{eq:28} to obtain that
\begin{displaymath}
 \sup_{  \begin{subarray}{c}
    j\in \Jg_\partial^o \\ 
   d(b_j,\partial\Omega_\perp^\#)>C_0\, \delta 
  \end{subarray}}
\big\|[\A_h,\eta_{j,h}]R_{j,h}\eta_{j,h}  \big\|   \leq
   C\, h^{2/3-\varrho} \,,
\end{displaymath}
and hence,
\begin{displaymath}
  \sum_{j\in\Jg_\partial^o}\|\B_j\|^2 \| \eta_{j,h}
f\|^2_2\leq C\,  h^{\frac{2-\varrho_\perp}{3}-\varrho}\|f\|_2^2 \,.
\end{displaymath}
Substituting the above together with 
\eqref{eq:57}, \eqref{eq:56}, \eqref{eq:35}, and \eqref{eq:41}, into
\eqref{eq:61} yields
\begin{displaymath}
  \|\Eg(h,\lambda)\| \leq C\, ( h^{\frac{2-\varrho_\perp}{3}-\varrho}+ h^{
    \frac{2}{3}(1-2\varrho_\perp)-q}+ h^{ 
    \varrho_\perp -1/3-a-q}+h^{2\varrho-\frac{2}{3}(\varrho_\perp+1)})\,,
\end{displaymath}
for $(a,\varrho,q)$ satisfying \eqref{condarhoq}. Assumptions \eqref{eq:87} and an appropriate
choice of $a$ then complete the proof of the lemma.
\end{proof}

From  \eqref{eq:46} it follows that $(I+\Eg(h,\lambda))^{-1}$ is
uniformly bounded as $h\to0$. Consequently, as in Section \ref{s6},
using this time the estimates \eqref{eq:27}, \eqref{eq:36},
\eqref{eq:32}, \eqref{eq:33} and \eqref{eq:32a}, we obtain the
existence of $h_0$ and $C>0$, such that for $h\in (0,h_0]$ and
$\lambda\in\partial B(h^{2/3}\lambda_1^\# +h\mu_1,h^{1+q})$
\begin{displaymath}
  \|(\A_h^\# -\lambda)^{-1}\|\leq C\, h^{-(1+q)}\,.
\end{displaymath}
We have thus established the
following result:
\begin{proposition}
Under Assumptions \ref{notdeg} and \ref{nondeg2}, 
there exist for $q\in (0,\frac{1}{15})$ positive constants $C$ and $h_0$ such that, for all $h\in
(0,h_0]$,
\begin{equation}
    \label{eq:60}
  \sup_{\lambda\in\partial B(iV(x_0)+h^{2/3}\lambda^\#_1+h\mu_1,h^{1+q})}\|(\A_h^\# -\lambda)^{-1}\|\leq C\, h^{-(1+q)}\,.
\end{equation}
\end{proposition}
 
We can now prove the upper bound for the spectrum. 
\begin{proposition}
There exist $h_0 >0$ and, for $h\in (0,h_0]$, an eigenvalue $\lambda
\in \sigma(\A_h^\#)$ satisfying
\begin{equation}  \label{eq:5bis} 
\lambda-  iV(x_0)- \lambda_1^\# ( |\nabla V (x_0)|,\kappa)h^\frac 23 -\mu_1(x_0)
h  =  o(h) \quad \text{as }h\to0\,.
\end{equation}
\end{proposition}
\begin{proof}
Let $U^{\#,1}$ be given by either (\ref{eq:125}b) or (\ref{eq:19}b),
and let $f=(\A_h^\#-\Lambda^{\#,1})U^{\#,1}$. Clearly,
\begin{displaymath}
  (\A_h^\# -\lambda)U^{\#,1} = f + (\Lambda^{\#,1}-\lambda)U^{\#,1}\,.
\end{displaymath}
Hence
\begin{displaymath}
  \langle U^{\#,1}, (\A_h^\# -\lambda)^{-1}U^{\#,1}\rangle =-\frac{1}{\lambda-\Lambda^{\#,1}}[1-\langle U^{\#,1}, (\A_h^\# -\lambda)^{-1}f\rangle] \,.
\end{displaymath}
By \eqref{eq:60} and either \eqref{eq:125} or \eqref{eq:19}, we then
obtain that
\begin{displaymath}
  \|(\A_h^\# -\lambda)^{-1}f\|_2\leq \frac{C}{h^{1+q}}\|f\|_2\leq C\, h^{1/6-q}\,.
\end{displaymath}
Consequently,
\begin{displaymath}
  \Big|\frac{1}{2\pi
    i}\oint_{\partial B(iV(x_0)+\lambda_1^\# ( |\nabla V (x_0)|,\kappa)h^\frac 23+h\mu_1,h^{1+q})}\langle U^{\#,1},
  (\A_h^\# -\lambda)^{-1}f\rangle \,  d\lambda+1\Big|\leq C\, h^{1/6-q} \,.
\end{displaymath}
Hence there exists $h_0 >0$ such that, for $h\in (0,h_0]$, $(\A_h^\#
-\lambda)^{-1}$ is not holomorphic in $B(iV(x_0)+ \lambda_1^\# (
|\nabla V (x_0)|,\kappa) h^\frac 23 +h\mu_1,h^{1+q})$ and the
proposition is proved.
\end{proof}

\appendix
\section{An integral estimate}
\label{appB}

In the following we prove the following estimate, which is useful 
in Section \ref{sModels}.
\begin{lemma}
Let $\alpha$ and $\beta$ be positive. Then,
\begin{equation}    \label{eq:64a}
  \int_0^\infty \exp (-\alpha t^3 + \beta t) \,dt \leq
\frac{\sqrt{\pi}}{(3\beta\alpha)^{1/4}}\exp \biggl(\frac{\beta^{3/2}}{(3\alpha)^{1/2}} \biggr) \,.
\end{equation}
\end{lemma}
\begin{proof}
Let
\begin{displaymath}
    I=\int_0^\infty \exp(-\alpha t^3 + \beta t) \,dt \,.
\end{displaymath}
Let $t=(\beta/\alpha)^{1/2}\, \tau$. Then we have
\begin{displaymath}
I = \sqrt{\frac{\beta}{\alpha}}\int_0^\infty \exp \bigl(\gamma(-\tau^3 + \tau)\bigr)
  \,d\tau \,,
\end{displaymath}
where $\gamma= (\beta^3/\alpha)^{1/2}$.  We now observe that for any
$\tau\geq0$
\begin{displaymath}
  -\tau^3 + \tau =
  \frac{2}{3\sqrt{3}}-\sqrt{3}\Big(\tau-\frac{1}{\sqrt{3}}\Big)^2
    -\Big(\tau-\frac{1}{\sqrt{3}}\Big)^3 \leq \frac{1}{\sqrt{3}}-\sqrt{3}\Big(\tau-\frac{1}{\sqrt{3}}\Big)^2\,.
\end{displaymath}
Consequently,
\begin{displaymath}
  I\leq \sqrt{\frac{\beta}{\alpha}}\exp \biggl( \frac{\gamma}{\sqrt{3}}
  \biggr) \int_0^\infty \exp \biggl( -\gamma\sqrt{3}\Big(\tau-\frac{1}{\sqrt{3}}\Big)^2\biggr)
  \,d\tau \leq 3^{-1/4}\sqrt{\frac{\beta\pi}{\alpha\gamma}}\exp \biggl( \frac{\gamma}{\sqrt{3}}
  \biggr) \,,
\end{displaymath}
from which we easily conclude (\ref{eq:64a}). 
\end{proof}

\section{The dependence on current of 1D eigenvalues}
\label{appC}
\subsection{Robin boundary condition}
\label{appcompRobin}

As in Section \ref{s2} we consider here for $\jf \neq 0$ and $\kappa \geq 0$
the operator 
\begin{equation}
\label{eq:100}
  \LL^R(\mathfrak{j},\kappa)= -\frac{d^2}{dx^2} + i\, \mathfrak{j} \, x 
\end{equation}
defined on 
\begin{displaymath}
  D( \LL^R(\mathfrak j,\kappa))= \{ u\in H^2(\R_+,\C) \cap L^2(\R_+,\C\,;\,x^2dx)\,| \,
  u^\prime(0)= \kappa \, u(0)\,\} \,.
\end{displaymath}
The eigenfunctions of this operator are given for $n\geq 1$ by 
\begin{equation}
  \label{eq:83}
u_n^R(x,\jf,\kappa) =
\Ai\bigl(e^{i2\pi/3}(- i \jf^{1/3}x + \lambda_n^R (\jf,\kappa) \,  \jf^{-2/3})\bigr)\,.
\end{equation}
Hence, the Robin boundary condition reads
\begin{equation}
\label{eq:lambda_eqn}
0 = \left(\frac{du_n}{dx} - \kappa \, u_n \right)_{x=0} = - i \jf^{1/3} e^{i2\pi/3} \Ai^\prime\big(e^{i2\pi/3} \lambda_n^R(\jf,\kappa) \jf^{-2/3}\big) 
- \kappa \Ai\big(e^{i2\pi/3} \lambda_n^R (\jf,\kappa) \jf^{-2/3}\big)\,.
\end{equation}
Setting $\lambda^R_n (y) = \lambda^R_n(1,\kappa)\,$, with $y = \kappa \jf^{-\frac 13}$
(note that $y=\kappa^*$ in Section 2),
one obtains the  following equation
\begin{equation}
\label{eq:mu_eqn}
i e^{i2\pi/3} \Ai^\prime(e^{i2\pi/3} \mu) + y  \Ai(e^{i2\pi/3} \mu) = 0\,.
\end{equation}
One can numerically get  $\lambda_n^R(y)$ and recover  the eigenvalue of
$\LL^R(\jf,\kappa)$ via the relation
\begin{equation}
\label{eq:lambdan}
\lambda_n^R(\jf,\kappa) =  \jf^{2/3} \, \lambda_n^R(\kappa \, \jf^{-1/3}) \,.
\end{equation}
Taking the derivative of \eqref{eq:lambdan} with respect to $\jf$ yields
\begin{equation}
\label{eq:lambda_diff}
\frac{\partial \lambda_n^R}{\partial \jf} (\jf,\kappa) = \frac{\lambda_n^R(\kappa \, \jf^{-1/3})}{3 \jf^{1/3}} \biggl(2 - \kappa \, \jf^{-1/3} \, 
\frac{(\mu^R_n)^\prime(\kappa \, \jf^{-1/3})}{\lambda_n^R(\kappa \, \jf^{-1/3})} \biggr)\,.
\end{equation}
We focus our interest on the dependence of $\Re \lambda_n^R $ on
$\jf$.  From \eqref{eq:lambda_diff} we easily obtain
\begin{equation}
\label{eq:lambda_diffre}
\frac{\partial \Re \lambda_n^R}{\partial \jf} (\jf,\kappa) = \frac{\Re \lambda_n^R (\kappa \, \jf^{-1/3})}{3 \jf^{1/3}} \biggl(2 - \kappa \, \jf^{-1/3} \, 
\frac{\Re (\mu^R_n)^\prime (\kappa \, \jf^{-1/3})}{\Re \lambda_n^R(\kappa \, \jf^{-1/3})} \biggr)\,.
\end{equation}
We now state
\begin{proposition} 
\label{prop:robin-monotone}
For fixed $\kappa \geq 0$ and any $n\geq 1$, $\jf \mapsto
\Re\{\lambda_n^R(\jf,\kappa)\}$ is monotonically increasing on
$[0,+\infty)$.
\end{proposition}
Before presenting the proof, we provide, in Fig.~\ref{fig:lambda1},
the result of numerical simulations manifesting the monotonic
dependence of $\Re\{\lambda_1^R(\jf,\kappa)\}$ on $\jf$.   In
this figure, we plot a numerical solution of Eq. (\ref{eq:mu_eqn}) for
three values of $\kappa$.  These curves (shown by lines) are compared
to the graph of $\inf \Re\sigma\big(\mathcal L^{R,D}
(\jf,\kappa)\big)$ (shown by symbols), where $\LL^{R,D}$, which is
given again by \eqref{eq:100}, is defined on
\begin{displaymath}
  D( \LL^{R,D}(\mathfrak j,\kappa))= \{ u\in H^2([0,10],\C) \cap L^2([0,10],\C)\,| \,
  u^\prime(0)= \kappa \, u(0) \;,\; u(10)=0\,\} \,.
\end{displaymath}  
The solution in the latter case is obtained via a Galerkin expansion
in the Laplacian basis of $D( \LL^{R,D})$ (see
\cite{Grebenkov07,Grebenkov10,Gr1} for details). One can observe a
perfect agreement between these numerical solutions for $\jf$ values
that are not too small. This perfect agreement is a consequence of the
localization of the associated eigenfunction near $0$ at large $\jf$
so that the boundary condition at the other endpoint $x = 10$ has no
effect on the eigenvalue.


\begin{figure}
\begin{center}
\includegraphics[width=140mm]{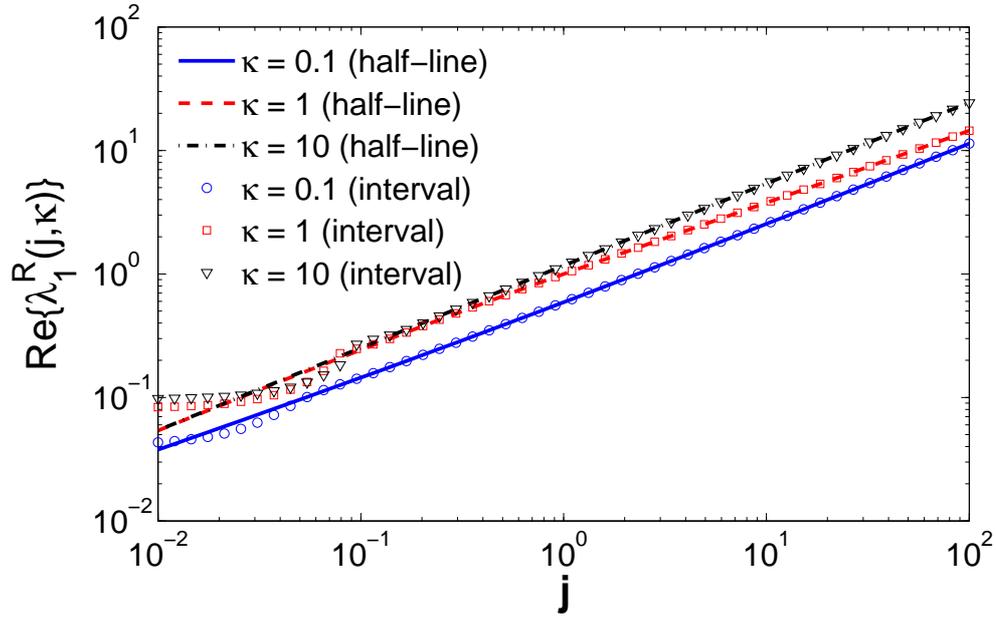}
\end{center}
\caption{
The graph of the real part of $\lambda_1^R(\jf,\kappa)$ as a function
of $ \jf $ for $\kappa = 0.1$ (blue solid line and circles), $\kappa =
1$ (red dashed line and squares), and $\kappa = 10$ (black dash-dotted
line and triangles).  Lines show the numerical solution of
(\ref{eq:lambda_eqn}), while symbols show $\inf \Re\sigma\big(\mathcal L^{R,D}
(\jf,\kappa)\big)$. }
\label{fig:lambda1}
\end{figure}

We also present a numerical solution of (\ref{eq:mu_eqn}) in
Fig.~\ref{fig:mu1} where we plot $\lambda_1^R$ as a function of
$y=\kappa\jf^{-1/3}$.  Note that, for any $y$, an eigenvalue of
$\mathcal L^R(1,y)$ is the unique continuous solution of
\begin{equation}
\label{eq:mu_eqna}
i e^{i\,2\pi/3} \Ai^\prime(e^{i\,2\pi/3} \lambda(y)) + y \,  \Ai(e^{i\,2\pi/3} \lambda(y)) = 0
\,,\, \lambda (0)= e^{-i\,2\pi/3} a^\prime_n\,,
\end{equation}
where $a_n^\prime$ denotes the $n$-th zero of $\Ai^\prime$ starting
from the right.  We note this extension by $\widetilde \lambda_n^R
(y)$ and we observe that $\widetilde \lambda_n^R (y) = \lambda_n^R(y)$
for $y\geq 0$ small enough by continuity. We will show at the end that
this is true for any $y$.  This notion is well defined since all the
solutions of \eqref{eq:mu_eqna} are simple, as established in
\cite{GHH}.

Considering in more detail the behavior of $\lambda_1^R(y)$, one 
expects that $\Re\{\lambda_1^R(y)\}$ monotonically grows from $\Re\{
e^{-i2\pi/3} a^\prime_1\} \approx 0.5094$ (at $y = 0$, corresponding
to a Neumann condition \cite{abst72}) to $\Re\{ e^{- i2\pi/3} a_1\}
\approx 1.1691$ (as $y\to + \infty$, corresponding to a Dirichlet
condition \cite{abst72}), where $a_1$ and $a^\prime_1$ are the
rightmost zeros of Airy function and its derivative, i.e., $\Re
(\lambda^R_1)^{\prime}(y) > 0\,$.  This expected behavior is clearly
manifested in Fig.~\ref{fig:mu1}.

\begin{figure}
\begin{center}
\includegraphics[width=140mm]{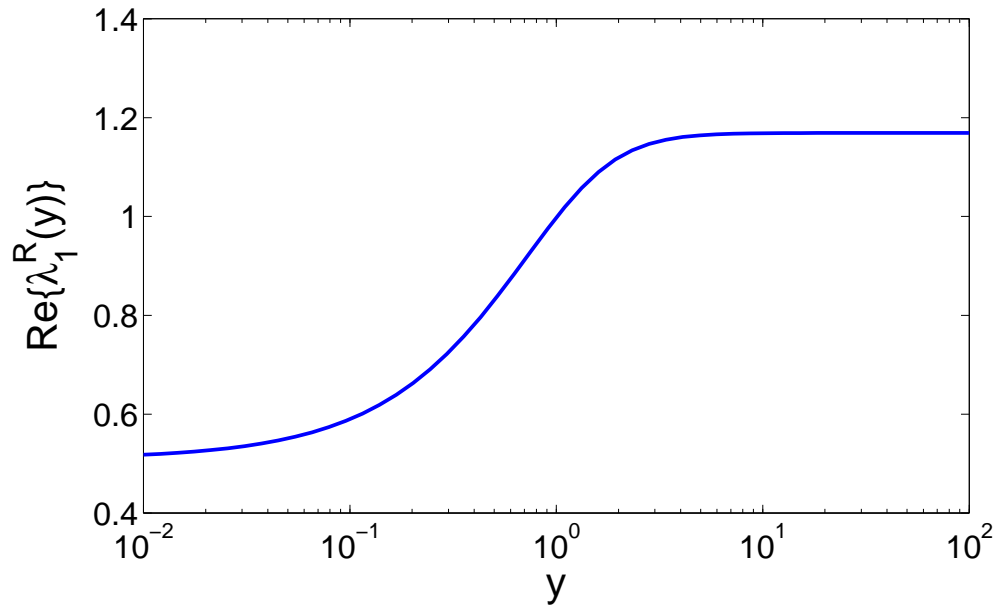}
\end{center}
\caption{
Robin case: the graph of the real part of $\lambda_1^R(y)$ obtained by
numerical solution of (\ref{eq:mu_eqn}). }
\label{fig:mu1}
\end{figure}

\begin{proof}[Proof of Proposition \ref{prop:robin-monotone}] \strut \\
Let
\begin{equation}
\label{defdelta} 
\delta^R_n(y) = 2 - y  \, \frac{\Re (\widetilde \lambda^R_n)^\prime(y)}{\Re \widetilde  \lambda_n^R (y)}\,,
\end{equation}
By \eqref{eq:lambda_diffre}, it is sufficient to prove that
\begin{displaymath}
 \delta^R_n (y) >0\,,\,  \forall y \in [0,+\infty)\,,
\end{displaymath}
in order to establish Proposition \ref{prop:robin-monotone}.  To this
end, we first establish a differential equation for $\lambda_n^R(y)$.
We set $\hat \lambda_n =e^{i\,2\pi/3} \widetilde \lambda^R_n$.  We may
represent \eqref{eq:mu_eqn} in the form
\begin{equation}
\label{eq:mu_eqnab}
i \,e^{i\,2\pi/3} \Ai^\prime(\hat \lambda_n(y)) + y\,  \Ai(\hat \lambda_n (y)) = 0 \,.
\end{equation}
Differentiating with respect to $y$ yields, with the aid of Airy's
equation,
\begin{equation*}
- i\, e^{i\,2\pi/3} \hat \lambda_n (y)  \hat \lambda^\prime_n (y) - i \, e^{-i\,2\pi/3} \, y ^2\hat \lambda^\prime_n(y)  = 1\,,
\end{equation*}
leading to the initial value problem, formulated for $\widetilde
\lambda_n^R$,
\begin{equation}
\label{mainedo}
\begin{cases}
  (\widetilde \lambda_n^R)^\prime(y) (\widetilde \lambda_n^R(y) + y^2) = -i \,, & \\
\widetilde \lambda_n^R(0) =  e^{-i\,2\pi/3} a^\prime_n\,. &
\end{cases}
\end{equation}
We may now conclude that,
\begin{equation}
  (\widetilde \lambda_n^R)^\prime(0)= -i\, /\,\widetilde \lambda_n^R(0)\,.
\end{equation}
Writing $u_n(y) = \Re\{\widetilde \lambda_n^R(y)\}$ and $v_n(y) =
\Im\{\widetilde \lambda_n^R(y)\}$, one gets
\begin{equation}
\begin{split}
u^\prime_n u_n - v^\prime_n v_n + u^\prime_n y^2 & = 0\,, \\
u^\prime_nv_n + v^\prime_n u_n + v^\prime_n y^2 & = 1\,, \\
\end{split}
\end{equation}
from which we easily obtain that
\begin{equation}\label{niceformula1}
u^\prime_n(y) =  \frac{v_n(y) }{(u_n+y^2)^2 + v_n(y)^2}
\end{equation}
and
\begin{equation}\label{niceformula2}
v^\prime_n(y) =  \frac{u_n(y)+y^2}{(u_n(y)+y^2)^2 + v_n(y)^2} > 0 \,.
\end{equation}
As both $u_n$ and $v_n$ must be positive for all $y\in\R_+$ and
$n\in\N^*$ (eigenvalues must belong to the numerical range), it easily
follows that both $y \mapsto u_n(y)$ and $y \mapsto v_n(y)$ are
monotone increasing for all $n\in\N$. In particular we have
\begin{equation*}
\frac12 |a_n^\prime| = u_n (0)\leq u_n(y)\leq \lim_{y\to + \infty}u_n(y)= \frac12 |a_{\ell (n)}|\,,
\end{equation*}
for some $\ell (n)\in \mathbb N^*$.\\
But the same argument shows that each continuous (with respect to $y$)
eigenvalue of $\mathcal L^R(1,y)$ is monotonous. This leads to a
contradiction if $\ell (n) >n$. By recursion on $n$, we get that $\ell
(n)=n$.  Hence we have
\begin{equation*}
  - a^\prime_n< -a_n < -a^\prime_{n+1} \quad  \forall n\in\N^* \,,
\end{equation*}
and we obtain that
\begin{displaymath}
  u_n(y)< - a_n < u_{n+1}(y) \,,\quad  \forall(y,n)\in\R_+\times\N^*\,,
\end{displaymath}
and that $\widetilde \lambda_n^R(y)=\lambda_n^R (y)$ for all $y$ as
expected.

We can now control the sign of $\delta^R_n (y)$, using the lower bound
for $u_n$,
\begin{equation}
\label{eq:127}
\frac{y u^\prime_n }{u_n} \leq \frac{yv_n} {u_n[2u_n y^2 + v^2_n]}\leq \frac{1}{2 ^\frac 32 u_n^{3/2}}
\leq \frac{1}{2 \sqrt{2}|a_1^\prime|^{3/2}} < 2\,.
\end{equation}
(cf. \cite{abst72} for the justification of the last inequality.)  The
proposition now follows from \eqref{defdelta} and
\eqref{eq:lambda_diffre}.
\end{proof}
\begin{remark}
We may also obtain the asymptotic behavior of $\delta^R_n$ as $y\to
+\infty$. As
\begin{equation}
\frac{y  u^\prime_n (y)}{u_n (y)} \leq \frac{1}{2 |a_n^\prime | \,y}\, \frac{2 v_1(y) \, y^2} {
  y^4 + v_n(y)^2} \leq \frac{1}{2 |a_n^\prime| \, y}\,,
\end{equation}
we  obtain that
\begin{displaymath} 
   \lim_{y\to + \infty} \delta^R_n(y) =2\,.
\end{displaymath}
\end{remark}
In Fig. \ref{fig:mu1_derivative} we plot $yu_1^\prime(y)/u_1(y)$ as a
function of $y$.

\begin{figure}
\begin{center}
\includegraphics[width=140mm]{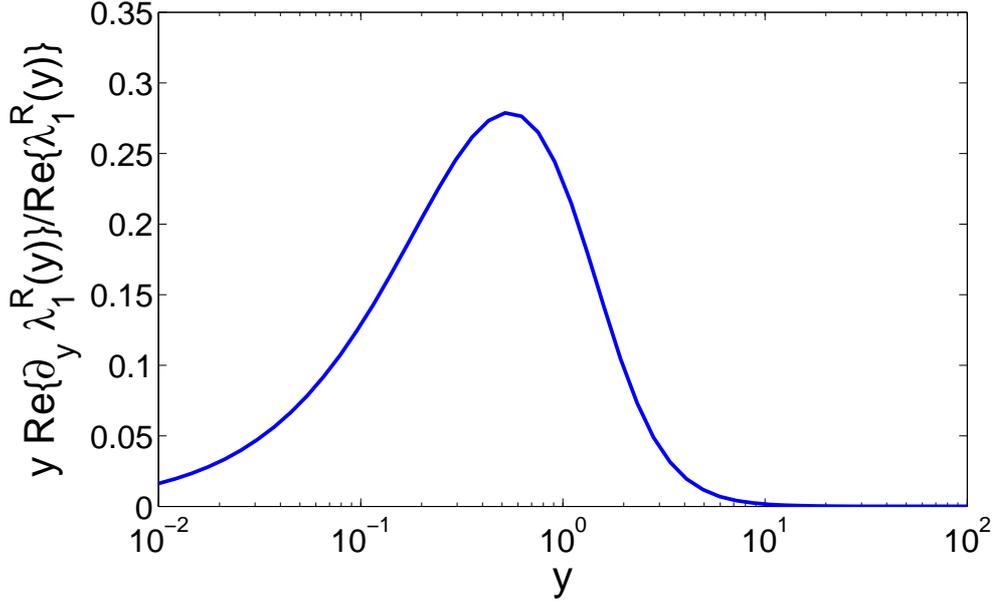}
\end{center}
\caption{
Robin case: the graph of $y\mapsto y \,
\Re\{(\lambda^{R}_1)^{\prime}(y)\}/\Re\{\lambda^R_1(y)\}$ obtained by a
numerical solution of (\ref{eq:mu_eqn}).}
\label{fig:mu1_derivative}
\end{figure}

%
%
Alternatively, one can solve the problem on the interval $[0,L]$ with
Robin and Dirichlet boundary conditions at $0$ and $L$ (with $L \gg
1$), respectively, using a Galerkin expansion, and then diagonalize
the underlying matrix.

\subsection{Transmission case}

We consider the eigenvalues of the operator $ \mathcal
L^{T}(\jf,\kappa) = -\frac{d^2}{dx^2} + i\jf x$ on the line with
transmission condition in \eqref{eq:98} for a given non-negative value
of $\kappa$.  Introducing $\lambda^T = \lambda^T(\jf,\kappa)\,
\jf^{-2/3}$, one gets the following equation
\begin{equation}
\label{eq:mu_eqnT}
2\pi \Ai^\prime(e^{i2\pi/3} \lambda) \, \Ai^\prime(e^{-i2\pi/3} \lambda) = - \kappa \, \jf^{-1/3} .
\end{equation}
Setting, as in the Robin case, $y=\kappa \, \jf^{-1/3}$ leads to
\begin{equation}\label{eqmunT}
2\pi  \Ai^\prime(e^{i2\pi/3} \lambda) \, \Ai^\prime(e^{-i2\pi/3} \lambda) = - y\,,
\end{equation}
for $y\geq0$.\\
This equation can be solved numerically to find $\lambda_n(\kappa
\jf^{-1/3})$, from which we obtain the eigenvalues via the relation
\begin{equation}
\label{eq:lambdanT}
\lambda_n^T(\jf,\kappa) =  \jf^{2/3} \, \lambda_n^T(\kappa \, \jf^{-1/3}) \,.
\end{equation}
Taking the derivative with respect to $\jf$ yields
\begin{equation}
\label{eq:lambda_diffT}
\frac{\partial}{\partial \jf } \Re\lambda_n^T(\jf,\kappa) = \frac{\Re\lambda_n^T(\kappa \, \jf^{-1/3})}{3 \jf^{1/3}} \biggl(2 - \kappa \, \jf^{-1/3} \, 
\frac{(\Re \lambda^T_n)^\prime(\kappa \, \jf^{-1/3})}{\Re\lambda_n^T(\kappa \, \jf^{-1/3})} \biggr)\,.
\end{equation}
We focus attention on the variation of $\Re \lambda_1^T (\jf,\kappa)$
with respect to $\jf$. To this end we attempt to determine the sign of
\begin{equation}\label{defdeltaT}
\delta^T (y) = 2 - y  \, \frac{\Re (\lambda^T_1)^\prime (y)}{\Re \lambda^T_1 (y)}\,,
\end{equation}
where $\lambda_1^T (y)$ is the unique continuous solution of
\eqref{eqmunT} satisfying $\lambda_1^T (0)=|a_1^\prime|e^{i\pi/3}$.

Obviously, \eqref{eqmunT} poses a significantly greater obstacle than
\eqref{eq:mu_eqn}. The following simple observation can still be made:
\begin{lemma}
There exists $y_0 >0$ such that $\delta^T (y) >0$ on $[0,y_0]$.
\end{lemma}
\begin{proof}
By continuity, it is enough to prove the statement at $y=0$. By
\eqref{defdeltaT} and the fact that $\Re\lambda_1^T
(0)=|a_1^\prime|/2>0$ we readily obtain $\delta^T (0)=2$.
\end{proof}
The following conjecture can
be made in the large $y$ limit
\begin{conjecture}\label{conj6}
There exists $y_1 >0$ such that $\delta^T (y) >0$ on $[y_1,+\infty)$.
\end{conjecture}
Note that as $y \to + \infty$ the transmission problem ``tends'' to $-
\frac{d^2}{dx^2} + i x$ on the line, which has no spectrum.  Following
\cite{Gr1}, we provide in the sequel  a formal justification
 for the conjecture together with an enlightening picture.

Figure \ref{fig:mu1_derivativeT} shows that the graph of $y \mapsto y
\, \frac{ \Re (\lambda^T_1)^\prime(y)}{\Re \lambda_1^T(y)}$ attends
its maximum at a value which is below $0.30$.  As a consequence, by
(\ref{eq:lambda_diff}), $\Re \lambda_1^T$ is monotone increasing in
$\jf$.

\begin{figure}
\begin{center}
\includegraphics[width=140mm]{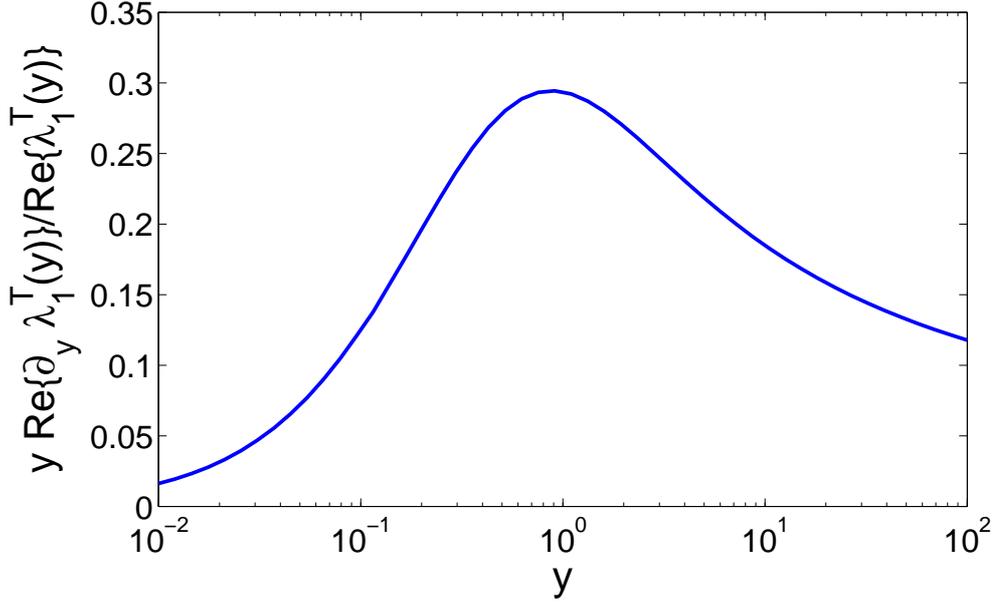}
\end{center}
\caption{
Transmission case: graph of $y\mapsto y \,
\Re\{\partial_y \lambda_1^T(y)\}/\Re\{\lambda_1^T(y)\}$ obtained by
numerical solution of (\ref{eq:mu_eqnT}).}
\label{fig:mu1_derivativeT}
\end{figure}

\subsubsection{Reminder on Airy functions}

The Wronskian for Airy functions is
\begin{equation}
\label{eq:Wronskian}
e^{-i2\pi/3} \Ai^\prime(e^{-i2\pi/3} z) \Ai(e^{i2\pi/3} z) - e^{i2\pi/3}
\Ai^\prime(e^{i2\pi/3} z) \Ai(e^{-i2\pi/3} z) = \frac{i}{2\pi} \,  \quad \forall~
z\in \mathbb C\, . 
\end{equation}
Note that these two functions are related to $\Ai(z)$ by the identity
\begin{equation}\label{ide}
\Ai (z) + e^{-i 2\pi/3} \Ai (e^{-i 2\pi/3} z) + e^{i2\pi/3} \Ai (e^{i2\pi/3}z) =0   \quad \forall~ z\in \mathbb C\, .
\end{equation}
By differentiation we also get
\begin{equation}\label{idenew}
\Ai^\prime (z) + e^{i 2\pi/3} \Ai^\prime (e^{-i 2\pi/3} z) + e^{-i2\pi/3} \Ai^\prime (e^{i2\pi/3}z) =0  \quad \forall~ z\in \mathbb C\, .
\end{equation}

The Airy function and its derivative satisfy different asymptotic
expansions depending on their argument: \\
For $|\arg z|<\pi$,
\begin{eqnarray} \label{5} 
\Ai(z) &=& \frac 12 \pi^{-\frac 12}z^{-1/4} \,  \exp\left(-\frac{2}{3}z^{3/2}\right) \bigl(1 + \mathcal O  (|z|^{-\frac 32}) \bigr), \\
\label{5a}
\Ai^\prime(z) &=& - \frac 12 \pi^{-\frac 12} \, z^{1/4}\, \exp\left(-\frac{2}{3}z^{3/2}\right) \bigl(1+ \mathcal O (|z|^{-\frac 32}) \bigr) \,, 
\end{eqnarray}
where moreover $\mathcal O$ is, for any $\epsilon >0$, uniform when $|\arg
z| \leq \pi -\epsilon$\,.

\subsubsection{The behavior as $y\ar +\infty$}

The behavior of $\lambda_1^T(y)$ as $y\ar +\infty$ (we choose the
eigenvalue with positive imaginary part) was analyzed in \cite{Gr1}
(in particular, see Fig. 1). Without full mathematical rigor, using in
particular the asymptotics of the previous sub-subsection, it is
established (see formula (40)) that
\begin{equation*}
\Re \lambda_1^T(y) \sim \biggl(\frac 34 \log y \biggr)^\frac 23 \mbox{ as } y\ar +\infty\,.
\end{equation*}
and that
\begin{equation*}
\lim_{y\ar +\infty} \Im \lambda_1^T(y) =0\,.
\end{equation*}
With more effort, one can establish at least formally that
\begin{equation*}
\Re ( \lambda_1^T)^\prime(y) \sim \frac 12 \,\frac 1 y\, \biggl(\frac 34 \log y\biggr)^{-\frac 13}\,,
\end{equation*}
which confirms the monotonicity of $\Re \lambda_1^T$ for large $y$ and
implies
\begin{equation*}
\lim_{y\ar + \infty} \delta^T(y) =2\,,
\end{equation*}
which confirms Conjecture \ref{conj6}. \\
One can also  formally obtain
\begin{equation*}
\Im \lambda_1^T (y) \sim   \frac \pi 2\,  \biggl(\frac 34 \log y \biggr)^{-\frac 13}\,.
\end{equation*}

\end{document}